\documentclass[12pt]{article}

\usepackage[left=1in,top=1in,right=1in,bottom=1in,head=.1in,nofoot]{geometry}
\usepackage{setspace,url,bm,amsmath} 
\usepackage{titlesec} 
\titlelabel{\thetitle.\quad} 
\usepackage[margin=20pt]{subcaption}
\usepackage{amsmath,amsfonts,amsthm,amssymb}
\usepackage{graphicx}
\usepackage[colorlinks=true,linkcolor=black,citecolor=blue,urlcolor=blue]{hyperref}
\usepackage[table]{xcolor}
\usepackage{multirow}
\usepackage{array}
\usepackage{booktabs}
\usepackage{threeparttable}
\usepackage{comment}
\usepackage{float}
\usepackage{natbib}
\usepackage{indentfirst}

\usepackage{longtable}
\usepackage{amsfonts}

\newtheorem{assumption}{Assumption}
\newtheorem*{theorem*}{Theorem}
\newtheorem{theorem}{Theorem}

\newtheorem{lemma}{Lemma}
\newtheorem{remark}{Remark}

\newtheorem*{corollary*}{Corollary}

\newcommand{\RNum}[1]{\uppercase\expandafter{\romannumeral #1\relax}}
\newcommand{\Perp}{\perp\!\!\!\perp}

\newcommand \Ya {Y_i(a)}
\newcommand \Yt {Y_i(1)}
\newcommand \Yc {Y_i(0)}

\newcommand \ind[1] {I(#1)}
\newcommand \diag[1] {\operatorname{diag}\{#1\} }
\newcommand \E[1]{E\{#1\}}
\newcommand \Econd[2]{E\{#1\mid#2\}}
\newcommand \Var[1]{\operatorname{var}\{#1\}}
\newcommand \Varcond[2]{\operatorname{var}\{#1\mid#2\}}
\newcommand \Prob[1]{\operatorname{pr}(#1)}
\newcommand \ProbB[1]{\operatorname{pr}\{#1\}}

\newcommand \inp {\overset{P}{\rightarrow}}
\newcommand \indist {\overset{d}{\rightarrow}}
\newcommand \equaldist {\overset{d}{=}}

\newcommand \ns[1]{n(#1)}
\newcommand \ps[1]{p(#1)}

\newcommand \Dns {D_{n}(s)}

\newcommand \nx[1]{n_{#1}}
\newcommand \px[1]{p_{#1}}
\newcommand \hatpx[1]{\hat p_{#1}}
\newcommand \nax[2]{n_{#1#2}}
\newcommand \nas[2]{n_{#1}(#2)}
\newcommand \tildenas[2] {\tilde{n}_{#1}(#2)}
\newcommand \naxs[3]{n_{#1#2}(#3)}
\newcommand \nxs[2]{n_{. #1}(#2)}
\newcommand \pxs[2]{p_{. #1}(#2)}

\newcommand \muax[2]{\mu_{#1#2}}
\newcommand \Yax[2]{\bar{Y}_{#1#2}}
\newcommand \sigmaSqax[2]{\sigma^2_{#1#2}}
\newcommand \hatsigmaSqax[2]{\hat \sigma^2_{#1#2}}

\newcommand \hatmuas[2]{\hat{\mu}_{n,#1}(#2)}
\newcommand \hatmuaxs[3]{\hat{\mu}_{n,#1#2}(#3)}

\newcommand \zetaStrY[1]{\varsigma_{\tilde{Y}|X=#1}^2(\pi)}
\newcommand \zetaStrA[1]{\varsigma^2_{A, Y|X=#1}(\pi)}
\newcommand \zetaStrH[1] {\varsigma^2_{H, Y|X=#1}}

\newcommand \zetaAddY[1] {\varsigma_{\tilde{r}_#1}^2(\pi)}
\newcommand \zetaAddA[1] {\varsigma_{A,{r_#1}}^2(\pi)}
\newcommand \zetaAddH[1] {\varsigma_{H,{r_#1}}^2}
\newcommand \zetaAddYxy[2] {\varsigma_{\tilde{r}_{#1#2}}(\pi)}
\newcommand \zetaAddAxy[2] {\varsigma_{A,{r_{#1#2}}}(\pi)}
\newcommand \zetaAddHxy[2] {\varsigma_{H,{r_{#1#2}}}}

\newcommand \zetaCheckY[1]{\varsigma_{\check{Y}|X_i = #1}^2(\pi)}
\newcommand \zetaCheckH[1] {\varsigma^2_{H, {Y}^#1|X_i = #1}}

\newcommand \hatzetaStrY[1]{\hat\varsigma_{\tilde{Y}|X=#1}^2(\pi)}
\newcommand \hatzetaStrA[1]{\hat\varsigma^2_{A, Y|X=#1}(\pi)}
\newcommand \hatzetaStrH[1] {\hat\varsigma^2_{H, Y|X=#1}}

\newcommand \hatzetaAddY[1] {\hat\varsigma_{\tilde{r}_#1}^2(\pi)}
\newcommand \hatzetaAddA[1] {\hat\varsigma_{A,{r_#1}}^2(\pi)}
\newcommand \hatzetaAddH[1] {\hat\varsigma_{H,{r_#1}}^2}
\newcommand \hatzetaAddYxy[2] {\hat\varsigma_{\tilde{r}_{#1#2}}(\pi)}
\newcommand \hatzetaAddAxy[2] {\hat\varsigma_{A,{r_{#1#2}}}(\pi)}
\newcommand \hatzetaAddHxy[2] {\hat\varsigma_{H,{r_{#1#2}}}}

\newcommand \hatzetaCheckY[1]{\hat\varsigma_{\check{Y}|X_i = #1}^2(\pi)}
\newcommand \hatzetaCheckH[1] {\hat\varsigma^2_{H, {Y}^#1|X_i = #1}}

\newcommand \tildeYa[1]{\tilde{Y}_i(#1)}
\newcommand \checkYa[1]{\check{Y}_i(#1)}
\newcommand \maxZ[2]{m_{#1#2}(Z_i)}

\newcommand \rax[2]{r_{#2,i}(#1)}
\newcommand \tilderax[2]{\tilde r_{#2,i}(#1)}
\newcommand \tilderaxs[2]{\tilde r^s_{#2,i}(#1)}
\newcommand \raxZ[2]{r_{#1#2}(Z_i)}

\newcommand \Mx[1]{M_{n, #1}}
\newcommand \Rx[1]{R_{n, #1}}
\newcommand \Zx[1]{Z_{n, #1}}
\newcommand \Dax[2]{D_{n,#1#2}}

\newcommand \Ex[1]{E_{n,#1}}
\newcommand \Ux[1]{U_{n,#1}}
\newcommand \Sx[1]{S_{n,#1}}
\newcommand \Lx[1]{L_{n,#1}}

\newcommand \etax[1]{\eta_{#1}}
\newcommand \upsilonx[1]{\upsilon_{#1}}
\newcommand \xix[1]{\xi_{#1}}

\newcommand \tildeYas[1]{\tilde{Y}^s_i(#1)}

\newcommand \Erx[1]{E_{n,r_#1}}
\newcommand \Urx[1]{U_{n,r_#1}}
\newcommand \Srx[1]{S_{n,r_#1}}

\newcommand \etarx[1]{\eta_{r_#1}}
\newcommand \upsilonrx[1]{\upsilon_{r_#1}}
\newcommand \xirx[1]{\xi_{r_#1}}

\newcommand \Scheckx[1]{S_{n, Y^#1}}
\newcommand \Lcheckx[1]{L_{n, Y^#1}}

\newcommand \etacheckx[1]{\eta_{Y^#1}}
\newcommand \xicheckx[1]{\xi_{Y^#1}}

\newcommand \checkYaxs[2]{\check Y^{#2,s}_i(#1)}

\newcommand \zetatold{\varsigma_{t_{old}}^2}
\newcommand \zetat{\varsigma_{t}^2}
\newcommand \zetatprime{\varsigma_{t^\prime}^2}

\begin{document}
\begin{singlespace}
\title{\bf Interaction tests with covariate-adaptive randomization}

\author{
\small
{
Likun Zhang, Wei Ma\thanks{\small{Correspondence: \texttt{mawei@ruc.edu.cn}}}
}
\\ \\
{\small Institute of Statistics and Big Data, Renmin University of China, Beijing, China}
}

\date{}
\maketitle
\end{singlespace}

\thispagestyle{empty}
\vskip -8mm 

\begin{singlespace}
\begin{abstract}
Treatment-covariate interaction tests are commonly applied by researchers to examine whether the treatment effect varies across patient subgroups defined by baseline characteristics. The objective of this study is to explore treatment-covariate interaction tests involving covariate-adaptive randomization. Without assuming a parametric data generating model, we investigate usual interaction tests and observe that they tend to be conservative: specifically, their limiting rejection probabilities under the null hypothesis do not exceed the nominal level and are typically strictly lower than it. To address this problem, we propose modifications to the usual tests to obtain corresponding valid tests. Moreover, we introduce a novel class of stratified-adjusted interaction tests that are simple, more powerful than the usual and modified tests, and broadly applicable to most covariate-adaptive randomization methods. The results are general to encompass two types of interaction tests: one involving stratification covariates and the other involving additional covariates that are not used for randomization. Our study clarifies the application of interaction tests in clinical trials and offers valuable tools for revealing treatment heterogeneity, crucial for advancing personalized medicine.

\vspace{12pt}
\end{abstract}

\end{singlespace}

\newpage

\clearpage
\setcounter{page}{1}

\allowdisplaybreaks
\baselineskip=24pt

\begin{singlespace}

\section{Introduction}
	Randomized controlled trials are considered the gold standard for evaluating the effectiveness of treatments. The overall average treatment effect is typically reported to measure the impact of the treatment across the entire population. However, the population often exhibits heterogeneity, leading to individuals experiencing variable responses to the same treatment, which can range from differing levels of benefit to potential harm \citep{gabler2009, vanderweele2014}. According to \cite{schork2015}, top ten highest-grossing drugs in the United States are effective for less than 1 in 4 of the individuals taking them. It is widely encouraged that the analysis of heterogeneity of treatment effects, undertaken to examine the nonrandom variability of treatment effects among individuals within a population, should follow the evaluation of the trial's primary outcome \citep{varadhan2013}. The guideline from European Medicines Agency recommends that ``grouping together patients with similar characteristics to explore variability of response to treatment between different groups of patients constitutes an integral part of the risk/benefit assessment''\citep{EMA2019}. Analyzing treatment effect heterogeneity is also crucial for cost-effective marketing strategies and personalized medicine \citep{gong2021}. The appropriate approach to assessing treatment effect heterogeneity involves using interaction tests to examine contrasting effects across different baseline covariate levels \citep{kent2018, christensen2021}.

 
Extensive research has been conducted on interaction tests in randomized controlled trials, with notable studies by \cite{altman2003}, \cite{royston2004}, \cite{tian2014}, and \cite{imbens2015}. Most of these studies focused primarily on the simple randomization scheme, which independently allocates patients to treatment and control groups at fixed probabilities. 
However, it is common practice to first stratify patients by baseline covariates, known or suspected to influence the outcome, before allocating them through stratified randomization.
Although interaction tests for these stratification covariates are a common practice \citep{wallach2017}, the dependency induced by stratified randomization complicates the identification of an optimal method for conducting valid interaction tests.
Moreover, concerns regarding the credibility of interaction testing have been raised when the covariate of interest is not used for stratification \citep{cui2002, grouin2005, sun2011}.
Therefore, the objective of this study is to investigate interaction tests under stratified randomization, and more broadly, covariate-adaptive randomization.

Covariate-adaptive randomization refers to randomization schemes that aim to balance treatments among baseline covariates. A commonly used scheme is stratified randomization. Stratified randomization groups units into strata based on one or more covariates. Within each stratum, a restricted randomization method such as block randomization \citep{ZELEN1974} or biased-coin design \citep{efron1971} is applied.
\cite{Pocock1975} introduced a minimization procedure to achieve balance over covariates' margins. \cite{Hu2012} generalized this approach to enable the simultaneous control of various types of imbalances. Covariate-adaptive randomization is arguably the most widely used design in randomized controlled trials, as evidenced by a survey of 298 randomized trials published in leading medical journals, over 80\% of which used covariate-adaptive randomization \citep{ciolino2019}. A comprehensive review of covariate-adaptive randomization in clinical trials can be found in \cite{Rosenberger2015}. 

Early studies on statistical inference under covariate-adaptive randomization were primarily based on simulations \citep{birkett1985, forsythe1987, weir2003}. In the last decade, theoretical studies have emerged, laying a solid foundation for the methodologies of conducting statistical inference under covariate-adaptive randomization. \cite{shao2010} highlighted that, provided that the simple linear regression model accurately represents the data generating model, the usual two sample $t$-test is conservative under stratified biased-coin design, meaning its limiting rejection probability under the null hypothesis does not exceed and is typically strictly less than the nominal level. Considering different randomization schemes and model assumptions, \cite{shao2013}, \cite{ma2015}, \cite{ma2020}, and \cite{wang2021} proposed various valid tests in which the limiting rejection probabilities were equal to the nominal level. More recently, from a model-free perspective, \cite{Bugni2018} investigated the two sample $t$-test and the $t$-test with strata fixed effects, without assuming any specific functional form for the data generating model. Built upon this work, several regression-adjusted estimators have been proposed to further improve statistical efficiency \citep{Bugni2019,ma2022,ye2022inference, Liu2023}.

Notably, the above-mentioned studies focused primarily on inferring the overall average treatment effect. Although several researchers have included treatment-covariate interactions in their regression models \citep{Bugni2019,ma2022,ye2022, Liu2023}, their objective was not to infer these interactions but to evaluate the overall average treatment effect more efficiently.
To the best of our knowledge, the inferential properties of interaction tests under covariate-adaptive randomization remain largely unexplored.
Our work aims to fill this gap by addressing the crucial question of whether the usual interaction test is valid and, if not, how to construct valid and more powerful tests.
Additionally, it is conceptually important to distinguish between the two types of covariates under investigation: the stratification covariates that are used for the covariate-adaptive randomization scheme, and the additional baseline covariates that may also exhibit associations with the outcome and possess potential interaction effects. 
Therefore, examining the interaction tests for both covariates is crucial. 
Moreover, we seek to determine whether both types of covariates can be accommodated within the same inference framework, allowing a universal algorithm to be used for testing interaction effects.

Considering these aspects, this study aims to examining the inferential properties of interaction tests under covariate-adaptive randomization. Specifically, we begin by providing a model-free definition of interaction effects. Subsequently, for both stratification covariates and additional baseline covariates, we assess the validity of usual interaction tests and determine that they tend to be conservative under covariate-adaptive randomization. To remedy this issue, we modify these conservative tests to derive valid tests. Additionally, we introduce and explore a novel category of interaction tests that incorporate stratified adjustment and apply for both types of the covariates under investigation. Our findings suggest that stratified-adjusted interaction tests, when applied under covariate-adaptive randomization, provide numerous advantages, including simplicity, high statistical power, and broad applicability.

The remainder of this paper is organized as follows. In Section \ref{section:notation}, we describe our framework and notation.
The study of interaction tests for binary covariates is divided into two parts: Section \ref{section:binary_used} covers stratification covariates and Section \ref{section:binary_add} discusses additional baseline covariates. Section \ref{section:simulation} presents the results of a simulation study on our interaction tests. In Section \ref{section:multi}, we extend the main results from binary to categorical covariates. In Section \ref{section:discussion}, we discuss further applications of our methods. Additional simulations and a clinical trial example are given in the Supplementary Material.

	\section{Framework and notation} \label{section:notation}
  In a covariate-adaptive randomized experiment with two treatments, let $A_i$ denote the $i \text{th}$ treatment assignment, with $A_i = 1$ for the treatment and $A_i = 0$ for the control. We use the Neyman--Rubin model of causal inference \citep{Neyman1990, Rubin1974}. Let $\Yt$ and $\Yc$ denote the two scalar potential outcomes for unit $i$, when exposed to the treatment and to the control, respectively. The observed outcome $Y_i$ is determined as
$$
Y_i = \Yt A_i + \Yc(1-A_i).
$$
Let $A^{(n)}=\left\{A_1, \ldots, A_n\right\}$ and $W^{(n)}=\left\{W_1, \ldots, W_n\right\}$, where $W_i = (\Yt, \Yc, Z_i^{\mathrm{T}})^{\mathrm{T}}$, $i = 1, \dots, n$, are independent and identically distributed samples from the population distribution $W = (Y(1), Y(0), Z^{\mathrm{T}})^{\mathrm{T}}$. $Z$ is a random vector of baseline covariates.
Moreover, let $V^{(n)}=\left\{V_1, \ldots, V_n\right\}$, where $V_i = (Y_i, A_i, Z_i^{\mathrm{T}})^{\mathrm{T}}$, $i = 1, \dots, n$, are observed data. 

The units are stratified into strata based on baseline covariates $Z_i$ using the stratification function $S$: $\text{supp}(Z_i) \rightarrow \mathcal{S}$, where $\mathcal{S}$ is a finite set, and $\text{supp}(Z_i)$ is the support of $Z_i$. In general, units are stratified according to certain stratification covariates, such as gender, grade, or location.  Let  $S^{(n)}=\left\{S_1, \ldots, S_n\right\}$, where $S_i = S(Z_i)$ is the stratum label for unit $i$. For simplicity, we assume that units are assigned to each stratum with positive probability, that is, $p(s) = \Prob{S_i = s} > 0$ for all $s \in \mathcal{S}$. 

We measure the imbalance in stratum $s\in \mathcal{S}$ by 
$$
\Dns=\sum_{i=1}^n\left(A_{i}-\pi\right) \ind{S_i = s}, 
$$
where $\pi \in (0, 1)$ is the target treatment proportion, which is invariant across strata. In stratum $s$, let $\ns{s} = \sum_{i=1}^n \ind{S_i = s}$, $\nas{1}{s} = \sum_{i=1}^n A_i\ind{S_i = s}$, and $\nas{0}{s} = \sum_{i=1}^n (1-A_i)\ind{S_i = s}$ denote the number of units, number of treated units, and number of control units, respectively. 

We use $X_i$ to denote a covariate, the additive treatment-covariate interaction of which is of interest. We assume that $X_i$ is categorical. If $X_i$ is continuous, it can be categorized before the analysis. $\mathcal{X}$ denotes the finite set of all possible levels of $X_i$. We also assume that $\px{x} = \Prob{X_i = x} > 0$ for all $x \in \mathcal{X}$. The additive treatment-covariate interaction can be mathematically expressed as
\begin{align*}
	\Econd{\Yt - \Yc}{X_i = x_1} \neq \Econd{\Yt - \Yc}{X_i = x_2},
\end{align*}
for some $x_1$, $x_2$ $\in \mathcal{X}$. This expression indicates that the treatment effect varies across individuals with different levels of $X_i$. Notably, the definition of interaction often varies with the statistical models and depends on the correct model specification. In this study, we use a model-free definition based on the variation in the conditional treatment effect. In the remainder of this paper, we refer to ``additive treatment-covariate interaction'' as ``interaction'' when there is no potential confusion. Our objective is to test the presence of interaction for covariate $X_i$.

For $x \in \mathcal{X}$, let $\nx{x} = \sum_{i=1}^n \ind{X_i = x}$, $\nax{1}{x}= \sum_{i=1}^n A_i \ind{X_i = x}$, and $\nax{0}{x} = \sum_{i=1}^n (1-A_i)\ind{X_i = x}$. To simplify the notation, we define
\begin{equation*}
	\tau_x = \Econd{\Yt - \Yc}{X_i = x}.
\end{equation*}
Furthermore, we define $\muax{1}{x} = \Econd{\Yt}{X_i = x}$ and $\muax{0}{x} = \Econd{\Yc}{X_i = x}$. A common estimator of $\tau_x$ is the difference-in-means estimator $\Yax{1}{x} - \Yax{0}{x}$, where $\Yax{1}{x} = (1 / \nax{1}{x})\sum_{i=1}^n Y_iA_i \ind{X_i = x}$ and $\Yax{0}{x} = (1/\nax{0}{x}) \sum_{i=1}^n Y_i(1-A_i) \ind{X_i = x}$. Let $\sigmaSqax{1}{x} = \Varcond{\Yt}{X_i = x}$, $\sigmaSqax{0}{x} = \Varcond{\Yc}{X_i = x}$ denote the conditional variances of potential outcomes, given $X_i = x$.  Let $\hatsigmaSqax{1}{x} = (1 / \nax{1}{x})\sum_{i=1}^{n}(Y_i - \Yax{1}{x})^2A_i\ind{X_i = x}$ and  $\hatsigmaSqax{0}{x} = (1 / \nax{0}{x})\sum_{i=1}^n(Y_i - \Yax{0}{x})^2(1-A_i)\ind{X_i = x}$ be sample versions of $\sigmaSqax{1}{x}$ and $\sigmaSqax{0}{x}$, respectively. 

We introduce the following assumptions regarding the data generating process and covariate-adaptive randomization:

\begin{assumption}
	$\E{|\Ya|^2} < \infty$,  and 
$$\max_{s \in \mathcal{S}} \Varcond{\Ya}{S_i = s, X_i = x} > 0,$$ for $a \in \{1, 0\}$, $x \in \mathcal{X}$.
	\label{assumption:nondegenerate}
\end{assumption}

\begin{assumption}
        $W^{(n)} \Perp A^{(n)} \mid S^{(n)}$.
	\label{assumption:cond ind}
\end{assumption}
\begin{assumption}
	$$
	\begin{aligned}
		\left[\left\{\frac{\Dns}{\sqrt{n}}\right\}_{s \in \mathcal{S}} \mid S^{(n)}\right] \stackrel{d}{\rightarrow} N\left(0, \Sigma_{D}\right) \text { a.s., where } 
		\Sigma_{D} &=\diag{p(s) q(s): s \in \mathcal{S}} \\ \text{ with } 0 \leq  q(s)\leq\pi(1-\pi).
	\end{aligned}
	$$
	\label{assumption:normal}
\end{assumption}

\begin{assumption}
	$\frac{\Dns}{\ns{s}} = o_P(1) $ for all $s \in \mathcal{S}$.
	\label{assumption:normal_weaker}
\end{assumption}

The second requirement in Assumption \ref{assumption:nondegenerate} is introduced to exclude degenerate cases. It is still possible that $\Prob{S_i = s, X_i = x} = 0$ for some combinations of $s$ and $x$. In such situations, we simply ignore these combinations, and our inference remains unchanged. Assumption \ref{assumption:cond ind} requires that given the strata, the treatment assignments are conditionally independent of the potential outcomes and baseline covariates. Assumption \ref{assumption:normal}, which delineates the asymptotic behavior of imbalances within strata, was proposed by \cite{Bugni2018} to study statistical inference under covariate-adaptive randomization. This assumption is satisfied by many covariate-adaptive randomization schemes, such as simple randomization, stratified block randomization, and stratified biased-coin design. Assumption \ref{assumption:normal_weaker} is less restrictive than Assumption \ref{assumption:normal}.
Randomization methods that exhibit a dependence structure across strata, such as the minimization methods proposed by \cite{Pocock1975} and class of designs proposed by \cite{Hu2012}, are accommodated by Assumption \ref{assumption:normal_weaker}. 

\section{{Interaction tests for stratification covariates}} \label{section:binary_used}

\subsection{Interaction term t-test}
This section discusses interaction tests for stratification covariates. The following assumption always holds if covariate $X_i$ is a stratification covariate:

\begin{assumption}
	There exists a map $X$: $\mathcal{S} \rightarrow \mathcal{X}$ and $X_i = \sum_{s \in \mathcal{S}}X(s) \ind{S_i = s}$.
	\label{assumption:stratify}
\end{assumption}

Assumption \ref{assumption:stratify} requires that in each stratum $s$, covariate $X_i$ remains constant across all units. In addition to stratification covariates, this assumption holds for any covariate for which the value is uniquely determined by the stratum label $S_i$.

First, we focus on the scenario in which $X_i$ is binary, i.e., $X(s)$ can either be 1 or 0. Although $X_i$ is binary, the number of strata can exceed two. For example, consider a scenario in which both gender (male, female) and discrete severity levels (mild, moderate, or severe) are used for stratification. Then, the total number of strata is six, and the interaction effect is tested for gender, which is binary. In Section \ref{section:multi}, we extend our analysis to categorical covariates that have more than two levels.

When $X_i$ is binary, the null hypothesis of no interaction can be formulated as
\begin{equation}
	\begin{aligned}
		H_0: \Econd{\Yt - \Yc}{X_i = 1} = \Econd{\Yt - \Yc}{X_i = 0}.
	\end{aligned}
	\label{null:t}
\end{equation}
Next, we consider the linear model with an interaction term:
\begin{equation}
	\begin{aligned}
		Y_i = \alpha + \beta A_i + \gamma X_i + \delta A_i X_i + \epsilon_i,
	\end{aligned}
	\label{linear:simple}
\end{equation}
where the interaction effect $\delta = \tau_1 - \tau_0 = \Econd{\Yt - \Yc}{X_i = 1} - \Econd{\Yt - \Yc}{X_i = 0}$.
The  ordinary least squares (OLS) estimator of $\delta$ in (\ref{linear:simple}) is
\begin{equation*}
	\hat \delta = \Yax{1}{1} - \Yax{0}{1} - \Yax{1}{0} + \Yax{0}{0}.
\end{equation*}
First, we show that the usual $t\text{-test}$ for $\delta = 0$ using heteroscedasticity-robust variance estimators is conservative. We refer to this test as the ``interaction term $t\text{-test}$'' because it is aimed at testing the coefficient on the interaction term $A_iX_i$. Next, we modify the interaction term $t\text{-test}$ to establish an valid test. We emphasize that we only use the linear model (\ref{linear:simple}) as a working model and allow the true data generating model to be arbitrarily complex. 

To describe the asymptotic behavior of $\hat \delta$, we introduce the following notation:
\begin{equation*}
\begin{aligned}
	 \zetaStrY{x} &= \px{x}\Big[\frac{1}{\pi}\Varcond{\tildeYa{1}}{X_i = x} + \frac{1}{1-\pi}\Varcond{\tildeYa{0}}{X_i = x} \Big], \\
	\zetaStrA{x} &=  \px{x}E\Big(q(S_i) \Big[\frac{1}{\pi} \Econd{\maxZ{1}{x}} {S_i} + \frac{1}{1-\pi} \Econd{\maxZ{0}{x}}{S_i}\Big]^2 \Bigm| X_i = x\Big), \\
	\zetaStrH{x} &=  \px{x}E\Big( \Big[\Econd{\maxZ{1}{x}}{S_i} -  \Econd{\maxZ{0}{x}}{S_i}\Big]^2 \Bigm| X_i = x\Big), 
    \end{aligned}
\end{equation*}
where $\tildeYa{a} = \Ya - \Econd{\Ya}{ S_i}$, $m_{ax}(Z_i) = \Econd{\Ya}{Z_i} - \mu_{ax}$, for $a \in \{1, 0\}$, $x \in \{1, 0\}$.

\begin{theorem}
	Under Assumptions  \ref{assumption:nondegenerate}, \ref{assumption:cond ind},  \ref{assumption:normal},  and \ref{assumption:stratify},
	\begin{equation}
		\begin{aligned}
			{\sqrt{n}} {\begin{pmatrix}
					\Yax{1}{0} - \Yax{0}{0} - \tau_0 \\
					\Yax{1}{1} - \Yax{0}{1}-  \tau_1 \\
			\end{pmatrix}}
			\overset{d}{\rightarrow}
			N\left({\begin{pmatrix}
					0 \\
					0 \\
			\end{pmatrix}}, {\begin{pmatrix}
					v^2_{0} & 0 \\
					0 & v^2_{1} \\
			\end{pmatrix}}\right), 
		\end{aligned}
	\end{equation}
	where $v^2_{0} = (1/p_0^2)\{\zetaStrY{0} + \zetaStrA{0} + \zetaStrH{0}\}$ and $v^2_{1} = (1/p_1^2)\{\zetaStrY{1} + \zetaStrA{1} + \zetaStrH{1}\}$.
	
	\label{thm:t}
\end{theorem}

Theorem \ref{thm:t} specifies that the joint distribution of simple difference-in-means estimators converges weakly to a bivariate normal distribution. The diagonal covariance matrix implies that $\Yax{1}{1} - \Yax{0}{1}$  and $\Yax{1}{0} - \Yax{0}{0}$ are asymptotically independent. This independence is nontrivial, as we do not assume the linear model (\ref{linear:simple}) to be the true data generating model. As described in the next section, a correlation between $\Yax{1}{1} - \Yax{0}{1}$  and $\Yax{1}{0} - \Yax{0}{0}$ may exist when $X_i$ is an additional covariate.

\begin{remark}
In addition to facilitating interaction testing, Theorem \ref{thm:t} also allows for the inference of conditional treatment effects $\tau_1$ and $\tau_0$. Separate confidence intervals or confidence region can be formed for these conditional treatment effects. However, without further exploration of this aspect, our primary focus in this article is on the contrasting interaction effect $\delta = \tau_1 - \tau_0$.
\end{remark}

\begin{theorem}
	Under Assumptions \ref{assumption:nondegenerate}, \ref{assumption:cond ind}, \ref{assumption:normal},  and \ref{assumption:stratify},
	\begin{equation*}
		\begin{aligned}
			\frac{\Yax{1}{1} - \Yax{0}{1} - \Yax{1}{0} + \Yax{0}{0} - (\tau_1 - \tau_0)}{
				\left(\frac{\hatsigmaSqax{1}{1}}{\nax{1}{1}} + \frac{\hatsigmaSqax{0}{1}}{\nax{0}{1}} + \frac{\hatsigmaSqax{1}{0}}{\nax{1}{0}} + \frac{\hatsigmaSqax{0}{0}}{\nax{0}{0}}\right)^{1/2}
			}
			\stackrel{d}{\rightarrow}
			N(0,\zeta^2_{t\text{-test}}),
		\end{aligned}
	\end{equation*}
	where $\zeta^2_{t\text{-test}} \leq 1$. The inequality is strict unless
	\begin{equation}
		\begin{aligned}
			\left\{\pi(1-\pi) - q(s)\right\}\Big[\frac{1}{\pi}\Econd{\maxZ{1}{1}}{S_i = s} + \frac{1}{1-\pi}\Econd{\maxZ{0}{1}}{S_i = s}\Big] = 0
			\label{cond:1}
		\end{aligned}
	\end{equation}
	for all $s\in \mathcal{S}$ satisfying $X(s) = 1$, and 
	\begin{equation}
		\begin{aligned}
			\left\{\pi(1-\pi) - q(s)\right\}\Big[\frac{1}{\pi}\Econd{\maxZ{1}{0}}{S_i = s} + \frac{1}{1-\pi}\Econd{\maxZ{0}{0}}{S_i = s}\Big] = 0 
			\label{cond:0}
		\end{aligned}
	\end{equation}
	for all $s\in \mathcal{S}$ satisfying $X(s) = 0$.
	\label{thm:t_old}
\end{theorem}

\begin{remark}
 One sufficient condition for $\zeta^2_{t\text{-test}} = 1$ is the use of simple randomization, where $q(s) = \pi(1-\pi) $ for all $s \in \mathcal{S}$. Another condition is that the strata are homogeneous within each level of 
 $X_i$. Specifically, given $a \in \{1,0\}$ and $x \in \{1, 0\}$, $\Econd{\Ya}{S_i = s} = \Econd{\Ya}{X_i = x}$ holds for all $s \in  X^{-1}(x) = \{s \in \mathcal{S}: X(s) = x\}$. This condition is satisfied, for instance, when $X_i$ is the sole covariate used for stratification.
 \label{re:Xonly}
\end{remark}

The usual interaction term $t\text{-test}$ can be defined as
\begin{equation}
	\begin{aligned}
		\phi_n^{t\text{-test}}(V^{(n)}) = I\{|T_n^{t\text{-test}}(V^{(n)})| > z_{1-\alpha/2}\},
	\end{aligned}
 \label{eq:t_test_usual}
\end{equation}
where 
\begin{equation*}
	\begin{aligned}
		T_n^{t\text{-test}}(V^{(n)}) = \frac{\Yax{1}{1} - \Yax{0}{1} - \Yax{1}{0} + \Yax{0}{0}}{
			\left(\frac{\hatsigmaSqax{1}{1}}{\nax{1}{1}} + \frac{\hatsigmaSqax{0}{1}}{\nax{0}{1}} + \frac{\hatsigmaSqax{1}{0}}{\nax{1}{0}} + \frac{\hatsigmaSqax{0}{0}}{\nax{0}{0}}\right)^{1/2}
		},
	\end{aligned}
\end{equation*}
and $z_{1-\alpha/2}$ is the $1-\alpha/2$ quantile of the standard normal distribution.

Then, under the assumptions of Theorem \ref{thm:t_old}, for testing (\ref{null:t}) at level 
$\alpha \in (0,1)$ when the null hypothesis is true, 
$$
\lim_{n \rightarrow \infty}\E{\phi_n^{t\text{-test}}(V^{(n)})} \leq \alpha,
$$ and the inequality is strict unless both (\ref{cond:1}) and (\ref{cond:0}) hold. Thus, the usual interaction term $t$-test is conservative.

Based on Theorem \ref{thm:t}, we can modify the usual interaction term $t$-test and establish an valid test by replacing $\px{x}$, $\zetaStrY{x}$, $\zetaStrA{x}$, and $\zetaStrH{x}$ with their sample versions $ \hatpx{x}$, $\hatzetaStrY{x}$, $\hatzetaStrA{x}$, and $\hatzetaStrH{x}$, respectively. The exact expressions for these sample quantities can be found in the Appendix. The modified interaction term $t$-test can be formulated as
\begin{equation}
	\begin{aligned}
		\phi_n^{t\text{-test,mod}}(V^{(n)}) = I\{|T_n^{t\text{-test,mod}}(V^{(n)})| > z_{1-\alpha/2}\},
	\end{aligned}
	\label{eq:t_test_mod}
\end{equation}
where 
\begin{equation*}
	\begin{aligned}
		T_n^{t\text{-test,mod}}(V^{(n)}) =	\frac{\sqrt{n}(\Yax{1}{1} - \Yax{0}{1} - \Yax{1}{0} + \Yax{0}{0})}{
			(\hat v^2_{1} + \hat v^2_{0})^{1/2}
		},
	\end{aligned}
\end{equation*}
in which $\hat v^2_{0} = (1/\hat p_0^2)\{\hatzetaStrY{0} + \hatzetaStrA{0} + \hatzetaStrH{0}\}$ and $\hat v^2_{1} = (1/\hat p_1^2)\{\hatzetaStrY{1} + \hatzetaStrA{1} + \hatzetaStrH{1}\}$.

\begin{theorem}
	\label{thm:t_mod}
	Under Assumptions  \ref{assumption:nondegenerate}, \ref{assumption:cond ind},
	\ref{assumption:normal},  and \ref{assumption:stratify}, for testing (\ref{null:t}) at level $\alpha \in (0,1)$ when the null hypothesis is true,  $$\lim_{n \rightarrow \infty}\E{\phi_n^{t\text{-test,mod}}(V^{(n)})} = \alpha.$$ 
\end{theorem}

\subsection{Stratified-adjusted interaction test}
In stratified randomized experiments, the following stratified-adjusted difference-in-means estimator
\begin{equation}
	\begin{aligned}
		\sum_{s \in \mathcal{S}} \frac{\ns{s}}{n}\left\{\hatmuas{1}{s} - \hatmuas{0}{s}\right\},
		\label{eq:Bugni2019}
	\end{aligned}
\end{equation}
is commonly used to estimate the overall treatment effect $\E{\Yt - \Yc}$, where $\hatmuas{a}{s} = \{{1}/{\nas{a}{s}}\}\sum_{i=1}^{n}Y_i \ind{A_i = a, S_i = s}$ denotes the sample mean of $Y_i$ with $A_i = a$ and $S_i = s$ for $a \in \{1,0\}$, $s \in \mathcal{S}$.
The stratified-adjusted estimator (\ref{eq:Bugni2019}) is intuitive given that
\begin{equation*}
	\begin{aligned}
		\E{\Yt - \Yc} = \sum_{s \in \mathcal{S}}p(s)\Econd{\Yt - \Yc}{S_i = s}.
	\end{aligned}
\end{equation*}
In certain stratified randomized experiments, such estimators are more appealing than the simple difference-in-means estimator $\bar{Y}_{n,1} - \bar{Y}_{n,0} $. 
\cite{Rosenberger2015} argued a like-stratified analysis for any clinical trial with a stratified randomized randomization.
The property of estimator  (\ref{eq:Bugni2019}) under stratified randomization has been studied by \cite{imbens2015} and \cite{Bugni2019} from different perspectives. 


When binary covariate $X_i$ satisfies Assumption \ref{assumption:stratify}, 
\begin{equation}
	\begin{aligned}
		\tau_x = \Econd{\Yt - \Yc}{X_i = x} = \frac{1}{\px{x} }\sum_{s,X(s)=x} p(s)	\Econd{\Yt - \Yc}{S_i = s},
	\end{aligned}
\end{equation}
for $x \in \{1,0\}$, we can use the stratified-adjusted estimator 
$$\sum_{s,X(s)=x} \frac{\ns{s}}{\nx{x}}\left\{\hatmuas{1}{s} - \hatmuas{0}{s}\right\}$$
to estimate $\tau_x$  and
\begin{equation*}
	\sum_{s,X(s)=1}  \frac{\ns{s}}{\nx{1}}\left\{ \hatmuas{1}{s} - \hatmuas{0}{s}\right\} -  \sum_{s,X(s)=0}  \frac{\ns{s}}{\nx{0}}\left\{\hatmuas{1}{s} - \hatmuas{0}{s}\right\}
\end{equation*}
to test the null hypothesis of no interaction.

\begin{theorem}
	Under Assumptions  \ref{assumption:nondegenerate}, \ref{assumption:cond ind},  \ref{assumption:normal_weaker}, and \ref{assumption:stratify},
	\begin{equation}
		\begin{aligned}
			{\sqrt{n}} {\begin{bmatrix}
					\sum_{s,X(s)=0} \{\ns{s}/\nx{0}\}\{\hatmuas{1}{s} - \hatmuas{0}{s}\}  - \tau_0\\
					\sum_{s,X(s)=1} \{\ns{s}/\nx{1}\}\{\hatmuas{1}{s} - \hatmuas{0}{s}\}  - \tau_1 \\
			\end{bmatrix}}
			\overset{d}{\rightarrow}
			N\left({\begin{pmatrix}
					0 \\
					0 \\
			\end{pmatrix}}, {\begin{pmatrix}
					u^2_{0} & 0 \\
					0 & u^2_{1} \\
			\end{pmatrix}}\right), 
		\end{aligned}
	\end{equation}
	where $u^2_{0} = (1/\px{0}^2)\{\zetaStrY{0} +  \zetaStrH{0}\}$ and $u^2_{1} =(1/\px{1}^2)\{\zetaStrY{1} +  \zetaStrH{1}\}$.
	\label{thm:t_strata}
\end{theorem}

Recalling the assumptions pertaining to Theorem \ref{thm:t}, we replace the more stringent Assumption \ref{assumption:normal} with the more relaxed Assumption \ref{assumption:normal_weaker}. This adjustment enables the accommodation of randomization methods that have a dependence structure within the strata. Furthermore, the limit distribution in Theorem \ref{thm:t_strata} is free of $q(s)$, solely relying on the underlying data generation model and target treatment proportion $\pi$. The limit distribution does not depend on specific covariate-adaptive randomization methods.

Based on Theorem \ref{thm:t_strata}, we can obtain another valid interaction test. This stratified-adjusted interaction test can be formulated as
\begin{equation}
	\begin{aligned}
		\phi_n^{\text{strat}}(V^{(n)}) = I\{|T_n^{\text{strat}}(V^{(n)})| > z_{1-\alpha/2}\},
	\end{aligned}
	\label{eq:t_test_strata}
\end{equation}
where 
\begin{equation*}
	\begin{aligned}
		T_n^{\text{strat}}(V^{(n)}) =	\frac{\sqrt{n}\left[\sum_{s,X(s)=1} \frac{\ns{s}}{\nx{1}}\{ \hatmuas{1}{s} - \hatmuas{0}{s}\}- \sum_{s,X(s)=0} \frac{\ns{s}}{\nx{0}}\{ \hatmuas{1}{s} - \hatmuas{0}{s}\} \right]}{(\hat u^2_{1} + \hat u^2_{0})^{1/2}
		},
	\end{aligned}
\end{equation*}
in which $\hat u^2_{0} = (1/\hat p_0^2)\{\hatzetaStrY{0} + \hatzetaStrH{0}\}$ and $\hat u^2_{1} = (1/\hat p_1^2)\{\hatzetaStrY{1}  + \hatzetaStrH{1}\}$.

\begin{theorem}
	\label{thm:t_strata_mod}
	Under Assumptions  \ref{assumption:nondegenerate}, \ref{assumption:cond ind}, \ref{assumption:normal_weaker}, and \ref{assumption:stratify}, for testing (\ref{null:t}) at level $\alpha \in (0,1)$ when the null hypothesis is true,
	$$\lim_{n \rightarrow \infty}\E{\phi_n^{\text{strat}}(V^{(n)})} = \alpha.$$ 
\end{theorem}

\begin{remark}
	The asymptotic variances derived using Theorems \ref{thm:t_strata} and \ref{thm:t} satisfy $u_1^2 + u_0^2 \leq v_1^2 + v_0^2$, and this inequality is strict unless $\zetaStrA{1} = \zetaStrA{0} = 0$. A sufficient condition for the equality to hold is that the randomization method satisfies $q(s) = 0$ for all $s \in \mathcal{S}$. Such randomization methods are considered to achieve ``strong balance'', as indicated by \cite{Bugni2018}. Consequently, the stratified-adjusted test (\ref{eq:t_test_strata}) is asymptotically more powerful than the modified test (\ref{eq:t_test_mod}). Through simulations  (Section \ref{section:simulation}), we demonstrate that the power of the stratified-adjusted test can be much higher than that of the modified test.
 \label{re:used_strat}
\end{remark}

\section{{Interaction tests for additional covariates}} \label{section:binary_add}

\subsection{Interaction term t-test}
As discussed in Section \ref{section:binary_used}, we explored interaction tests for binary stratification covariates. Now, we extend our study to include interaction tests for binary additional covariates that are not used for stratification. These additional covariates typically do not meet Assumption \ref{assumption:stratify}. However, we continue to analyze the linear model (\ref{linear:simple}) and introduce new notation to facilitate our analysis:
\begin{equation*}	
\begin{aligned}
	\zetaAddY{x}
	&= \frac{1}{\pi}\Var{\tilde{r}_{x,i}(1)} + \frac{1}{1-\pi}\Var{\tilde{r}_{x,i}(0)}, \\	
	\zetaAddA{x}
	&= E\Big(q(S_i) \Big[\frac{1}{\pi} \Econd{\raxZ{1}{x}}{S_i} + \frac{1}{1-\pi} \Econd{\raxZ{0}{x}}{S_i}\Big]^2 \Big),\\	
	\zetaAddH{x}
	&=  E\Big( \Big[\Econd{\raxZ{1}{x}}{S_i} -  \Econd{\raxZ{0}{x}}{S_i}\Big]^2\Big),\\
	\zetaAddYxy{x}{y}
	&= -E\Big[\frac{1}{\pi}\Econd{\raxZ{1}{x}}{S_i}\Econd{\raxZ{1}{y}}{S_i} + \frac{1}{1-\pi} \Econd{\raxZ{0}{x}}{S_i}\Econd{\raxZ{0}{y}}{S_i}\Big], \\
	\zetaAddAxy{x}{y}
	&=E\Big(q(S_i) \Big[\frac{1}{\pi}\Econd{\raxZ{1}{x}}{S_i} + \frac{1}{1-\pi} \Econd{\raxZ{0}{x}}{S_i}\Big] \\
    &\quad\qquad \Big[\frac{1}{\pi} \Econd{\raxZ{1}{y}}{S_i} + \frac{1}{1-\pi} \Econd{\raxZ{0}{y}}{S_i}\Big]\Big),\\
	\zetaAddHxy{x}{y}
	&=  E\Big(\Big[ \Econd{\raxZ{1}{x}}{S_i}- \Econd{\raxZ{0}{x}}{S_i}\Big]\Big[ \Econd{\raxZ{1}{y}}{S_i}- \Econd{\raxZ{0}{y}}{S_i}\Big]\Big),\\
 \end{aligned}
\end{equation*}
where $r_{x,i}(a) = \{\Ya - \mu_{ax}\}\ind{X_i = x}$, $\tilde{r}_{x,i}(a) = r_{x,i}(a) - \Econd{r_{x,i}(a)}{S_i}$, and ${r}_{ax}(Z_i) =  \Econd{r_{x,i}(a)}{Z_i}$, for $a \in \{1, 0\}$, $x, y \in \{1, 0\}$.

\begin{remark}
    Compared with the notation used in Section \ref{section:binary_used}, the notation used in this analysis is more complex. This increased complexity arises from the fact that the value of covariate $X_i$ is no longer uniquely determined by the stratum label $S_i$. Consequently, the stratum-specific effect can influence both $\Yax{1}{1} - \Yax{0}{1}$ and $\Yax{1}{0} - \Yax{0}{0}$, introducing additional complexity to the inference. Fortunately, when we later consider the stratified-adjusted test, this complexity is alleviated, leading to more concise results.
\end{remark}
\begin{theorem}
	Under Assumptions  \ref{assumption:nondegenerate}, \ref{assumption:cond ind}, and \ref{assumption:normal},
	\begin{equation*}
		\begin{aligned}
			{\sqrt{n}} {\begin{pmatrix}
					\Yax{1}{0} - \Yax{0}{0} - \tau_0 \\
					\Yax{1}{1} - \Yax{0}{1}-  \tau_1 \\
			\end{pmatrix}}
			\overset{d}{\rightarrow}
			N\left({\begin{pmatrix}
					0 \\
					0 \\
			\end{pmatrix}}, {\begin{pmatrix}
					w^2_{0} & w_{10} \\
					w_{10} & w^2_{1} \\
			\end{pmatrix}}\right), 
		\end{aligned}
	\end{equation*}
	where $w^2_{0} = (1/p_0^2)\{\zetaAddY{0} + \zetaAddA{0} + \zetaAddH{0}\}$, $w^2_{1} = (1/p_1^2)\{\zetaAddY{1} + \zetaAddA{1} + \zetaAddH{1}\}$, and $w_{10} =(\px{1}\px{0})^{-1}\{\zetaAddYxy{1}{0} + \zetaAddAxy{1}{0} + \zetaAddHxy{1}{0}\} $.
	\label{thm:t_add}
\end{theorem}

Theorem \ref{thm:t_add} generalizes Theorem \ref{thm:t}, which additionally requires Assumption \ref{assumption:stratify}. However, the asymptotic covariance matrix in Theorem \ref{thm:t_add} is more complex than that in Theorem \ref{thm:t}. Specifically, $\Yax{1}{1} - \Yax{0}{1}$ and $\Yax{1}{0} - \Yax{0}{0}$ are not necessarily asymptotically independent. A sufficient condition for the asymptotic covariance $w_{10}$ to be zero is the use of simple randomization, where $q(s) = \pi(1-\pi)$ holds for all $s \in \mathcal{S}$. Conversely, under other randomization methods where $q(s) < \pi(1-\pi)$ for some $s \in \mathcal{S}$, the simple difference-in-means estimators $\Yax{1}{1} - \Yax{0}{1}$ and $\Yax{1}{0} - \Yax{0}{0}$ may be correlated. 


\begin{theorem}
	\label{thm:t_add_old}
	Under Assumptions  \ref{assumption:nondegenerate}, \ref{assumption:cond ind}, and \ref{assumption:normal},
	\begin{equation*}
		\begin{aligned}
			\frac{\Yax{1}{1} - \Yax{0}{1} - \Yax{1}{0} + \Yax{0}{0} - (\tau_1 - \tau_0)}{
				\left(\frac{\hatsigmaSqax{1}{1}}{\nax{1}{1}} + \frac{\hatsigmaSqax{0}{1}}{\nax{0}{1}} + \frac{\hatsigmaSqax{1}{0}}{\nax{1}{0}} + \frac{\hatsigmaSqax{0}{0}}{\nax{0}{0}}\right)^{1/2}
			}
			\stackrel{d}{\rightarrow}
			N(0,\zeta^2_{t\text{-test}^\prime}),
		\end{aligned}
	\end{equation*}
	where $\zeta^2_{t\text{-test}^\prime}\leq 1$. The inequality is strict unless
	\begin{equation}
		\begin{aligned}
			\{\pi(1-\pi)-q(s)\}&\Big(\frac{1}{\px{1}}\Big[\frac{1}{\pi} \Econd{\raxZ{1}{1}}{S_i=s} + \frac{1}{1-\pi} \Econd{\raxZ{0}{1}}{S_i=s}\Big]\\ 
            &\quad - \frac{1}{\px{0}}\Big[\frac{1}{\pi} \Econd{\raxZ{1}{0}}{S_i=s} + \frac{1}{1-\pi} \Econd{\raxZ{0}{0}}{S_i=s}\Big]\Big) = 0
		\end{aligned}
            \label{cond:1_add}
	\end{equation}
	for all $s \in \mathcal{S}$.
	\label{theorem:t_add}
\end{theorem}

\begin{remark}
 One sufficient condition for $\zeta^2_{t\text{-test}^\prime} = 1$ is the use of simple randomization. Another condition is that the strata are homogeneous within each level of $X_i$. Specifically, given $a \in \{1,0\}$ and $x \in \{1,0\}$, $\Econd{\Ya}{S_i = s, X_i = x} = \Econd{\Ya}{X_i = x}$ holds for all $s \in \mathcal{S}$. This condition is satisfied, for instance, in the absence of stratification, resulting in a single stratum that includes all units.
  \label{re:noStratum}
\end{remark}	

Consider the usual interaction term $t\text{-test}$ specified in (\ref{eq:t_test_usual}). Under the assumptions of Theorem \ref{theorem:t_add}, for testing (\ref{null:t}) at level 
$\alpha \in (0,1)$ when the null hypothesis is true, 
$$
\lim_{n \rightarrow \infty}\E{\phi_n^{t\text{-test}}(V^{(n)})} \leq \alpha,
$$ and the inequality is strict unless (\ref{cond:1_add}) holds.

Similarly, based on Theorem \ref{thm:t_add}, we can modify the usual interaction term $t$-test and obtain an valid test. The modified interaction term $t$-test can be expressed as
\begin{equation}
	\begin{aligned}
		\phi_n^{t\text{-test,mod}^\prime}(V^{(n)}) = I\{| T_n^{t\text{-test,mod}^\prime}(V^{(n)})| > z_{1-\alpha/2}\},
	\end{aligned}
	\label{eq:t_test_mod_add}
\end{equation}
where 
\begin{equation*}
	\begin{aligned}
		T_n^{t\text{-test,mod}^\prime}(V^{(n)}) =	\frac{\sqrt{n}(\Yax{1}{1} - \Yax{0}{1} - \Yax{1}{0} + \Yax{0}{0})}{
			(\hat w^2_{1} + \hat w^2_{0} - 2\hat w_{10})^{1/2}
		},
	\end{aligned}
\end{equation*}
in which $\hat w^2_{0} = (1/\hat p_0^2)\{\hatzetaAddY{0} + \hatzetaAddA{0} + \hatzetaAddH{0}\}$, $\hat w^2_{1} = (1/\hat p_1^2)\{\hatzetaAddY{1} + \hatzetaAddA{1} + \hatzetaAddH{1}\}$, and $\hat w_{10} =(\hatpx{1}\hatpx{0})^{-1}\{\hatzetaAddYxy{1}{0} + \hatzetaAddAxy{1}{0} + \hatzetaAddHxy{1}{0}\} $.

\begin{theorem}
	\label{thm:t_add_mod}
	Under Assumptions \ref{assumption:nondegenerate}, \ref{assumption:cond ind}, and \ref{assumption:normal}, for testing (\ref{null:t}) at level $\alpha \in (0,1)$ when the null hypothesis is true,
	$$\lim_{n \rightarrow \infty}\E{\phi_n^{t\text{-test,mod}^\prime}(V^{(n)})} = \alpha.$$ 
\end{theorem}

\subsection{Stratified-adjusted interaction test}

For $x \in \mathcal{X}$ and $s \in \mathcal{S}$, let $\nxs{x}{s} = \sum_{i=1}^n \ind{S_i=s, X_i = x}$, $\naxs{1}{x}{s}= \sum_{i=1}^n A_i \ind{S_i=s, X_i = x}$, $\naxs{0}{x}{s} = \sum_{i=1}^n (1-A_i)\ind{S_i=s, X_i = x}$; let $\hatmuaxs{1}{x}{s} = \{{1}/{\naxs{1}{x}{s}}\}\sum_{i=1}^{n}Y_i A_i \ind{S_i=s, X_i = x}$ and $\hatmuaxs{0}{x}{s} = \{{1}/{\naxs{0}{x}{s}}\}\sum_{i=1}^{n}Y_i (1-A_i) \ind{S_i=s, X_i = x}$. Given that
$\tau_x$ equals 
\begin{equation*}
	\begin{aligned}
		\frac{1}{\px{x}}\sum_{s \in \mathcal{S}}\Prob{S_i = s, X_i=x}\Econd{\Yt - \Yc}{S_i = s,X_i=x} ,
	\end{aligned}
\end{equation*}
we can use the stratified-adjusted estimator $\sum_{s \in \mathcal{S}} \{{\nxs{x}{s}}/{\nx{x}}\}\{\hatmuaxs{1}{x}{s} - \hatmuaxs{0}{x}{s}\}$ to estimate $\tau_x$ and
\begin{equation}
	\begin{aligned}
		\sum_{s \in \mathcal{S}} \frac{\nxs{1}{s}}{\nx{1}}\{\hatmuaxs{1}{1}{s} - \hatmuaxs{0}{1}{s}\} - \sum_{s \in \mathcal{S}} \frac{\nxs{0}{s}}{\nx{0}}\{\hatmuaxs{1}{0}{s} - \hatmuaxs{0}{0}{s}\}
	\end{aligned}
	\label{eq:strata_add}
\end{equation}
to test the null hypothesis of no interaction.

To clarify the asymptotic behavior of the quantity in (\ref{eq:strata_add}), we define
\begin{equation*}	
\begin{aligned}
	\zetaCheckY{x} 
	&= \px{x}\Big[\frac{1}{\pi}\Varcond{\checkYa{1}}{X_i = x} + \frac{1}{1-\pi}\Varcond{\checkYa{0}}{X_i = x} \Big], \\
	\zetaCheckH{x} 
	&=  \px{x}E\Big[ \Big(\Econd{\maxZ{1}{x}}{S_i, X_i} -  \Econd{\maxZ{0}{x}}{S_i, X_i}\Big)^2 \Bigm| X_i = x\Big],
\end{aligned}
\end{equation*}
where $\checkYa{a} = \Ya - \Econd{\Ya}{S_i, X_i}$, for $a \in \{1, 0\}$, $x \in \{1, 0\}$.

\begin{theorem}
	\label{thm:t_strata_add}
	Under Assumptions  \ref{assumption:nondegenerate}, \ref{assumption:cond ind}, and \ref{assumption:normal_weaker},
	\begin{equation*}
		\begin{aligned}
			{\sqrt{n}} {\begin{bmatrix}
					\sum_{s \in \mathcal{S}} \{{\nxs{0}{s}}/{\nx{0}}\}\{\hatmuaxs{1}{0}{s} - \hatmuaxs{0}{0}{s}\}  - \tau_0\\
					\sum_{s \in \mathcal{S}} \{{\nxs{1}{s}}/{\nx{1}}\}\{\hatmuaxs{1}{1}{s} - \hatmuaxs{0}{1}{s}\} - \tau_1\\
			\end{bmatrix}}
			\overset{d}{\rightarrow}
			N\left({\begin{pmatrix}
					0 \\
					0 \\
			\end{pmatrix}}, {\begin{pmatrix}
					s^2_{0} & 0 \\
					0 & s^2_{1} \\
			\end{pmatrix}}\right), 
		\end{aligned}
	\end{equation*}
where $s^2_{0} = ({1}/{\px{0}^2})\{\zetaCheckY{0}  + \zetaCheckH{0}\}$, and $s^2_{1} = ({1}/{\px{1}^2})\{\zetaCheckY{1}  + \zetaCheckH{1}\}$.
\end{theorem}

Compared with the covariance matrix in Theorem \ref{thm:t_add}, the covariance matrix of the stratified-adjusted estimators in Theorem \ref{thm:t_strata_add} has a much simpler structure. Specifically, the covariance matrix is always diagonal, regardless of whether Assumption \ref{assumption:stratify} is satisfied. This implies that the stratified-adjusted estimators are always asymptotically independent. Furthermore, the covariance matrix does not rely on $q(s)$ and thus does not depend on specific randomization methods.

The valid stratified-adjusted interaction test is given by
\begin{equation}
	\begin{aligned}
		\phi_n^{\text{strat}^\prime}(V^{(n)}) = I\{|  T_n^{\text{strat}^\prime}(V^{(n)})| > z_{1-\alpha/2}\},
	\end{aligned}
	\label{eq:t_test_strata_add}
\end{equation}
where 
\begin{equation*}
	\begin{aligned}
		T_n^{\text{strat}^\prime}(V^{(n)}) =	\frac{\sqrt{n}\left[\sum_{s \in \mathcal{S}} \frac{\nxs{1}{s}}{\nx{1}}\{\hatmuaxs{1}{1}{s} - \hatmuaxs{0}{1}{s}\}- \sum_{s \in \mathcal{S}} \frac{\nxs{0}{s}}{\nx{0}}\{\hatmuaxs{1}{0}{s} - \hatmuaxs{0}{0}{s}\}\right]}{(\hat s^2_{1} + \hat s^2_{0})^{1/2}
		},
	\end{aligned}
\end{equation*}
in which $\hat s^2_{0} = ({1}/{\hatpx{0}^2})\{\hatzetaCheckY{0}  + \hatzetaCheckH{0}\}$ and $\hat s^2_{1} = ({1}/{\hatpx{1}^2})\{\hatzetaCheckY{1}  + \hatzetaCheckH{1}\}$.
\begin{theorem}
	\label{thm:add_strata}
	Under Assumptions  \ref{assumption:nondegenerate}, \ref{assumption:cond ind}, and \ref{assumption:normal_weaker}, for testing (\ref{null:t}) at level $\alpha \in (0,1)$ when the null hypothesis is true,
	$$\lim_{n \rightarrow \infty}\E{ \phi_n^{\text{strat}^\prime}(V^{(n)})} = \alpha.$$ 
\end{theorem}

\begin{remark}
	 Lemma 8 (Appendix) shows that $ s^2_{1} + s^2_{0} \leq w^2_{1} + w^2_{0} - 2 w_{10}$, indicating that the stratified-adjusted test (\ref{eq:t_test_strata_add}) is asymptotically more powerful than the modified test (\ref{eq:t_test_mod_add}). However, unlike the case in which $X_i$ is a stratification covariate, the randomization method achieving ``strong balance'' cannot ensure that the equality holds. In our simulations (Section \ref{section:simulation}), we demonstrate that the power of the stratified-adjusted test may be higher than that of the modified test, even when the randomization method achieves ``strong balance''.
 \label{re:add_strat}
\end{remark}

To conclude this section, we emphasize that all inferential properties and procedures discussed herein (Section \ref{section:binary_add}) will reduce to their counterparts in Section \ref{section:binary_used} when Assumption \ref{assumption:stratify} is met. Therefore, while it is conceptually crucial to distinguish between stratification covariates and additional baseline covariates, the framework utilized in this section is applicable to both types. A significant implication, from a computational standpoint, is that the same algorithmic procedure can be employed for testing interaction effects, irrespective of the type of covariate under investigation. This advantage offers practical convenience for clinical researchers and statisticians, as it allows for the reuse of the same programming code.

\section{{Simulation study}}\label{section:simulation}
We examine the type \RNum{1} error and power of tests for all three interaction tests, i.e.,  the usual test, modified test, and stratified-adjusted test. We consider three randomization methods: simple randomization (SR), stratified block randomization (SBR), and stratified biased-coin design (SBCD). The block size is 6 and the biased-coin probability is  0.75. The nominal level $\alpha$ is $5\%$, the  number of units $n$ is 800, and simulations are based on 10,000 runs. In this section, we focus on scenarios in which the binary covariate $X$ is tested for interaction. 
	
Model 1: Linear model:
$$
\begin{cases}
	\Yt = \mu_1 +(\beta_{1} + \delta_{1})X_{i}   + \alpha_{1}W^*_{i} + \gamma_{1}X_{i}W^*_{i}+  \sigma_1 \epsilon_{1,i}, \\
	\Yc  = \mu_0 +\beta_{1} X_{i} + \alpha_{1}W^*_{i} + \sigma_0 \epsilon_{0,i}, \\
\end{cases}
$$
where $X$ is a Bernoulli random variable with a probability of 0.5, $W^* \sim N(0, 3^2)$ independent of $X$ is categorized to $W = \ind{W^* > 0}$. We assume that $\mu_1 = 4$, $\mu_0 = 1$, $\beta_{1} =3$, $\alpha_{1} = -2$, $\gamma_{1} = 4$, $\sigma_1 = 1$, $\sigma_0 = 0.5$, $ \epsilon_{1,i} \sim N(0, 1)$,  and $ \epsilon_{0,i} \sim N(0, 1)$. $\delta_1$ equals $0$ (null) or $1.5$ (alternative).

Model 2: Nonlinear model:
$$
\begin{cases}
	\Yt = \mu_1 + \exp\big\{(\beta_{1} + \delta_{1})X^*_{i}\big\}   + \alpha_{1}W^*_{i}+ \gamma_{1}X^*_{i}W^*_{i}+  \sigma_1(X^*_{i})\epsilon_{1,i}, \\
	\Yc  = \mu_0 +\exp\big(\beta_{1} X^*_{i}\big) + \alpha_{1}W^*_{i} + \sigma_0(X^*_{i})\epsilon_{0,i}, \\
\end{cases}
$$
where $X^* \sim \text{Unif}(-1, 1)$ is categorized to $X = \ind{X^* > 0}$, and $W^* \sim N(0, 2^2)$ independent of $X^*$ is categorized to $W = \ind{W^* > 0}$. We assume that $\mu_1 = 5$, $\mu_0 = 4$, $\beta_{1} =0.5$,  $\alpha_{1} = 2$, $\gamma_{1} = 6$, $\sigma_1(X^*_{i}) = \exp(0.5X^*_{i})$, $\sigma_0(X^*_{i}) = 0.5\exp(0.5X^*_{i})$, $ \epsilon_{1,i} \sim N(0, 1)$, and $ \epsilon_{0,i} \sim N(0, 1)$. $\delta_1$ equals $0$ (null) or $1.2$ (alternative).

Model 3: Binary outcome model with no overall treatment effect:
$$
\begin{cases}
	\Yt = I\big\{\mu_1 + (\beta_{1} + \delta_{1})(X_i-2/3)+ \alpha_{1}W_{i} + \gamma_{1}X_{i}W_{i} > U_{1,i}\big\}, \\
	\Yc = I\big\{\mu_0 + \beta_{1} (X_i-2/3) +\alpha_{1}W_{i}  > U_{0,i}\big\}, \\
\end{cases}
$$
where $X$ is a Bernoulli random variable with a probability of $2/3$, and $W$ equals $1$ or $-1$ with equal probabilities. We assume that $\mu_1 = \mu_0 = 4$, $\beta_{1} =1$, $\alpha_{1} = -3$, $\gamma_{1} = 6$, $U_{1,i} \sim \text{Unif}(0, 10)$, and $U_{0,i} \sim \text{Unif}(0, 10)$. $\delta_1$ equals $0$ (null) or $1.5$ (alternative).

Table \ref{tab:simulation_equal} presents the empirical type \RNum{1} error under equal allocation, $\pi = {1}/{2}$. First, we consider the case in which $X$ is used for stratification. When only $X$ is used to form the strata, the usual test is not conservative, as mentioned in Remark \ref{re:Xonly}. When both $X$ and $W$ are used for stratification, under SBR and SBCD, the rejection probability of the usual test is strictly less than the nominal level. Next, we consider the case in which $X$ is an additional covariate. When no covariate is used for stratification, resulting in a single stratum that includes all units, the usual test is not conservative, as mentioned in Remark \ref{re:noStratum}. When only $W$ is used to form the strata, the usual test is conservative under SBR and SBCD. Regardless of the stratification covariates and randomization method used, both the modified test and stratified-adjusted test exhibit rejection probabilities that closely match the nominal level.

Table \ref{tab:simulation_equal} presents the power of tests under the alternative hypothesis. First, we consider the case in which $X$ is used for stratification. When only $X$ is used for stratification, all three tests demonstrate similar levels of power regardless of the randomization method used. When both $X$ and $W$ are used for stratification, the stratified-adjusted test is the most powerful test. Under SBR and SBCD, two randomization methods that achieve ``strong balance'', the modified test is as powerful as the stratified-adjusted test, as mentioned in Remark \ref{re:used_strat}. Subsequently, we consider the case in which $X$ is an additional covariate. When no covariate is used to form the strata, all three tests exhibit the same power.  When only $W$ is used for stratification, the stratified-adjusted test still has the greatest power. However, even among the two randomization methods that achieves ``strong balance'', the modified test is less powerful than the stratified-adjusted test, as mentioned in Remark \ref{re:add_strat}.

Table \ref{tab:simulation_unequal} presents the empirical type \RNum{1} error and power of tests under unequal allocation, $\pi = {2}/{3}$. The results are similar to those presented in Table \ref{tab:simulation_equal}. 

\begin{table}[]
\caption{Rejection probabilities (in percentage points) for three interaction tests (on binary $X$), stratification mechanisms, and randomization methods under equal allocation, $\pi = 1/2$, and sample size $n = 800$}{
\tabcolsep=4.25pt
\renewcommand\arraystretch{0.5}
\small
\begin{tabular}{cccccccc}
\hline
              &           & \multicolumn{2}{c}{Stratification covariates} &                      & \multicolumn{3}{c}{Randomization methods}                        \\ \cline{3-4} \cline{6-8}
\multicolumn{1}{l}{} &\multicolumn{1}{c}{Model} & $X$ stratified         & $S(\cdot)$             &                      & SR                & SBR                  & SBCD                  \\ \hline 

$H_0$         & 1         & Yes                  & $X$                    &                      & 4.8 / 4.8 / 4.8 & 5.3 / 5.3 / 5.3 & 5.3 / 5.3 / 5.3   \\
              &           &                      & $X \times W$         &                      &  5.3 / 5.4 / 5.4 & 2.2 / 5.4 / 5.3 & 2.0 / 5.8 / 5.6 \\ 
              &           & No                   & --                     &                      & 5.4 / 5.5 / 5.5 & 4.9 / 4.9 / 4.9 & 5.2 / 5.2 / 5.2 \\ 
              &           &                      & $W$                  &                      & 5.4 / 5.4 / 5.3 & 3.7 / 5.3 / 5.5 & 3.5 / 5.3 / 5.5 \\ 
              &           &                      &                        & \multicolumn{1}{l}{} &                    & \multicolumn{1}{l}{} & \multicolumn{1}{l}{} \\
              & 2         & Yes                  & $X$                    &                      & 5.2 / 5.3 / 5.3 & 5.2 / 5.2 / 5.2 & 4.6 / 4.6 / 4.6 \\ 
              &           &                      & $X \times W$         &                      & 4.8 / 4.9 / 5.2 & 1.3 / 5.4 / 5.4 & 1.3 / 5.2 / 5.2 \\ 
              &           & No                   & --                     &                      & 5.0 / 5.0 / 5.0 & 5.5 / 5.5 / 5.5 & 5.3 / 5.3 / 5.3 \\ 
              &           &                      & $W$                  &                      & 5.2 / 5.3 / 5.7 & 3.3 / 4.9 / 5.1 & 3.6 / 5.1 / 4.9 \\ 
              &           &                      &                        & \multicolumn{1}{l}{} &                    & \multicolumn{1}{l}{} & \multicolumn{1}{l}{} \\
              & 3         & Yes                  & $X$                    &                      & 5.3 / 5.4 / 5.4 & 5.1 / 5.1 / 5.1 & 4.9 / 4.9 / 4.9 \\ 
              &           &                      & $X \times W$         &                      & 4.7 / 4.8 / 5.0 & 2.5 / 5.3 / 5.4 & 2.4 / 5.5 / 5.5 \\ 
              &           & No                   & --                     &                      & 5.5 / 5.5 / 5.5 & 5.1 / 5.1 / 5.1 & 4.8 / 4.8 / 4.8 \\ 
              &           &                      & $W$                  &                      & 5.3 / 5.3 / 5.5 & 4.3 / 5.5 / 5.1 & 4.3 / 5.4 / 5.2 \\  
                &           &                      &                        & \multicolumn{1}{l}{} &                    & \multicolumn{1}{l}{} & \multicolumn{1}{l}{} \\  
$H_1$         & 1         & Yes                  & $X$                    &                      & 42.2 / 42.3 / 42.3 & 41.9 / 41.9 / 41.9 & 41.9 / 41.9 / 41.9 \\ 
              &           &                      & $X \times W$         &                      & 41.7 / 41.8 / 55.8 & 40.8 / 57.1 / 56.9 & 40.8 / 56.6 / 56.4 \\ 
              &           & No                   & --                     &                      & 41.9 / 42.0 / 42.0 & 41.9 / 41.9 / 41.9 & 42.2 / 42.2 / 42.2 \\ 
              &           &                      & $W$                  &                      & 42.3 / 42.4 / 56.6 & 42.2 / 49.2 / 56.6 & 41.4 / 48.4 / 56.9 \\   
              &           &                      &                        & \multicolumn{1}{l}{} &                    & \multicolumn{1}{l}{} & \multicolumn{1}{l}{} \\
              & 2         & Yes                  & $X$                    &                      & 43.5 / 43.4 / 43.4 & 43.9 / 43.9 / 43.9 & 43.8 / 43.8 / 43.8 \\ 
              &           &                      & $X \times W$         &                      & 44.2 / 44.3 / 63.9 & 42.9 / 64.2 / 64.2 & 43.0 / 63.8 / 64.0 \\ 
              &           & No                   & --                     &                      & 44.1 / 44.1 / 44.1 & 44.1 / 44.2 / 44.2 & 43.8 / 43.9 / 43.9 \\ 
              &           &                      & $W$                  &                      & 44.3 / 44.3 / 64.6 & 44.1 / 49.9 / 63.5 & 44.7 / 50.6 / 64.8 \\ 
              &           &                      &                        & \multicolumn{1}{l}{} &                    & \multicolumn{1}{l}{} & \multicolumn{1}{l}{} \\
              & 3         & Yes                  & $X$                    &                      & 37.7 / 37.8 / 37.8 & 37.6 / 37.6 / 37.6 & 37.4 / 37.4 / 37.4 \\ 
              &           &                      & $X \times W$         &                      & 37.5 / 37.6 / 47.1 & 35.9 / 46.9 / 46.7 & 36.0 / 46.6 / 46.7 \\ 
              &           & No                   & --                     &                      & 38.2 / 38.2 / 38.2 & 37.8 / 37.9 / 37.9 & 36.9 / 36.9 / 36.9 \\ 
              &           &                      & $W$                  &                      & 37.4 / 37.5 / 46.9 & 36.1 / 39.5 / 46.7 & 37.3 / 40.5 / 47.7 \\
              \hline
\end{tabular}}
\begin{tablenotes} 
\small
\item Notes: $H_0$ and $H_1$ denote the null and alternative conditions, respectively. $X$ stratified indicates whether $X$ is used for stratification. $S(\cdot)$ are the covariates used for stratification. $X \times W$ indicates that both $X$ and $W$ are used to form the strata. -- indicates that no covariate is used to form the strata. SR, simple randomization; SBR, stratified block randomization; SBCD, stratified biased-coin design. The three numbers in each column correspond to the rejection probabilities for the usual test, modified test, and stratified-adjusted test.
\end{tablenotes}
\label{tab:simulation_equal}
\end{table}

\begin{table}[]
\caption{Rejection probabilities (in percentage points) for three interaction tests (on binary $X$), stratification mechanisms, and randomization methods under unequal allocation, $\pi = 2/3$, and sample size $n = 800$}{
\tabcolsep=4.25pt
\renewcommand\arraystretch{0.5}
\small
\begin{tabular}{cccccccc}
\hline
              &           & \multicolumn{2}{c}{Stratification covariates} &                      & \multicolumn{3}{c}{Randomization methods}                    \\ \cline{3-4} \cline{6-8}
              
\multicolumn{1}{l}{} &\multicolumn{1}{c}{Model}  & $X$ stratified         & $S(\cdot)$             &                      & SR                & SBR                  & SBCD                  \\
\hline
$H_0$         & 1         & Yes                  & $X$                    &                      &  4.8 / 5.0 / 5.0 & 5.6 / 5.6 / 5.6 & 5.6 / 5.6 / 5.6 \\ 
              &           &                      & $X \times W$         &                      & 5.4 / 5.5 / 5.7 & 1.6 / 5.7 / 5.7 & 1.6 / 5.4 / 5.4 \\ 
              &           & No                   & --                     &                      & 5.2 / 5.2 / 5.2 & 4.9 / 4.9 / 4.9 & 5.4 / 5.4 / 5.4 \\ 
              &           &                      & $W$                  &                      & 5.4 / 5.4 / 5.5 & 4.7 / 5.5 / 5.8 & 4.4 / 5.3 / 5.2 \\ 
              &           &                      &                        & \multicolumn{1}{l}{} &                    & \multicolumn{1}{l}{} & \multicolumn{1}{l}{} \\
              & 2         & Yes                  & $X$                    &                      & 5.4 / 5.4 / 5.4 & 5.2 / 5.2 / 5.2 & 5.1 / 5.1 / 5.1 \\ 
              &           &                      & $X \times W$         &                      & 4.7 / 4.7 / 5.2 & 1.5 / 5.5 / 5.4 & 1.3 / 5.1 / 5.0 \\
              &           & No                   & --                     &                      & 5.0 / 5.1 / 5.1 & 5.2 / 5.2 / 5.2 & 4.6 / 4.7 / 4.7 \\ 
              &           &                      & $W$                  &                      & 5.5 / 5.6 / 5.2 & 4.1 / 4.9 / 5.0 & 4.3 / 5.0 / 5.4 \\ 
              &           &                      &                        & \multicolumn{1}{l}{} &                    & \multicolumn{1}{l}{} & \multicolumn{1}{l}{} \\
              & 3         & Yes                  & $X$                    &                      & 5.1 / 5.2 / 5.2 & 5.8 / 5.8 / 5.8 & 5.2 / 5.1 / 5.1 \\
              &           &                      & $X \times W$         &                      & 5.4 / 5.5 / 5.5 & 2.3 / 5.4 / 5.3 & 2.4 / 5.2 / 5.1 \\ 
              &           & No                   & --                     &                      & 5.3 / 5.3 / 5.3 & 4.6 / 4.6 / 4.6 & 5.2 / 5.1 / 5.1 \\ 
              &           &                      & $W$                  &                      & 5.3 / 5.3 / 5.6 & 4.7 / 5.2 / 5.3 & 4.9 / 5.4 / 5.2 \\ 
              &           &                      &                        & \multicolumn{1}{l}{} &                    & \multicolumn{1}{l}{} & \multicolumn{1}{l}{} \\
$H_1$         & 1         & Yes                  & $X$                    &                      & 37.9 / 38.0 / 38.0 & 38.2 / 38.2 / 38.2 & 38.4 / 38.2 / 38.2 \\ 
              &           &                      & $X \times W$         &                      & 37.7 / 37.8 / 54.3 & 36.1 / 53.8 / 53.8 & 36.0 / 53.6 / 53.9 \\ 
              &           & No                   & --                     &                      & 38.4 / 38.5 / 38.5 & 38.0 / 38.1 / 38.1 & 39.1 / 39.1 / 39.1 \\ 
              &           &                      & $W$                  &                      & 39.0 / 39.2 / 54.3 & 38.5 / 41.4 / 54.2 & 38.0 / 40.8 / 54.4 \\ 
              &           &                      &                        & \multicolumn{1}{l}{} &                    & \multicolumn{1}{l}{} & \multicolumn{1}{l}{} \\
              & 2         & Yes                  & $X$                    &                      &  47.4 / 47.5 / 47.5 & 48.9 / 49.0 / 49.0 & 48.3 / 48.3 / 48.3 \\ 
              &           &                      & $X \times W$         &                      & 48.1 / 48.2 / 69.1 & 47.5 / 68.5 / 68.8 & 47.2 / 68.8 / 68.8 \\ 
              &           & No                   & --                     &                      & 48.8 / 48.8 / 48.8 & 48.5 / 48.6 / 48.6 & 47.5 / 47.4 / 47.4 \\ 
              &           &                      & $W$                  &                      & 48.4 / 48.4 / 69.5 & 48.4 / 51.7 / 68.7 & 48.6 / 51.7 / 68.8 \\ 
              &           &                      &                        & \multicolumn{1}{l}{} &                    & \multicolumn{1}{l}{} & \multicolumn{1}{l}{} \\
              & 3         & Yes                  & $X$                    &                      & 33.5 / 33.6 / 33.6 & 33.7 / 33.7 / 33.7 & 34.1 / 34.0 / 34.0 \\ 
              &           &                      & $X \times W$         &                      & 33.2 / 33.3 / 43.2 & 31.2 / 43.6 / 43.5 & 31.6 / 43.6 / 43.5 \\ 
              &           & No                   & --                     &                      & 33.7 / 33.8 / 33.8 & 33.4 / 33.5 / 33.5 & 34.0 / 34.0 / 34.0 \\ 
              &           &                      & $W$                  &                      & 33.1 / 33.2 / 43.6 & 32.0 / 33.6 / 43.0 & 33.8 / 35.2 / 43.7 \\
              \hline
\end{tabular}}
\begin{tablenotes} 
\small
\item Notes: $H_0$ and $H_1$ denote the null and alternative conditions, respectively. $X$ stratified indicates whether $X$ is used for stratification. $S(\cdot)$ are the covariates used for stratification. $X \times W$ indicates that both $X$ and $W$ are used to form the strata. -- indicates that no covariate is used to form the strata. SR, simple randomization; SBR, stratified block randomization; SBCD, stratified biased-coin design. The three numbers in each column correspond to the rejection probabilities for the usual test, modified test, and stratified-adjusted test.
\end{tablenotes}
\label{tab:simulation_unequal}
\end{table}

\section{{Extension to categorical covariates}} \label{section:multi}
\subsection{Interaction term Wald test}
In Sections \ref{section:binary_used} and \ref{section:binary_add}, we assume that covariate $X_i$ is binary. As discussed in this section, we extend the previous results to a categorical $X_i$. Details of the simulations are presented in the Appendix. 

Suppose that $X_i$ is a $(K+1)$-level categorical variable with $\mathcal{X} = \{0,1,\ldots, K\}$. Then, we can formulate the null hypothesis of no interaction as
\begin{equation}
	\begin{aligned}
		H_0: \Econd{\Yt - \Yc}{X_i = k} = \Econd{\Yt - \Yc}{X_i = 0}, \text{ for } k = 1, 2, \ldots, K.
	\end{aligned}
	\label{null:multi}
\end{equation}
In this section, we directly study the general case in which Assumption \ref{assumption:stratify} is not implemented.	

Consider the linear model
\begin{equation}
	\begin{aligned}
		Y_i = \alpha + \beta A_i + \sum_{k=1}^{K}\gamma_k X_{ki} +  \sum_{k=1}^{K}\delta_k A_{i} X_{ki}  + \epsilon_i,
	\end{aligned}
	\label{eq:multi_model_1}
\end{equation}
where $X_{ki}$ is the $k \text{th} $ dummy variable for $X_i$. The OLS estimator of $\delta_k$ is $\hat{\delta}_k = \Yax{1}{k} - \Yax{0}{k} - \Yax{1}{0} + \Yax{0}{0}$ for $k = 1, 2, \ldots, K$. The heteroscedasticity-robust Wald statistic for testing $\delta = (\delta_1, \delta_2, \dots,   \delta_k)^{\mathrm{T}} =  0$ can be expressed as 
\begin{equation*}
	\begin{aligned}
		W_n^{\text{Wald}}(V^{(n)}) = n(	{R}\hat{{\tau}})^{\mathrm{T}}(	{R}\hat{\Sigma}_{hc}	{R}^{\mathrm{T}})^{-1}(	{R}\hat{{\tau}}),
	\end{aligned}
\end{equation*}
where 
\begin{equation}
		{R}= \begin{pmatrix}
			-1 & 1 & 0&\cdots &0\\
			-1 & 0& 1 & &0 \\
			\vdots & \vdots& &\ddots &\vdots \\
			-1 & 0& \cdots& 0  & 1 \\
		\end{pmatrix}_{K \times (K+1)}, \quad
        \hat{{\tau}} =  {\begin{pmatrix}
		\Yax{1}{0} - \Yax{0}{0}\\
		\Yax{1}{1} - \Yax{0}{1}\\
		\vdots \\
		\Yax{1}{K} - \Yax{0}{K}
		\end{pmatrix}_{(K+1) \times 1}},
	\label{eq:linear_form}
\end{equation}
and 
\begin{equation*}
	\begin{aligned}
		\hat{\Sigma}_{hc} =  \operatorname{diag}\left(\frac{n\hatsigmaSqax{1}{x}}{\nax{1}{x}} + \frac{n\hatsigmaSqax{0}{x}}{\nax{0}{x}}: x \in \mathcal{X}\right).
	\end{aligned}
\end{equation*}
The corresponding interaction term Wald test can be formulated as
\begin{equation}
	\begin{aligned}
		\phi_n^{ \text{Wald}}(V^{(n)}) = I\{W_n^{ \text{Wald}}(V^{(n)}) > \chi^2_{K,1-\alpha}\},
	\end{aligned}
	\label{test:heter_wald}
\end{equation}
where $\chi^2_{K,1-\alpha}$ is the $1-\alpha$ quantile of the $\chi^2$ random variable with $K$ degrees of freedom.

\begin{remark}
	Instead of using dummy variables, we can equivalently formulate linear model (\ref{eq:multi_model_1}) as 
	\begin{equation*}
		\begin{aligned}
			Y_i =  \sum_{k=0}^{K}\beta_k \ind{X_i = k} +  \sum_{k=0}^{K}\tau_k A_{i} \ind{X_i = k}  + \epsilon_i.
		\end{aligned}
	\end{equation*}
	The OLS estimator of ${\tau} =  (\tau_0, \tau_1, \dots, \tau_K)^{\mathrm{T}}$ is the same as that in (\ref{eq:linear_form}).
\end{remark}

\begin{theorem}
	Under Assumptions  \ref{assumption:nondegenerate}, \ref{assumption:cond ind}, and \ref{assumption:normal},
	\begin{equation}
		\begin{aligned}
			{\sqrt{n}} {(\hat{\tau} - \tau)}
			\overset{d}{\rightarrow}
			N\left({\begin{pmatrix}
					0 \\
					0 \\
					\vdots\\
					0
			\end{pmatrix}}, 
			{\begin{pmatrix}
					w^2_{0} & w_{10} & \cdots & w_{K0}\\
					w_{10} &  w^2_{1}& \cdots &w_{K1} \\
					\vdots & \vdots& \ddots &\vdots \\
					w_{K0} & w_{K1} & \cdots &w^2_{K}
			\end{pmatrix}}\right), 
		\end{aligned}
		\label{eq:wald_add_cov}
	\end{equation}
	where $w^2_{x} = (1/\px{x}^2)\{\zetaAddY{x}+ \zetaAddA{x}+ \zetaAddH{x}\}$ and  $w_{xy} =(\px{x}\px{y})^{-1} \{\zetaAddYxy{x}{y}+ \zetaAddAxy{x}{y} + \zetaAddHxy{x}{y}\} $, $x, y \in \mathcal{X}$.
	\label{thm:wald_add}
\end{theorem}

The subsequent theorem shows that the usual interaction term Wald test is conservative.

\begin{theorem}
	Under Assumptions  \ref{assumption:nondegenerate}, \ref{assumption:cond ind}, and \ref{assumption:normal},
	for testing  (\ref{null:multi}) at level $\alpha \in (0,1)$ when the null hypothesis is true,
	$$\lim_{n \rightarrow \infty}\E{\phi_n^{\text{Wald}}(V^{(n)})} \leq \alpha,$$ 
	and the inequality is strict unless 
	\begin{equation*}
		\begin{aligned}
			\{\pi(1-\pi)-q(s)\}&\Big(\frac{1}{\px{x}}\Big[\frac{1}{\pi} \Econd{\raxZ{1}{x}}{S_i=s} + \frac{1}{1-\pi} \Econd{\raxZ{0}{x}}{S_i=s}\Big]\\ 
            &\quad - \frac{1}{\px{0}}\Big[\frac{1}{\pi} \Econd{\raxZ{1}{0}}{S_i=s} + \frac{1}{1-\pi} \Econd{\raxZ{0}{0}}{S_i=s}\Big]\Big) = 0
		\end{aligned}
	\end{equation*}
	for all $x \in \mathcal{X}$, $s \in \mathcal{S}$.
	\label{thm:chi_test}
\end{theorem}
We can modify the usual interaction term Wald test and obtain an valid test based on Theorem \ref{thm:wald_add}. Let 
\begin{equation*}
	\begin{aligned}
		\hat{\Sigma}_{mod} &= \begin{pmatrix}
			\hat w^2_{0} & \hat w_{10} & \cdots & \hat w_{K0}\\
			\hat w_{10}& \hat w^2_{1}& \cdots & \hat w_{K1} \\
			\vdots & \vdots& \ddots &\vdots \\
			\hat w_{K0} & \hat w_{K1} & \cdots & \hat w^2_{K}
		\end{pmatrix}\\
	\end{aligned}
\end{equation*}
where $\hat w^2_{x} = (1/\hatpx{x}^2)\{\hatzetaAddY{x}+ \hatzetaAddA{x}+ \hatzetaAddH{x}\}$ and  $\hat w_{xy} =(\hatpx{x}\hatpx{y})^{-1} \{\hatzetaAddYxy{x}{y}+ \hatzetaAddAxy{x}{y} + \hatzetaAddHxy{x}{y}\} $, $x, y \in \mathcal{X}$. The modified interaction term Wald statistic can be written as 
\begin{equation}
	\begin{aligned}
		W_n^{\text{Wald,mod}}(V^{(n)}) = n(	{R}\hat{\tau})^{\mathrm{T}}({R}\hat{\Sigma}_{mod}	{R}^{\mathrm{T}})^{-1}(	{R}\hat{\tau}),
	\end{aligned}
	\label{eq:wald_statistic_mod}
\end{equation}
where $R$ is defined as in (\ref{eq:linear_form}). The modified interaction term Wald test can be formulated as
\begin{equation}
	\begin{aligned}
		\phi_n^{\text{Wald,mod}}(V^{(n)}) = I\{W_n^{ \text{Wald,mod}}(V^{(n)}) > \chi^2_{K,1-\alpha}\}.
	\end{aligned}
	\label{eq:wald_test_modified}
\end{equation}

\begin{theorem}
	\label{thm:chi_test_mod}
	Under Assumptions  \ref{assumption:nondegenerate}, \ref{assumption:cond ind}, and \ref{assumption:normal}, for testing (\ref{null:multi}) at level $\alpha \in (0,1)$ when the null hypothesis is true 
	$$\lim_{n \rightarrow \infty}\E{\phi_n^{\text{Wald,mod}}(V^{(n)})} = \alpha.$$ 
\end{theorem}

Although the modified interaction term Wald test is valid, it has two apparent drawbacks. First, to estimate the covariance matrix in (\ref{eq:wald_add_cov}), it is necessary to estimate $(K^2 + K)/{2}$ pairwise covariances. Second, the calculation of the modified Wald statistic (\ref{eq:wald_statistic_mod}) involves inverting a nonsparse matrix ${R}\hat{\Sigma}_{mod}{R}^{\mathrm{T}}$. These calculations are cumbersome and unstable, particularly in scenarios in which $K$ is large.

These drawbacks are attributable to the complex asymptotic covariance structure of $\hat{\tau}$. To address these issues, we can use the same technique as that used in Section \ref{section:binary_add}. We construct the stratified-adjusted estimator $\hat{{\tau}}_{strat}$ that possesses a simpler asymptotic covariance structure. Consequently, we can formulate a simpler test based on $\hat{{\tau}}_{strat}$. Furthermore, this technique allows us to replace the more stringent Assumption \ref{assumption:normal} with the more relaxed Assumption \ref{assumption:normal_weaker}.

\subsection{Stratified-adjusted interaction iest}
The stratified-adjusted estimator of ${\tau}$ can be expressed as
\begin{equation}
	\begin{aligned}
		\hat{{\tau}}_{strat} =  {\begin{pmatrix}
				\sum_{s \in \mathcal{S}} \{{\nxs{0}{s}}/{\nx{0}}\}\{\hatmuaxs{1}{0}{s} - \hatmuaxs{0}{0}{s}\}\\
				\sum_{s \in \mathcal{S}} \{{\nxs{1}{s}}/{\nx{1}}\}\{\hatmuaxs{1}{1}{s} - \hatmuaxs{0}{1}{s}\}\\
				\vdots \\
				\sum_{s \in \mathcal{S}} \{{\nxs{K}{s}}/{\nx{K}}\}\{\hatmuaxs{1}{K}{s} - \hatmuaxs{0}{K}{s}\}
		\end{pmatrix}},
	\end{aligned}
\end{equation}
and we can derive the following results: 

\begin{theorem}
	Under Assumptions  \ref{assumption:nondegenerate}, \ref{assumption:cond ind}, and \ref{assumption:normal_weaker},
	\begin{equation*}
		\begin{aligned}
			{\sqrt{n}} {(\hat{\tau}_{strat} - \tau)}
			\overset{d}{\rightarrow}
			N\left({\begin{pmatrix}
					0 \\
					0 \\
					\vdots\\
					0
			\end{pmatrix}}, {\begin{pmatrix}
					s^2_{0} & 0 & \cdots &0\\
					0 &s^2_{1}& \cdots &0 \\
					\vdots & \vdots& \ddots &\vdots \\
					0 & 0 & \cdots &s^2_{K}
			\end{pmatrix}}\right), 
		\end{aligned}
	\end{equation*}
	where $s^2_{x} = ({1}/{\px{x}^2})\{\zetaCheckY{x}  + \zetaCheckH{x}\}$, $x \in \mathcal{X}$.
	\label{thm:wald_add_strata}
\end{theorem}

Let
\begin{equation*}
         \hat{\Sigma}_{strat} =\operatorname{diag}\left(\hat s^2_x: x \in \mathcal{X}\right),
\end{equation*}
where $\hat s^2_{x} = ({1}/{\hatpx{x}^2})\{\hatzetaCheckY{x}  + \hatzetaCheckH{x}\}$, $x \in \mathcal{X}$.
The stratified-adjusted Wald statistic can be written as 
\begin{equation*}
	\begin{aligned}
		W_n^{\text{strat}}(V^{(n)}) = n(	{R}\hat{\tau}_{strat})^{\mathrm{T}}({R}\hat{\Sigma}_{strat}	{R}^{\mathrm{T}})^{-1}(	{R}\hat{\tau}_{strat}),
	\end{aligned}
\end{equation*}
and the stratified-adjusted Wald test can be formulated as
\begin{equation}
	\begin{aligned}
		\phi_n^{\text{Wald,strat}}(V^{(n)}) = I\{W_n^{ \text{strat}}(V^{(n)}) > \chi^2_{K,1-\alpha}\},
	\end{aligned}
	\label{eq:wald_test_strata}
\end{equation}

\begin{theorem}
	\label{thm:chi_test_strata}
	Under Assumptions  \ref{assumption:nondegenerate}, \ref{assumption:cond ind}, and \ref{assumption:normal_weaker}, for testing (\ref{null:multi}) at level $\alpha \in (0,1)$ when the null hypothesis is true,
	$$\lim_{n \rightarrow \infty}\E{\phi_n^{\text{Wald,strat}}(V^{(n)})} = \alpha.$$ 
\end{theorem}

Unlike that of $\hat{\tau}$, the asymptotic covariance matrix of $\hat{\tau}_{strat}$ is diagonal. There are two advantages to this configuration. First, $(K^2 + K)/{2}$ pairwise covariances are known to be zero, with only $K+1$ variances left to be estimated. Moreover, $R\hat{\Sigma}_{strat}R^{\mathrm{T}} = \text{diag}(\hat s^2_{1}, \hat s^2_{2}, \ldots, \hat s^2_{K}) + \hat s^2_{0}1_K 1_K^{\mathrm{T}}$, and thus its inverse can be easily calculated using the Sherman--Morrison--Woodbury formula.

\begin{remark}
	Let $\Sigma$ and $\Sigma_{strat}$ denote the asymptotic covariance matrices in Theorems \ref{thm:wald_add} and \ref{thm:wald_add_strata}, respectively. Lemma 8 (Appendix) shows that $\Sigma \succeq \Sigma_{strat}$ ($\Sigma - \Sigma_{strat}$ is positive semi-definite). Then, under the alternative hypothesis, the stratified-adjusted test (\ref{eq:wald_test_strata}) is asymptotically more powerful than the modified  test (\ref{eq:wald_test_modified}).
\end{remark}

\section{{Discussion}} \label{section:discussion}
In this study, we investigated three types of interaction tests under covariate-adaptive randomization. We showed that the usual interaction test is conservative in the sense that its limiting rejection probability can be strictly less than the nominal level. We proposed the modified interaction test and the stratified-adjusted interaction test and proved that both of them are valid. 
Our theoretical framework is flexible enough to accommodate both stratification covariates and additional baseline covariates, which provides practical convenience as the same algorithmic procedure can be used. As practical guidance, we recommend using the stratified-adjusted interaction test because it is simple, more powerful, and applicable to various randomization procedures.


Our findings have broader implications for subgroup analysis, beyond specific tests for interaction. First, our results can be directly applied to subgroup-specific tests. Subgroup analyses that involve factors not stratified at randomization may raise concerns, as there may still be imbalances in prognostic factors and the number of patients between treatment groups within the subgroup \citep{cui2002, grouin2005,sun2011}. 
Stratification during randomization is generally regarded as a means to enhance the reliability of subgroup findings \citep{tanniou2016}. However, \cite{kaiser2013} questioned the necessity of stratified randomization based on subgroup membership. The author demonstrated that, on average, treatment allocation and covariates are balanced within each subgroup, and the variability in covariate imbalance between treatment groups increases only slightly when using a randomization method that is not stratified on the subgroup, compared with a method that is stratified. Our results address this concern, particularly in scenarios with large sample sizes, by showing that the subgroup factor does not need to be used for stratification to obtain valid tests for interaction or subgroup-specific tests. Furthermore, the results presented in Section \ref{section:multi} can be readily extended to test general linear hypotheses, rather than solely focusing on the null hypothesis of no interaction ($R\tau = 0$).
Our work may also provide insights for addressing other challenges in subgroup analysis, such as statistical power limitations and adjustments for multiple comparisons.

This paper focus on testing treatment-covariate interaction of discrete covariates. We leave the use of additional variables to further improve efficiency for future work. Moreover, for continuous covariates, though it is common practice to categorize them in the clinical trial, analyzing continuous covariates on their original scale is sometimes preferable \citep{Royston2013}. It would be of interest to establish valid and efficient model-free interaction analysis of continuous covariates under covariate-adaptive randomization, or more generally, covariate-adjusted response-adaptive (CARA) randomization \citep{zhu2013}. Rather than focusing on the average treatment effect, \cite{zhang2020} explored the estimation and bootstrap inference for the quantile treatment effect, which may offer an alternative approach for studying interaction effects.

\section*{Acknowledgments}
This work was supported by the National Natural Science Foundation of China (12171476).

\bibliographystyle{apalike}
\bibliography{Inter}

\newpage
\appendix
\titleformat{\section}{\normalfont\Large\bfseries}{Appendix \thesection}{1em}{}

\section{Notations used throughout the main text and appendices}
This section provides a comprehensive list of notations used in the main text and those introduced in the appendices, which are utilized in the following proofs.

\renewcommand{\arraystretch}{1.6} 
\begin{longtable}{ @{} l p{14cm} @{} }
		$\nx{x}$ & Number of units with $X_i = x$ \\
		$\nax{a}{x}$ & Number of units with $A_i = a$ and $X_i = x$ \\
		$\ns{s}$ & Number of units in stratum $s \in \mathcal{S}$ \\
		$\nas{a}{s}$ &  Number of units with $A_i = a$ in stratum $s \in \mathcal{S}$ \\
		$\nxs{x}{s}$ & Number of units with $X_i = x$ in stratum $s \in \mathcal{S}$ \\
		$\naxs{a}{x}{s}$ & Number of units with $A_i = a$ and $X_i = x$ in stratum $s \in \mathcal{S}$ \\
		$\px{x}$ & $\Prob{X_i = x}$ \\
		$\ps{s}$ & $\Prob{S_i = s}$ \\
		$\pxs{x}{s}$ & $\Prob{S_i = s, X_i = x}$ \\
		$\muax{a}{x}$ & $\Econd{\Ya}{X_i = x}$ \\
		$\tau_x$ & $\Econd{\Yt - \Yc}{X_i = x}$ \\
		$\sigmaSqax{a}{x}$ & $\Varcond{\Ya}{X_i = x}$ \\
		$\hatpx{x}$ & $\frac{\nx{x}}{n}$ \\
		 $\Yax{a}{x}$ &   $ \frac{1}{\nax{a}{x}}\sum_{i=1}^{n}Y_i\ind{A_i = a, X_i=x}$ \\
		 $\hatsigmaSqax{a}{x}$ &$ \frac{1}{\nax{a}{x}}\sum_{i=1}^n(Y_i - \Yax{a}{x})^2\ind{A_i = a, X_i=x}$\\
		 $\hatmuas{a}{s}$ & $\frac{1}{\nas{a}{s}}\sum_{i=1}^nY_i\ind{A_i=a, S_i=s}$ \\
		  $\hatmuaxs{a}{x}{s}$ & $\frac{1}{\nax{a}{x}(s)}\sum_{i=1}^{n}Y_i\ind{A_i=a, S_i=s,X_i=x}$ \\
		 $\sigma^2_{\tilde{Y}(a)\mid X = x}$ & $\Varcond{\tildeYa{a}}{X_i=x}$\\
		 $\sigma^2_{\tilde{r}_{x}(a)}$ & $\Var{\tilderax{a}{x}}$\\
		 $\sigma^2_{\check{Y}(a)|X = x}$ & $\Varcond{\checkYa{a}}{X_i=x}$\\
		 $ \zetaStrY{x}
		 $ & $\px{x}\big\{\frac{1}{\pi}\sigma^2_{\tilde{Y}(1)|X = x} + \frac{1}{1-\pi}\sigma^2_{\tilde{Y}(0)|X = x} \big\} $\\
		 $\zetaStrA{x} $
		 &  $\sum_{s,X(s)=x}p(s)q(s) \big[\frac{1}{\pi} \Econd{\maxZ{1}{x}}{S_i = s} + \frac{1}{1-\pi} \Econd{\maxZ{0}{x}}{S_i = s}\big]^2$\\
		 $\zetaStrH{x}$ &
		 $\sum_{s,X(s)=x}p(s) \big[\Econd{\maxZ{1}{x}}{S_i = s} -  \Econd{\maxZ{0}{x}}{S_i = s}\big]^2$\\
		 $\hatzetaStrY{x}$ &
		 $\hatpx{x}\big[\frac{1}{\pi}\big\{\frac{1}{\nax{1}{x}}\sum_{i=1}^n Y_i^2A_i \ind{X_i = x} -\sum_{s,X(s)=x}\frac{\ns{s}}{\nx{x}}\hat{\mu}^2_{n,1}( s)\big\}+ \frac{1}{1-\pi}\big\{\frac{1}{\nax{0}{x}}\sum_{i=1}^nY_i^2(1-A_i)\ind{X_i = x} -\sum_{X(s)=x}\frac{\ns{s}}{\nx{x}}\hat{\mu}^2_{n,0}(s)\big\}\big] $\\
		 $\hatzetaStrA{x}$ &
		 $\sum_{s,X(s) = x}\frac{\ns{s}}{n}q(s)\big[\frac{1}{\pi}\big\{\hatmuas{1}{s} - \Yax{1}{x}\big\} + \frac{1}{1-\pi}\big\{\hatmuas{0}{s} - \Yax{0}{x}\big\}\big]^2 $\\
		 $\hatzetaStrH{x}$ &
		 $\sum_{s,X(s) = x}\frac{\ns{s}}{n} \big[\big\{\hatmuas{1}{s} - \Yax{1}{x}\big\} - \big\{\hatmuas{0}{s} - \Yax{0}{x}\big\}\big]^2$\\
		 $\zetaAddY{x}$ & $\frac{1}{\pi}\sigma^2_{\tilde{r}_{x}(1)} + \frac{1}{1-\pi}\sigma^2_{\tilde{r}_{x}(0)} $\\
		 $\zetaAddA{x}$ &
		 $\sum_{s \in \mathcal{S}} p(s)q(s)\big[\frac{1}{\pi} \Econd{\raxZ{1}{x}}{S_i = s} + \frac{1}{1-\pi} \Econd{\raxZ{0}{x}}{S_i = s}\big]^2$\\
		 $ \zetaAddH{x}$ &
		 $\sum_{s \in \mathcal{S}}p(s) \big[ \Econd{\raxZ{1}{x}}{S_i = s} -  \Econd{\raxZ{0}{x}}{S_i = s}\big]^2$ \\
		 $\zetaAddYxy{x}{y}$ &
		$ -\sum_{s \in \mathcal{S}}  p(s)\big[\frac{1}{\pi}\Econd{\raxZ{1}{x}}{S_i = s}\Econd{\raxZ{1}{y}}{S_i = s} + \frac{1}{1-\pi} \Econd{\raxZ{0}{x}}{S_i = s}\Econd{\raxZ{0}{y}}{S_i = s}\big] $\\
		$\zetaAddAxy{x}{y}$ &
		$\sum_{s \in \mathcal{S}}p(s)q(s) \big[\frac{1}{\pi} \Econd{\raxZ{1}{x}}{S_i = s} + \frac{1}{1-\pi} \Econd{\raxZ{0}{x}}{S_i = s}\big]\big[\frac{1}{\pi} \Econd{\raxZ{1}{y}}{S_i = s} + \frac{1}{1-\pi} \Econd{\raxZ{0}{y}}{S_i = s}\big]$ \\
		$\zetaAddHxy{x}{y}$ &
		$\sum_{s \in \mathcal{S}}p(s)\big[\Econd{\raxZ{1}{x}}{S_i = s}- \Econd{\raxZ{0}{x}}{S_i = s}\big]\big[ \Econd{\raxZ{1}{y}}{S_i = s}- \Econd{\raxZ{0}{y}}{S_i = s}\big]$ \\
		$\hatzetaAddY{x} $ &
		$\frac{1}{\pi}\big[\hat{p}_x\frac{1}{\nax{1}{x}}\sum_{i=1}^n(Y_i- \Yax{1}{x})^2A_i \ind{X_i=x} -\sum_{s \in \mathcal{S}}\frac{\ns{s}}{n}\big\{\frac{\nxs{x}{s}}{\ns{s}}\big\}^2\big\{\hatmuaxs{1}{x}{s} - \Yax{1}{x}\big\}^2\big] + \frac{1}{1-\pi}\big[\hat{p}_x\frac{1}{\nax{0}{x}}\sum_{i=1}^n(Y_i- \Yax{0}{x})^2(1-A_i)\ind{X_i=x} -\sum_{s \in \mathcal{S}}\frac{\ns{s}}{n}\big\{\frac{\nxs{x}{s}}{\ns{s}}\big\}^2\big\{\hatmuaxs{0}{x}{s} - \Yax{0}{x}\big\}^2\big]$\\
		$\hatzetaAddA{x}$ & $ \sum_{s \in \mathcal{S}}\frac{\ns{s}}{n}q(s)\big\{\frac{\nxs{x}{s}}{\ns{s}}\big\}^2\big[\frac{1}{\pi}\big\{\hatmuaxs{1}{x}{s} - \Yax{1}{x}\big\} + \frac{1}{1-\pi}\big\{\hatmuaxs{0}{x}{s} - \Yax{0}{x}\big\}\big]^2 $\\
		$\hatzetaAddH{x} $ &
		$\sum_{s \in \mathcal{S}}\frac{\ns{s}}{n}\big\{\frac{\nxs{x}{s}}{\ns{s}}\big\}^2 \big[\big\{\hatmuaxs{1}{x}{s} - \Yax{1}{x}\big\} - \big\{\hatmuaxs{0}{x}{s} - \Yax{0}{x}\big\}\big]^2$ \\
		$\hatzetaAddYxy{x}{y}$ &
		$-\sum_{s \in \mathcal{S}}  \frac{\ns{s}}{n}\frac{\nxs{x}{s}}{\ns{s}}\frac{\nxs{y}{s}}{\ns{s}}\big[\frac{1}{\pi}\big\{\hatmuaxs{1}{x}{s} - \Yax{1}{x}\big\}\big\{\hatmuaxs{1}{y}{s} - \Yax{1}{y}\big\} + \frac{1}{1-\pi} \big\{\hatmuaxs{0}{x}{s} - \Yax{0}{x}\big\}\big\{\hatmuaxs{0}{y}{s} - \Yax{0}{y}\big\}\big] $\\
		$\hatzetaAddAxy{x}{y} $ &
		$\sum_{s \in \mathcal{S}}\frac{\ns{s}}{n}q(s)\frac{\nxs{x}{s}}{\ns{s}}\frac{\nxs{y}{s}}{\ns{s}} \big[\frac{1}{\pi} \big\{\hatmuaxs{1}{x}{s} - \Yax{1}{x}\big\} + \frac{1}{1-\pi} \big\{\hatmuaxs{0}{x}{s} - \Yax{0}{x}\big\}\big]\big[\frac{1}{\pi} \big\{\hatmuaxs{1}{y}{s} - \Yax{1}{y}\big\} + \frac{1}{1-\pi} \big\{\hatmuaxs{0}{y}{s} - \Yax{0}{y}\big\}\big]$\\
		$\hatzetaAddHxy{x}{y} $ &
		$ \sum_{s \in \mathcal{S}}\frac{\ns{s}}{n}\frac{\nxs{x}{s}}{\ns{s}}\frac{\nxs{y}{s}}{\ns{s}}\big[\big\{\hatmuaxs{1}{x}{s} - \Yax{1}{x}\big\}- \big\{\hatmuaxs{0}{x}{s} - \Yax{0}{x}\big\}\big]\big[ \big\{\hatmuaxs{1}{y}{s} - \Yax{1}{y}\big\}- \big\{\hatmuaxs{0}{y}{s} - \Yax{0}{y}\big\}\big]$ \\
		$\zetaCheckY{x} $ & $\px{x}\big\{\frac{1}{\pi}\sigma^2_{\check {Y}(1)|X_i = x} +  \frac{1}{1-\pi}\sigma^2_{\check {Y}(0)|X_i = x}\big\}$\\
		$\zetaCheckH{x}$ &  $\sum_{s \in \mathcal{S}}\pxs{x}{s} \big[\Econd{\maxZ{1}{x}}{S_i = s, X_i=x} - \Econd{\maxZ{0}{x}}{S_i = s, X_i=x}\big]^2$\\
		$\hatzetaCheckY{x} $ &	$\frac{1}{\pi}\big[\hat{p}_x\frac{1}{\nax{1}{x}}\sum_{i=1}^n(Y_i- \Yax{1}{x})^2A_i \ind{X_i=x} -\sum_{s \in \mathcal{S}}\frac{\nxs{x}{s}}{n}\big\{\hatmuaxs{1}{x}{s} - \Yax{1}{x}\big\}^2\big] + \frac{1}{1-\pi}\big[\hat{p}_x\frac{1}{\nax{0}{x}}\sum_{i=1}^n(Y_i- \Yax{0}{x})^2(1-A_i)\ind{X_i=x} -\sum_{s \in \mathcal{S}}\frac{\nxs{x}{s}}{n}\big\{\hatmuaxs{0}{x}{s} - \Yax{0}{x}\big\}^2\big] $ \\
		$\hatzetaCheckH{x} $ & $\sum_{s \in \mathcal{S}}\frac{\nxs{x}{s}}{n} \big[\big\{\hatmuaxs{1}{x}{s} - \Yax{1}{x}\big\} - \big\{\hatmuaxs{0}{x}{s} - \Yax{0}{x}\big\}\big]^2$\\
\end{longtable}

\section{Proof of the main results}
        \subsection{Proof of Theorem 1}
\begin{proof}
     We first note that 
	\begin{equation*}
		\begin{aligned}
			&\quad \sqrt{n} (\Yax{1}{1} - \Yax{0}{1} - \tau_1 ) \\
			&= \sqrt{n}\Big[\frac{1}{\nax{1}{1}}\sum_{i=1}^{n}\{\Yt - \muax{1}{1}\}A_i\ind{X_i = 1} - \frac{1}{\nax{0}{1}}\sum_{i=1}^{n}\{\Yc - \muax{0}{1}\}(1-A_i)\ind{X_i = 1}\Big] \\
			&= \frac{1}{\sqrt{n}}\Big[\Big(\frac{\Dax{1}{1}}{n} + \pi \px{1}\Big)^{-1}\sum_{i=1}^{n}\{\Yt - \muax{1}{1}\}A_i\ind{X_i = 1} \\
            &\quad - \Big(\frac{\Dax{0}{1}}{n} + (1-\pi) \px{1}\Big)^{-1}\sum_{i=1}^{n}\{\Yc - \muax{0}{1}\}(1-A_i)\ind{X_i = 1}\Big] \\
			&= \frac{1}{\sqrt{n}}\Rx{1}\Big[\sum_{i=1}^{n}\Big\{\frac{\Dax{0}{1}}{n} + (1-\pi) \px{1}\Big\}\{\Yt - \muax{1}{1}\}A_i\ind{X_i = 1} \\
            &\quad - \sum_{i=1}^{n}\Big\{\frac{\Dax{1}{1}}{n} + \pi \px{1}\Big\}\{\Yc - \muax{0}{1}\}(1-A_i)\ind{X_i = 1}\Big] \\
			&= \frac{1}{\sqrt{n}}\Rx{1}\Big[\sum_{i=1}^{n}\frac{\Dax{0}{1}}{n}\{\Yt - \muax{1}{1}\}A_i\ind{X_i = 1} - \sum_{i=1}^{n}\frac{\Dax{1}{1}}{n} \{\Yc - \muax{0}{1}\}(1-A_i)\ind{X_i = 1}\Big] \\
			&\quad + \frac{1}{\sqrt{n}}\Rx{1}\Big[\sum_{i=1}^{n}(1-\pi)\px{1}\{\Yt - \muax{1}{1}\}A_i\ind{X_i = 1} - \sum_{i=1}^{n}\pi \px{1}\{\Yc - \muax{0}{1}\}(1-A_i)\ind{X_i = 1}\Big] \\
			&= \frac{1}{\sqrt{n}}\Rx{1} (\Zx{1} + \Mx{1}),
		\end{aligned}
	\end{equation*}
	and similarly, 
	\begin{equation*}
		\begin{aligned}
			\sqrt{n} (\Yax{1}{0} - \Yax{0}{0} - \tau_0 ) = \frac{1}{\sqrt{n}}\Rx{0} (\Zx{0} + \Mx{0}),
		\end{aligned}
	\end{equation*}
	where 
\begin{equation}
	\begin{aligned}
		\Dax{1}{1} &= \nax{1}{1} - n\pi \px{1},  \\
		\Dax{0}{1} &=  \nax{0}{1} - n(1-\pi) \px{1}, \\
		\Dax{1}{0} &= \nax{1}{0} - n\pi \px{0}, \\
		\Dax{0}{0} &=  \nax{0}{0} - n(1-\pi) \px{0}, 
	\end{aligned} 
\label{Dax}
\end{equation}
	and
	\begin{equation}
	\begin{aligned}
		\Rx{1} &= \Big(\frac{\Dax{1}{1}}{n} + \pi \px{1}\Big)^{-1}\Big\{\frac{\Dax{0}{1}}{n} + (1-\pi) \px{1}\Big\}^{-1}, \\
		\Rx{0} &= \Big(\frac{\Dax{1}{0}}{n} + \pi \px{0}\Big)^{-1}\Big\{\frac{\Dax{0}{0}}{n} + (1-\pi) \px{0}\Big\}^{-1}, \\
		\Zx{1} &= \sum_{i=1}^n \Big[\frac{\Dax{0}{1}}{n}\{\Yt - \muax{1}{1}\}A_i\ind{X_i = 1} - \frac{\Dax{1}{1}}{n} \{\Yc - \muax{0}{1}\}(1-A_i)\ind{X_i = 1}\Big], \\
		\Zx{0} &= \sum_{i=1}^{n}\Big[\frac{\Dax{0}{0}}{n} \{\Yt - \muax{1}{0}\}A_i\ind{X_i = 0} - \frac{\Dax{1}{0}}{n} \{\Yc - \muax{0}{0}\}(1-A_i)\ind{X_i = 0}\Big], \\
		\Mx{1} &= \px{1}\sum_{i=1}^{n}\Big[(1-\pi)\{\Yt - \muax{1}{1}\}A_i\ind{X_i = 1} - \pi \{\Yc - \muax{0}{1}\}(1-A_i)\ind{X_i = 1}\Big], \\
		\Mx{0} &= \px{0}\sum_{i=1}^{n}\Big[(1-\pi) \{\Yt - \muax{1}{0}\}A_i\ind{X_i = 0} - \pi \{\Yc - \muax{0}{0}\}(1-A_i)\ind{X_i = 0}\Big]. \\
	\end{aligned}
\label{RZMx}
\end{equation}
	
	By Lemma \ref{lemma:weak_law}, $\Rx{1} \inp  \{{\pi(1-\pi)\px{1}^2}\}^{-1}$, $\Rx{0} \inp \{{\pi(1-\pi)\px{0}^2}\}^{-1}$, and $\Dax{a}{x} = o_P(n)$ for  $a \in \{1,0\}$, $x \in \{1,0\}$. By Lemma \ref{lemma:thm_1_M}, 
	\begin{equation}
		\begin{aligned}
			\frac{1}{\sqrt{n}} {\begin{pmatrix}
					\Mx{1} \\
					\Mx{0} \\
			\end{pmatrix}}
			\indist
			\pi(1-\pi)
			N\left({\begin{pmatrix}
					0 \\
					0 \\
			\end{pmatrix}}, {\begin{pmatrix}
					\px{1}^4 v^2_{1} & 0 \\
					0 & \px{0}^4 v^2_{0} \\
			\end{pmatrix}} \right),
		\end{aligned}
	\end{equation}
 	where $v^2_{0} = (1/p_0^2)\{\zetaStrY{0} + \zetaStrA{0} + \zetaStrH{0}\}$, and $v^2_{1} = (1/p_1^2)\{\zetaStrY{1} + \zetaStrA{1} + \zetaStrH{1}\}$.
	We also have ${\Zx{1}} = o_P(\sqrt{n})$, ${\Zx{0}} = o_P(\sqrt{n})$ by Lemma \ref{lemma:thm_1_M} and Slutsky's theorem; then the desired result follows.
	
\end{proof}

\subsection{Proof of Theorem 2
}
\begin{proof}
	We first prove that
	\begin{equation}
		\begin{aligned}
			n \left(\frac{\hatsigmaSqax{1}{1}}{\nax{1}{1}} + \frac{\hatsigmaSqax{0}{1}}{\nax{0}{1}} + \frac{\hatsigmaSqax{1}{0}}{\nax{1}{0}} + \frac{\hatsigmaSqax{0}{0}}{\nax{0}{0}} \right) \overset{P}\rightarrow \zetatold,
		\end{aligned}
		\label{eq:var_t_old}
	\end{equation}
	where $\zetatold =  \sigmaSqax{1}{1} / ({\px{1}\pi})+  \sigmaSqax{0}{1} / \{\px{1}(1-\pi)\} +  \sigmaSqax{1}{0} / ({\px{0}\pi}) + \sigmaSqax{0}{0}/ \{\px{0}(1-\pi)\}$.
	
	Notice that
	\begin{equation}
		\begin{aligned}
			\Yax{1}{1} = \frac{1}{\nax{1}{1}}\sum_{i=1}^{n} Y_i A_i \ind{X_i = 1} = \muax{1}{1} + \frac{1}{\nax{1}{1}}\sum_{i=1}^{n} (Y_i - \muax{1}{1})A_i\ind{X_i = 1}.
		\end{aligned}
	\end{equation}
	Then

	\begin{equation*}
		\begin{aligned}
			\frac{n\hatsigmaSqax{1}{1}}{\nax{1}{1}} &= \frac{n}{\nax{1}{1}} \frac{1}{\nax{1}{1}}\sum_{i=1}^n(Y_i - \Yax{1}{1})^2A_i\ind{X_i = 1} \\
			&= \frac{n}{\nax{1}{1}} \frac{1}{\nax{1}{1}}\sum_{i=1}^n\{Y_i(1) -\muax{1}{1} - \Yax{1}{1}+\muax{1}{1}\}^2A_i\ind{X_i = 1} \\
			&= \frac{n}{\nax{1}{1}}\Big\{\frac{1}{\nax{1}{1}}\sum_{i=1}^n(\Yt -\muax{1}{1})^2A_i\ind{X_i = 1} - (\Yax{1}{1} -\muax{1}{1})^2\Big\} \\
			&= \Big(\frac{n}{\nax{1}{1}}\Big)^2 \frac{1}{n}\sum_{i=1}^n\{\Yt -\muax{1}{1}\}^2A_i\ind{X_i = 1} - \Big(\frac{n}{\nax{1}{1}}\Big)^3 \Big[\frac{1}{n}\sum_{i=1}^{n} \{\Yt - \muax{1}{1}\}A_i\ind{X_i = 1}\Big]^2\\
			&\inp \frac{1}{\px{1}\pi} \sigmaSqax{1}{1},
		\end{aligned}
	\end{equation*}
where the convergence follows from Lemma \ref{lemma:weak_law} and the continuous mapping theorem. Following the above steps and we can get (\ref{eq:var_t_old}).
	
	We next define $\zetat = v^2_{1} + v^2_{0}$, where $v^2_{0}$ and $v^2_{1}$ are defined as in Theorem 1. We next compare $\zetat$ with $\zetatold$.
 
	Notice that 
	\begin{equation}
		\begin{aligned}
			\sigma^2_{\tilde{Y}(a)\mid X = x}
			&= \Econd{Y^2_i(a)}{X_i=x} - 2E[\Ya \Econd{\Ya}{S_i}\mid X_i=x] \\
            &\qquad + E[\Econd{\Ya}{S_i}^2\mid X_i=x]\\
			&= \sigmaSqax{a}{x} - E[\Econd{\Ya}{S_i}^2 \mid X_i=x] + \mu_{ax}^2\\
			&= \sigmaSqax{a}{x} - \operatorname{var}[\Econd{\Ya}{S_i}\mid X_i=x]\\
			&= \sigmaSqax{a}{x} - E[\Econd{\maxZ{a}{x}}{S_i}^2\mid X_i = x]\\
			&= \sigmaSqax{a}{x}- \frac{1}{\px{x}}\sum_{s,X(s)=x}p(s)\Econd{\maxZ{a}{x}}{S_i=s}^2,
		\end{aligned}
	\label{eq:sigmaSqtildeY}
	\end{equation}
	for $a \in \{1, 0\}$ and $x \in \{1, 0\}$. So we have
  {\allowdisplaybreaks
		\begin{align*}
			&\quad \zetatold - \zetat \\
			&= \frac{1}{\px{1}\pi} \sigmaSqax{1}{1} + \frac{1}{\px{1}(1-\pi)} \sigmaSqax{0}{1} + \frac{1}{\px{0}\pi} \sigmaSqax{1}{0} + \frac{1}{\px{0}(1-\pi)} \sigmaSqax{0}{0} - (v^2_{1} + v^2_{0})\\
			&= \frac{1}{\px{1}\pi} \Big(\sigmaSqax{1}{1} - \sigma^2_{\tilde{Y}(1)|X = 1}\Big) + \frac{1}{\px{1}(1-\pi)} \Big(\sigmaSqax{0}{1}- \sigma^2_{\tilde{Y}(0)|X = 1}\Big) \\
            &\quad+ \frac{1}{\px{0}\pi} \Big(\sigmaSqax{1}{0} - \sigma^2_{\tilde{Y}(1)|X = 0}\Big)
             + \frac{1}{\px{0}(1-\pi)} \Big(\sigmaSqax{0}{0} - \sigma^2_{\tilde{Y}(0)|X = 0}\Big) \\
             &\quad - \Big[\frac{1}{\px{1}^2} \big\{ \zetaStrA{1} + \zetaStrH{1}\big\} + \frac{1}{\px{0}^2}\big\{ \zetaStrA{0} + \zetaStrH{0}\big\}\Big] \\ 
			&= \frac{1}{\px{1}\pi} \Big[\frac{1}{\px{1}}\sum_{s,X(s)=1}p(s)\Econd{\maxZ{1}{1}}{S_i=s}^2\Big] + \frac{1}{\px{1}(1-\pi)} \Big[\frac{1}{\px{1}}\sum_{s,X(s)=1}p(s)\Econd{\maxZ{0}{1}}{S_i=s}^2\Big] \\
			&\quad + \frac{1}{\px{0}\pi} \Big[\frac{1}{\px{0}}\sum_{s,X(s)=0}p(s)\Econd{\maxZ{1}{0}}{S_i=s}^2\Big] + \frac{1}{\px{0}(1-\pi)} \Big[\frac{1}{\px{0}}\sum_{s,X(s)=0}p(s)\Econd{\maxZ{0}{0}}{S_i=s}^2\Big]\\
			&\quad - \frac{1}{\px{1}^2}   \sum_{s,X(s)=1}p(s) \Big[ \Econd{\maxZ{1}{1}}{S_i=s} -  \Econd{\maxZ{0}{1}}{S_i=s}\Big]^2 \\
			&\quad - \frac{1}{\px{0}^2} \sum_{s,X(s)=0}p(s) \Big[\Econd{\maxZ{1}{0}}{S_i=s} -  \Econd{\maxZ{0}{0}}{S_i=s}\Big]^2 \\
			&\quad - \frac{1}{\px{1}^2}  \sum_{s,X(s)=1}p(s)q(s) \Big[\frac{1}{\pi} \Econd{\maxZ{1}{1}}{S_i=s} + \frac{1}{1-\pi} \Econd{\maxZ{0}{1}}{S_i=s}\Big]^2\\
			&\quad - \frac{1}{\px{0}^2}  \sum_{s,X(s)=0}p(s)q(s) \Big[\frac{1}{\pi} \Econd{\maxZ{1}{0}}{S_i=s} + \frac{1}{1-\pi} \Econd{\maxZ{0}{0}}{S_i=s}\Big]^2  \\
			&= \frac{1}{p^2_1} \sum_{s,X(s)=1}p(s)\{\pi(1-\pi) - q(s)\}\Big[\frac{1}{\pi}\Econd{\maxZ{1}{1}}{S_i=s} + \frac{1}{1-\pi}\Econd{\maxZ{0}{1}}{S_i=s}\Big]^2 \\
			&\quad+ \frac{1}{p^2_0} \sum_{s,X(s)=0}p(s)\{\pi(1-\pi) - q(s)\}\Big[\frac{1}{\pi}\Econd{\maxZ{1}{0}}{S_i=s} + \frac{1}{1-\pi}\Econd{\maxZ{0}{0}}{S_i=s}\Big]^2 \\
			& \geq 0.
		\end{align*}
 }
	The inequality is strict unless
	\begin{equation*}
		\begin{aligned}
			\{\pi(1-\pi) - q(s)\}\Big[\frac{1}{\pi}\Econd{\maxZ{1}{1}}{S_i = s} + \frac{1}{1-\pi}\Econd{\maxZ{0}{1}}{S_i = s}\Big] = 0 \\
		\end{aligned}
	\end{equation*}
	for all $s\in \mathcal{S}$ satisfying $X(s) = 1$, and 
	\begin{equation*}
		\begin{aligned}
			\{\pi(1-\pi) - q(s)\}\Big[\frac{1}{\pi}\Econd{\maxZ{1}{0}}{S_i = s} + \frac{1}{1-\pi}\Econd{\maxZ{0}{0}}{S_i = s}\Big] = 0 
		\end{aligned}
	\end{equation*}
	for all $s\in \mathcal{S}$ satisfying $X(s) = 0$.
\end{proof}

\subsection{Proof of Theorem 3
}

\begin{proof}
	We first show $\hatzetaStrY{x} \inp \zetaStrY{x} $.
	For all $a \in \{1, 0\}$, $x \in \{1, 0\}$, and $s \in \mathcal{S}$,
	\begin{equation*}
		\begin{aligned}
			\hatmuas{a}{s} &= \frac{1}{\nas{a}{s}}\sum_{i=1}^n Y_i\ind{A_i=a, S_i=s} \\
			&= \frac{n}{\nas{a}{s}}\frac{1}{n}\sum_{i=1}^n \big[\Ya - \Econd{\Ya}{S_i = s}\big]\ind{A_i=a, S_i=s}  + \Econd{\Ya}{S_i = s}\\
			&\inp \Econd{\Ya}{S_i = s},
		\end{aligned}
	\end{equation*}	
	where the convergence follows from Lemma \ref{lemma:weak_law} and Slutsky's theorem. Moreover, 
	\begin{equation*}
		\begin{aligned}
			&\quad \frac{1}{\nax{a}{x}}\sum_{i=1}^nY_i^2 \ind{A_i=a, X_i=x}\\		
			&= \frac{n}{\nax{a}{x}}\frac{1}{n}\sum_{i=1}^n\big[Y_i^2(a) - \Econd{Y^2_i(a)}{X_i = x}\big]\ind{A_i=a, X_i=x} + \Econd{Y^2_i(a)}{X_i = x}\\
			&\inp \Econd{Y^2_i(a)}{X_i = x}
		\end{aligned}
	\end{equation*}	
	for the same reason. Hence, 
	\begin{equation*}
		\begin{aligned}
			&\quad \frac{1}{\nax{a}{x}}\sum_{i=1}^n Y_i^2 \ind{A_i = a, X_i = x} -\sum_{s,X(s)=x}\frac{n(s)}{n_x}\hat{\mu}^2_{n,a}(s)\\
			&\inp
			\Econd{Y^2_i(a)}{X_i = x} - \sum_{s,X(s)=x} \frac{p(s)}{\px{x}} \Econd{\Ya}{S_i = s}^2\\
			&= \sigma^2_{\tilde{Y}(a)\mid X = x},
		\end{aligned}
	\end{equation*}
	where the last equality has been given in (\ref{eq:sigmaSqtildeY}), and then $\hatzetaStrY{x} \inp \zetaStrY{x}$ follows.
	
	By further noticing $\Yax{a}{x} \inp \mu_{ax}$ for $a \in \{1, 0\}$, $x \in \{1, 0\}$, we can obtain
	$\hatzetaStrA{x} \inp \zetaStrA{x}$ and $\hatzetaStrH{x} \inp \zetaStrH{x}$. Hence, $\hat v^2_{1} + \hat v^2_{0} \inp v^2_{1} + v^2_{0}$, and when the null hypothesis is true,
	\begin{equation*}
		\begin{aligned}
		T_n^{t\text{-test,mod}}(V^{(n)}) &=	\frac{\sqrt{n}(\Yax{1}{1} - \Yax{0}{1} - \Yax{1}{0} + \Yax{0}{0})}{
			(\hat v^2_{1} + \hat v^2_{0})^{1/2}} \\
			&\indist N(0,1)
		\end{aligned}
	\end{equation*}
	by Theorem 1, the continuous mapping theorem, and Slutsky's theorem.
\end{proof}

\subsection{Proof of Theorem 4
}
	\begin{proof}
		First, for each $s \in \mathcal{S}$,
	\begin{equation*}
		\begin{aligned}
			 \hatmuas{1}{s} - \hatmuas{0}{s}
			&= \frac{\sum_{i=1}^{n}\Yt A_i\ind{S_i = s}}{\nas{1}{s}} - \frac{\sum_{i=1}^{n}Y_i(0)(1-A_i)\ind{S_i = s}}{\nas{0}{s}}\\
			&= \frac{\sum_{i=1}^{n}\big[\tildeYa{1} + \Econd{\maxZ{1}{x}}{S_i = s} + \muax{1}{x}\big]A_i\ind{S_i = s}}{\nas{1}{s}} \\
            &\quad - \frac{\sum_{i=1}^{n}\big[\tildeYa{0} + \Econd{\maxZ{0}{x}}{S_i = s} + \muax{0}{x}\big](1-A_i)\ind{S_i = s}}{\nas{0}{s}}\\
			&= \frac{\sum_{i=1}^{n}\tildeYa{1} A_i\ind{S_i = s}}{\nas{1}{s}} - \frac{\sum_{i=1}^{n}\tildeYa{0}(1-A_i)\ind{S_i = s}}{\nas{0}{s}} \\
            &\quad + \Econd{\maxZ{1}{x} - \maxZ{0}{x}}{S_i = s} + \tau_x ,\\
		\end{aligned}
	\end{equation*}
	where $x = X(s)$.
	
	Then, 
	\begin{equation*}
		\begin{aligned}
			&\quad \sqrt{n}\Big[\sum_{s,X(s)=1} \frac{n(s)}{\nx{1}}\big\{\hatmuas{1}{s} - \hatmuas{0}{s}\}  - \tau_1 \Big]\\
			&= \sqrt{n}\sum_{s,X(s)=1} \frac{n(s)}{\nx{1}} \frac{\sum_{i=1}^{n}\tildeYa{1} A_i\ind{S_i = s}}{\nas{1}{s}} - \sqrt{n} \sum_{s,X(s)=1}\frac{n(s)}{\nx{1}} \frac{\sum_{i=1}^{n}\tildeYa{0}(1-A_i)\ind{S_i = s}}{\nas{0}{s}}\\
			&\quad +  \sqrt{n}\sum_{s,X(s)=1} \frac{n(s)}{\nx{1}} \Econd{\maxZ{1}{1} - \maxZ{0}{1}}{S_i = s}\\
			&= \frac{1}{\px{1}}(\Lx{1} + \Sx{1}) + o_P(1),
		\end{aligned}
	\end{equation*}
	and similarly, 
	\begin{equation*}
		\begin{aligned}
			\sqrt{n}\Big[\sum_{s,X(s)=0} \frac{n(s)}{n_0}\{\hatmuas{1}{s} - \hatmuas{0}{s}\}  - \tau_0 \Big]= \frac{1}{\px{0}}(\Lx{0} + \Sx{0}) + o_P(1),
		\end{aligned}
	\end{equation*}
	where for $x \in \{1, 0\}$, 
		\begin{equation}
		\begin{aligned}
		 \Lx{x} &=   \sum_{s,X(s)=x} \frac{n(s)}{\sqrt{n}} \frac{\sum_{i=1}^{n}\tildeYa{1} A_i\ind{S_i = s}}{\nas{1}{s}} -  \sum_{s,X(s)=x}\frac{n(s)}{\sqrt{n}} \frac{\sum_{i=1}^{n}\tildeYa{0}(1-A_i)\ind{S_i = s}}{\nas{0}{s}},		
		\end{aligned}
		\label{Lx} 
		\end{equation}
	and $\Sx{x}$ is defined as in (\ref{EUSx}).
	
	In Lemma \ref{lemma:thm_strata_joint}, we show
	\begin{equation*}
		\begin{aligned}
			{\begin{pmatrix}
					\Lx{1} + \Sx{1} \\
					\Lx{0} + \Sx{0} \\
			\end{pmatrix}}
			\indist
			N\left({\begin{pmatrix}
					0 \\
					0 \\
			\end{pmatrix}}, {\begin{pmatrix}
					p^2_{1}u^2_{1} & 0 \\
					0 & p^2_{0}u^2_{0}\\
			\end{pmatrix}} \right),
		\end{aligned}
	\end{equation*}
	where $u^2_{0} = (1/\px{0}^2)\{\zetaStrY{0} +  \zetaStrH{0}\}$, and $u^2_{1} =(1/\px{1}^2)\{\zetaStrY{1} +  \zetaStrH{1}\}$.
	Then using Slutsky's theorem, we can obtain the desired result.
\end{proof}

\subsection{Proof of Theorem 5
}
\begin{proof}
	Recall the proof of Theorem 3. 
This is a direct result of Theorem 4
and the continuous mapping theorem and  Slutsky's theorem. 
\end{proof}

\subsection{Proof of Theorem 6
}
\begin{proof}
	Repeat the steps in the proof of Theorem 1 
 and we get
	\begin{equation*}
		\begin{aligned}
			\sqrt{n} \Big(\Yax{1}{1} - \Yax{0}{1} - \tau_1  \Big) = \frac{1}{\sqrt{n}}\Rx{1} (\Zx{1} + \Mx{1}),
		\end{aligned}
	\end{equation*}
	and
	\begin{equation*}
		\begin{aligned}
			\sqrt{n} \Big(\Yax{1}{0} - \Yax{0}{0} - \tau_0  \Big) = \frac{1}{\sqrt{n}}\Rx{0} (\Zx{0} + \Mx{0}),
		\end{aligned}
	\end{equation*}
	where $\Rx{x}$, $\Zx{x}$ and $\Mx{x}$, $x \in \{1, 0\}$, are defined as in (\ref{RZMx}). By Lemma \ref{lemma:weak_law}, we have $\Rx{1} \inp  \{{\pi(1-\pi)\px{1}^2}\}^{-1}$, $\Rx{0} \inp \{{\pi(1-\pi)\px{0}^2}\}^{-1}$; by Lemma \ref{lemma:thm_add_M} and Slutsky's theorem, we have ${\Zx{1}} = o_P(\sqrt{n})$, ${\Zx{0}} = o_P(\sqrt{n})$.
	
	In Lemma \ref{lemma:thm_add_M}, we show
	\begin{equation*}
		\begin{aligned}
			\frac{1}{\sqrt{n}} {\begin{pmatrix}
					\Mx{1} \\
					\Mx{0} \\
			\end{pmatrix}}
			\indist
			\pi(1-\pi)
			N\left({\begin{pmatrix}
					0 \\
					0 \\
			\end{pmatrix}}, {\begin{pmatrix}
					\px{1}^4 w^2_{1} & \px{1}^2\px{0}^2w_{10} \\
					\px{1}^2\px{0}^2w_{10} & \px{0}^4 w^2_{0} \\
			\end{pmatrix}} \right),
		\end{aligned}
	\end{equation*}
 	where $w^2_{0} = (1/p_0^2)\{\zetaAddY{0} + \zetaAddA{0} + \zetaAddH{0}\}$,  $w^2_{1} = (1/p_1^2)\{\zetaAddY{1} + \zetaAddA{1} + \zetaAddH{1}\}$, and $w_{10} =(\px{1}\px{0})^{-1}\{\zetaAddYxy{1}{0} + \zetaAddAxy{1}{0} + \zetaAddHxy{1}{0}\} $.
	Then using Slutsky's theorem, we can obtain the desired result.
\end{proof}

\subsection{Proof of Theorem 7
}
\begin{proof}
	In the proof of Theorem 2,
 we show
	\begin{equation*}
		\begin{aligned}
			n \left(\frac{\hatsigmaSqax{1}{1}}{\nax{1}{1}} + \frac{\hatsigmaSqax{0}{1}}{\nax{0}{1}} + \frac{\hatsigmaSqax{1}{0}}{\nax{1}{0}} + \frac{\hatsigmaSqax{0}{0}}{\nax{0}{0}} \right) \overset{P}\rightarrow \zetatold,
		\end{aligned}
	\end{equation*}
	where $\zetatold =  \sigmaSqax{1}{1} / ({\px{1}\pi})+  \sigmaSqax{0}{1} / \{\px{1}(1-\pi)\} +  \sigmaSqax{1}{0} / ({\px{0}\pi}) + \sigmaSqax{0}{0}/ \{\px{0}(1-\pi)\}$.
	
	We next define $\zetatprime =  ({1}/{\px{1}^2})\big\{\zetaAddY{1} + \zetaAddA{1} + \zetaAddH{1}\big\} + ({1}/{\px{0}^2})\big\{\zetaAddY{0} + \zetaAddA{0} + \zetaAddH{0}\big\}- \{{2}/({\px{0}\px{1}})\} \big\{\zetaAddYxy{1}{0}+ \zetaAddAxy{1}{0}+ \zetaAddHxy{1}{0}\big\}$ and compare it with $\zetatold$.
	
	Notice that 	
	\begin{equation}
		\begin{aligned}
			\sigma^2_{\tilde{r}_{x}(a)} 	&= 	E\Big(\Big[\{\Ya - \muax{a}{x}\}\ind{X_i = x} - \Econd{\raxZ{a}{x}}{S_i}\Big]^2\Big)\\
			&= 	E\Big[\{\Ya - \muax{a}{x}\}^2\ind{X_i = x} - 2\Econd{\raxZ{a}{x}}{S_i}\{\Ya - \muax{a}{x}\}\ind{X_i = x} \\
            &\quad\qquad+ \Econd{\raxZ{a}{x}}{S_i}^2\Big]\\
			&= \px{x} \sigmaSqax{a}{x} - E\Big[ \Econd{\raxZ{a}{x}}{S_i}^2\Big]\\
			&= \px{x} \sigmaSqax{a}{x} - \sum_{s \in \mathcal{S}}p(s)\Econd{\raxZ{a}{x}}{S_i=s}^2,
		\end{aligned}
	\label{eq:sigmaSqtilder}
	\end{equation}
	for $a \in \{1, 0\}$ and $x \in \{1, 0\}$.
	
	Then we have
  {\allowdisplaybreaks
		\begin{align*}
			\zetatold - \zetatprime &= \frac{1}{\px{1}\pi} \sigmaSqax{1}{1} + \frac{1}{\px{1}(1-\pi)} \sigmaSqax{0}{1} + \frac{1}{p_0\pi} \sigma^2_{10} + \frac{1}{p_0(1-\pi)} \sigma^2_{00}\\
			&\quad - \Big[\frac{1}{\px{1}^2} \big\{\zetaAddY{1} + \zetaAddA{1} + \zetaAddH{1}\big\} + \frac{1}{\px{0}^2}\big\{\zetaAddY{0} + \zetaAddA{0} + \zetaAddH{0}\big\} \\
            &\quad\qquad -\frac{2}{\px{1} \px{0}}\big\{\zetaAddYxy{1}{0} + \zetaAddAxy{1}{0}+\zetaAddHxy{1}{0}\big\}\Big] \\
			&= \frac{1}{\px{1}\pi} \Big(\sigmaSqax{1}{1} - \frac{1}{\px{1}} \sigma^2_{\tilde{r}_{1}(1)}\Big) + \frac{1}{\px{1}(1-\pi)} \Big(\sigmaSqax{0}{1}- \frac{1}{\px{1}} \sigma^2_{\tilde{r}_{1}(0)}\Big) + \frac{1}{\px{0}\pi} \Big(\sigmaSqax{1}{0} - \frac{1}{\px{0}} \sigma^2_{\tilde{r}_{0}(1)}\Big) \\
            &\quad + \frac{1}{\px{0}(1-\pi)} \Big(\sigmaSqax{0}{0} - \frac{1}{\px{0}} \sigma^2_{\tilde{r}_{0}(0)}\Big) - \Big[\frac{1}{\px{1}^2} \big\{ \zetaAddA{1} + \zetaAddH{1}\big\} + \frac{1}{\px{0}^2}\big\{ \zetaAddA{0} + \zetaAddH{0}\big\}\\
            &\hspace{2.5in} -\frac{2}{\px{1} \px{0}}\big\{\zetaAddYxy{1}{0} + \zetaAddAxy{1}{0}+\zetaAddHxy{1}{0}\big\}\Big] \\
			&= \frac{1}{\px{1}^2\pi} \sum_{s\in \mathcal{S}}p(s)\Econd{\raxZ{1}{1}}{S_i=s}^2 + \frac{1}{\px{1}^2(1-\pi)} \sum_{s\in \mathcal{S}}p(s)\Econd{\raxZ{0}{1}}{S_i=s}^2\\
			&\quad + \frac{1}{\px{0}^2\pi} \sum_{s\in \mathcal{S}}p(s)\Econd{\raxZ{1}{0}}{S_i=s}^2 + \frac{1}{\px{0}^2(1-\pi)} \sum_{s\in \mathcal{S}}p(s)\Econd{\raxZ{0}{0}}{S_i=s}^2 \\
			&\quad - \frac{1}{\px{1}^2} \sum_{s \in \mathcal{S}} p(s)q(s)  \Big[\frac{1}{\pi} \Econd{\raxZ{1}{1}}{ S_i=s} + \frac{1}{1-\pi} \Econd{\raxZ{0}{1}}{ S_i=s}\Big]^2\\
			&\quad - \frac{1}{\px{1}^2} \sum_{s \in \mathcal{S}}p(s) \Big[\Econd{\raxZ{1}{1}}{ S_i=s} -  \Econd{\raxZ{0}{1}}{ S_i=s}\Big]^2\\
			&\quad - \frac{1}{\px{0}^2} \sum_{s \in \mathcal{S}} p(s)q(s)  \Big[\frac{1}{\pi} \Econd{\raxZ{1}{0}}{ S_i = s} + \frac{1}{1-\pi} \Econd{\raxZ{0}{0}}{ S_i = s}\Big]^2\\
			&\quad - \frac{1}{\px{0}^2} \sum_{s \in \mathcal{S}}p(s) \Big[ \Econd{\raxZ{1}{0}}{ S_i = s} -  \Econd{\raxZ{0}{0}}{ S_i = s}\Big]^2\\
			&\quad - \frac{2}{\px{1}\px{0}}  \sum_{s \in \mathcal{S}}  p(s)\Big[\frac{1}{\pi}\Econd{\raxZ{1}{1}}{ S_i = s}\Econd{\raxZ{1}{0}}{ S_i = s} \\
            &\hspace{1.3in}+ \frac{1}{1-\pi} \Econd{\raxZ{0}{1}}{ S_i = s}\Econd{\raxZ{0}{0}}{ S_i = s}\Big]\\
			&\quad + \frac{2}{\px{1}\px{0}}\sum_{s \in \mathcal{S}}p(s)q(s) \Big[\frac{1}{\pi} \Econd{\raxZ{1}{1}}{ S_i=s} + \frac{1}{1-\pi} \Econd{\raxZ{0}{1}}{ S_i=s}\Big]\\
            &\hspace{2in}\Big[\frac{1}{\pi} \Econd{\raxZ{1}{0}}{ S_i = s} + \frac{1}{1-\pi} \Econd{\raxZ{0}{0}}{ S_i = s}\Big] \\
			&\quad + \frac{2}{\px{1}\px{0}} \sum_{s \in \mathcal{S}}p(s)\Big[ \Econd{\raxZ{1}{1}}{ S_i=s}- \Econd{\raxZ{0}{1}}{ S_i=s}\Big]\\
            &\hspace{1.6in}\Big[ \Econd{\raxZ{1}{0}}{ S_i = s}- \Econd{\raxZ{0}{0}}{ S_i = s}\Big].\\
		\end{align*}
 }
	Continue with the derivation:
          {\allowdisplaybreaks
		\begin{align*}
			\zetatold - \zetatprime 
   			&= \frac{1}{\px{1}^2\pi} \sum_{s\in \mathcal{S}}p(s)\Econd{\raxZ{1}{1}}{S_i=s}^2 + \frac{1}{\px{1}^2(1-\pi)} \sum_{s\in \mathcal{S}}p(s)\Econd{\raxZ{0}{1}}{S_i=s}^2\\
			&\quad + \frac{1}{\px{0}^2\pi} \sum_{s \in \mathcal{S}}p(s)\Econd{\raxZ{1}{0}}{S_i=s}^2 + \frac{1}{\px{0}^2(1-\pi)} \sum_{s\in \mathcal{S}}p(s)\Econd{\raxZ{0}{0}}{S_i=s}^2 \\
			&\quad - \frac{1}{\px{1}^2} \sum_{s \in \mathcal{S}}p(s) \Big[ \Econd{\raxZ{1}{1}}{ S_i=s} -  \Econd{\raxZ{0}{1}}{ S_i=s}\Big]^2\\
			&\quad - \frac{1}{\px{0}^2} \sum_{s \in \mathcal{S}}p(s) \Big[ \Econd{\raxZ{1}{0}}{ S_i = s} -  \Econd{\raxZ{0}{0}}{ S_i = s}\Big]^2\\
			&\quad -  \sum_{s \in \mathcal{S}}p(s)q(s)\Big(\frac{1}{\px{1}}\Big[\frac{1}{\pi} \Econd{\raxZ{1}{1}}{ S_i=s} + \frac{1}{1-\pi} \Econd{\raxZ{0}{1}}{ S_i=s}\Big]\\ 
			&\hspace{1.3in} - \frac{1}{\px{0}}\Big[\frac{1}{\pi} \Econd{\raxZ{1}{0}}{ S_i = s} + \frac{1}{1-\pi} \Econd{\raxZ{0}{0}}{ S_i = s}\Big]\Big)^2\\
			&\quad - \frac{2}{\px{1}\px{0}}  \sum_{s \in \mathcal{S}}  p(s)\Big[\frac{1}{\pi}\Econd{\raxZ{1}{1}}{ S_i = s}\Econd{\raxZ{1}{0}}{ S_i = s} \\
            &\hspace{1.3in}+ \frac{1}{1-\pi} \Econd{\raxZ{0}{1}}{ S_i = s}\Econd{\raxZ{0}{0}}{ S_i = s}\Big]\\
			&\quad + \frac{2}{\px{1}\px{0}} \sum_{s \in \mathcal{S}}p(s)\Big[ \Econd{\raxZ{1}{1}}{ S_i=s}- \Econd{\raxZ{0}{1}}{ S_i=s}\Big]\\
            &\hspace{1.4in} \Big[\Econd{\raxZ{1}{0}}{ S_i = s}- \Econd{\raxZ{0}{0}}{ S_i = s}\Big]\\
			&=   \sum_{s \in \mathcal{S}}p(s)\{\pi(1-\pi)-q(s)\}\Big(\frac{1}{\px{1}}\Big[\frac{1}{\pi} \Econd{\raxZ{1}{1}}{ S_i=s} + \frac{1}{1-\pi} \Econd{\raxZ{0}{1}}{ S_i=s}\Big]\\ 
			&\qquad - \frac{1}{\px{0}}\Big[\frac{1}{\pi} \Econd{\raxZ{1}{0}}{ S_i = s} + \frac{1}{1-\pi} \Econd{\raxZ{0}{0}}{ S_i = s}\Big]\Big)^2\\
			&\geq 0.
		\end{align*} }
		The inequality is strict unless
			\begin{equation}
			\begin{aligned}
				&\quad\{\pi(1-\pi)-q(s)\}\Big(\frac{1}{\px{1}}\Big[\frac{1}{\pi} \Econd{\raxZ{1}{1}}{S_i = s} + \frac{1}{1-\pi} \Econd{ \raxZ{0}{1}}{S_i = s}\Big]\\ &\qquad - \frac{1}{\px{0}}\Big[\frac{1}{\pi} \Econd{\raxZ{1}{0}}{S_i = s} + \frac{1}{1-\pi} \Econd{\raxZ{0}{0}}{S_i = s}\Big]\Big) = 0
			\end{aligned}
		\end{equation}
		for all $s \in \mathcal{S}$.
\end{proof}

\subsection{Proof of Theorem 8
}

\begin{proof}
	We first show $\hatzetaAddY{x} \inp \zetaAddY{x} $.
	For all $a \in \{1, 0\}$, $x \in \{1, 0\}$, and $s \in \mathcal{S}$,
	\begin{equation*}
		\begin{aligned}
			\hat{\mu}_{n,ax}(s) &= \frac{1}{\nax{a}{x}(s)}\sum_{i=1}^n Y_i\ind{A_i=a, S_i=s, X_i = x} \\
			&= \frac{1}{\nax{a}{x}(s)}\sum_{i=1}^n \Big[\Ya - \Econd{\Ya}{S_i = s, X_i = x}\Big]\ind{A_i=a, S_i=s, X_i = x} \\
            &\quad + \Econd{\Ya}{S_i = s, X_i = x}\\
			&= \frac{n}{\nax{a}{x}(s)} \frac{1}{n} \sum_{i=1}^n \Big[\Ya- \Econd{\Ya}{S_i = s, X_i = x}\Big]\ind{A_i=a, S_i=s, X_i = x} \\
            &\quad + \Econd{\Ya}{S_i = s, X_i = x}\\
			&\inp \Econd{\Ya}{S_i = s, X_i = x},
		\end{aligned}
	\end{equation*}	
	where the convergence follows from Lemma \ref{lemma:weak_law} and Slutsky's theorem. Also, 
	\begin{equation*}
		\begin{aligned}
			\frac{1}{\nax{a}{x}}\sum_{i=1}^n(Y_i- \Yax{a}{x})^2\ind{A_i = a, X_i=x}
			\inp \sigmaSqax{a}{x},
		\end{aligned}
	\end{equation*}	
	as shown in the proof of Theorem 2.
	It follows that
	\begin{equation*}
		\begin{aligned}
			&\quad\hat{p}_x\frac{1}{\nax{a}{x}}\sum_{i=1}^n(Y_i- \Yax{a}{x})^2\ind{A_i = a, X_i=x} -\sum_{s \in \mathcal{S}}\frac{n(s)}{n}\Big\{\frac{\nxs{x}{s}}{\ns{s}}\Big\}^2\big\{\hat{\mu}_{n,ax}(s) - \Yax{a}{x}\big\}^2\\
			&\inp
			\px{x}\sigmaSqax{a}{x} - \sum_{s \in \mathcal{S}} p(s)\Big\{\frac{\pxs{x}{s}}{\ps{s}}\Big\}^2 \big[\Econd{\Ya}{S_i = s, X_i = x}- \mu_{ax}\big]^2\\
			&=  \px{x} \sigmaSqax{a}{x} - \sum_{s \in \mathcal{S}}p(s)\Econd{\raxZ{a}{x}}{S_i=s}^2 \\
			&= \sigma^2_{\tilde{r}_{x}(a)},
		\end{aligned}
	\end{equation*}
	where the last equality has been given in (\ref{eq:sigmaSqtilder}). Hence, we can get $\hatzetaAddY{x} \inp \zetaAddY{x}$.
	
	Since $\{{\nxs{x}{s}}/{\ns{s}}\}\{\hat{\mu}_{n,ax}(s) - \Yax{a}{x}\} \inp \Econd{\raxZ{a}{x}}{S_i=s}$ for $a \in \{1, 0\}$ and $x \in \{1, 0\}$, we can establish
	$\hatzetaAddA{x} \inp \zetaAddA{x}$, $\hatzetaAddH{x} \inp \zetaAddH{x}$, $\hatzetaAddYxy{1}{0} \inp \zetaAddYxy{1}{0}$, $\hatzetaAddAxy{1}{0} \inp \zetaAddAxy{1}{0}$, and $\hatzetaAddHxy{1}{0} \inp \zetaAddHxy{1}{0}$. Therefore, $\hat w^2_{1} + \hat w^2_{0}  - 2\hat w_{10} \inp w^2_{1} + w^2_{0} - 2 w_{10}$.
	
	Hence, when the null hypothesis is true, 
	\begin{equation*}
		\begin{aligned}
			T_n^{t\text{-test,mod}^\prime}({V}^{(n)}) &=	\frac{\sqrt{n}(\Yax{1}{1} - \Yax{0}{1} - \Yax{1}{0} + \Yax{0}{0})}{(\hat w^2_{1} + \hat w^2_{0}  - 2\hat w_{10})^{1/2}}\\
			&\indist N(0,1),
		\end{aligned}
	\end{equation*}
	by Theorem 6,
 the continuous mapping theorem, and Slutsky's theorem. 
\end{proof}

\subsection{Proof of Theorem 9
}
\begin{proof}
	
	Let $\checkYa{a} = \Ya - \Econd{\Ya}{S_i, X_i}$, and $\maxZ{a}{x}$ be defined as in the proof of Lemma 	\ref{lemma:thm_1_M}, for $a \in \{1,0\}$ and $x \in \{1,0\}$. Then 
	\begin{equation*}
		\begin{aligned}
			\Ya - \muax{a}{x} = \checkYa{a} + \Econd{\maxZ{a}{x}}{S_i, X_i}.
		\end{aligned}
	\end{equation*}
        Notice that 
	\begin{equation*}
		\begin{aligned}
	       &\quad \hatmuaxs{1}{x}{s} - \hatmuaxs{0}{x}{s} \\
            &= \frac{\sum_{i=1}^{n}\Yt A_i\ind{S_i = s, X_i = x}}{\naxs{1}{x}{s}} - \frac{\sum_{i=1}^{n}Y_i(0)(1-A_i)\ind{S_i = s, X_i = x}}{\naxs{0}{x}{s}}\\
			&= \frac{\sum_{i=1}^{n}\big(\Yt-\muax{1}{x}\big)A_i\ind{S_i = s, X_i = x}}{\naxs{1}{x}{s}}\\ 
            &\quad- \frac{\sum_{i=1}^{n}\big(Y_i(0)-\muax{0}{x}\big)(1-A_i)\ind{S_i = s, X_i = x}}{\naxs{0}{x}{s}} + \tau_x\\
			&= \frac{\sum_{i=1}^{n}\big[\checkYa{1} + \Econd{\maxZ{1}{x}}{S_i = s, X_i = x} \big]A_i\ind{S_i = s, X_i = x}}{\naxs{1}{x}{s}}\\
			&\quad - \frac{\sum_{i=1}^{n}\big(\checkYa{0} + \Econd{\maxZ{0}{x}}{S_i = s, X_i  = x}\big)(1-A_i)\ind{S_i = s, X_i = x}}{\naxs{0}{x}{s}} + \tau_x\\
			&= \frac{\sum_{i=1}^{n}\checkYa{1} A_i\ind{S_i = s, X_i = x}}{\naxs{1}{x}{s}} - \frac{\sum_{i=1}^{n}\checkYa{0}(1-A_i)\ind{S_i = s, X_i = x}}{\naxs{0}{x}{s}}\\
			&\qquad + \Econd{\maxZ{1}{x} - \maxZ{0}{x}}{S_i = s, X_i  = x}  + \tau_x.
		\end{aligned}
	\end{equation*}
	
	Then, 
	\begin{equation*}
		\begin{aligned}
			&\quad \sqrt{n}\Big[\sum_{s \in \mathcal{S}} \frac{\nxs{x}{s}}{\nx{x}}\big\{\hatmuaxs{1}{x}{s} - \hatmuaxs{0}{x}{s}\big\}  - \tau_x\Big]\\
			&= \sqrt{n}\sum_{s \in \mathcal{S}} \frac{\nxs{x}{s}}{\nx{x}} \frac{\sum_{i=1}^{n}\checkYa{1} A_i\ind{S_i = s, X_i = x}}{\naxs{1}{x}{s}}\\
            &\quad - \sqrt{n} \sum_{s \in \mathcal{S}}\frac{\nxs{x}{s}}{\nx{x}} \frac{\sum_{i=1}^{n}\checkYa{0}(1-A_i)\ind{S_i = s, X_i = x}}{\naxs{0}{x}{s}}\\
			&\quad +  \sqrt{n}\sum_{s \in \mathcal{S}} \frac{\nxs{x}{s}}{\nx{x}} \Econd{\maxZ{1}{x}  -  \maxZ{0}{x}}{S_i = s, X_i=x} \\
			&= \frac{1}{\px{x}}(\Lcheckx{x} + \Scheckx{x}) + o_P(1),
		\end{aligned}
	\end{equation*}
	where 
	\begin{equation}
		\begin{aligned}
		\Lcheckx{x}&=   \sum_{s \in \mathcal{S}} \frac{\nxs{x}{s}}{\sqrt{n}} \frac{\sum_{i=1}^{n}\checkYa{1} A_i\ind{S_i = s, X_i = x}}{\naxs{1}{x}{s}} \\
        &\quad - \sum_{s \in \mathcal{S}}\frac{\nxs{x}{s}}{\sqrt{n}} \frac{ \sum_{i=1}^{n}\checkYa{0}(1-A_i)\ind{S_i = s, X_i = x}}{\naxs{0}{x}{s}},\\
		\Scheckx{x} &=  \sum_{s \in \mathcal{S}} \frac{\nxs{x}{s}}{\sqrt{n}} \Econd{\maxZ{1}{x}  -  \maxZ{0}{x}}{S_i = s, X_i=x},
		\end{aligned}
		\label{LSyx}
	\end{equation}
	for $x \in \{1, 0\}$.
	
	In Lemma \ref{lemma:thm_strata_joint_add}, we show
	\begin{equation*}
		\begin{aligned}
			{\begin{pmatrix}
					\Lcheckx{1} + \Scheckx{1}\\
					\Lcheckx{0} + \Scheckx{0}\\
			\end{pmatrix}}
			\indist
			N\left({\begin{pmatrix}
					0 \\
					0 \\
			\end{pmatrix}}, {\begin{pmatrix}
					\px{1}^2 s^2_{1}& 0 \\
					0 & \px{0}^2 s^2_{0}\\
			\end{pmatrix}} \right),
		\end{aligned}
	\end{equation*}
    where $s^2_{0} = ({1}/{\px{0}^2})\{\zetaCheckY{0}  + \zetaCheckH{0}\}$, and $s^2_{1} = ({1}/{\px{1}^2})\{\zetaCheckY{1}  + \zetaCheckH{1}\}$.
	Then using Slutsky's theorem, we can obtain the desired result.
\end{proof}
\subsection{Proof of Theorem 10
}

\begin{proof}
	We first show $\hatzetaCheckY{x} \inp \zetaCheckY{x} $.
	For all $a \in \{1, 0\}$, $x \in \{1, 0\}$, notice that 

\begin{equation}
	\begin{aligned}
		\sigma^2_{\checkYa{a}|X = x} 
		&= \Econd{Y^2_i(a)}{X_i=x} - 2E[\Ya \Econd{\Ya}{S_i, X_i}\mid X_i=x] \\
        &\quad + E[\Econd{\Ya}{S_i, X_i}^2\mid X_i=x]\\
		&= \sigmaSqax{a}{x} - E[\Econd{\Ya}{S_i, X_i}^2\mid X_i=x] + \mu_{ax}^2\\
		&= \sigmaSqax{a}{x} - \operatorname{var}[\Econd{\Ya}{S_i,X_i}\mid X_i=x]\\
		&= \sigmaSqax{a}{x} - E[\Econd{\maxZ{a}{x}}{S_i,X_i}^2\mid X_i = x]\\
		&= \sigmaSqax{a}{x}- \frac{1}{\px{x}}\sum_{s \in \mathcal{S}}\pxs{x}{s}\Econd{\maxZ{a}{x}}{S_i=s, X_i = x}^2,
	\end{aligned}
	\label{eq:sigmaSqcheckY}
\end{equation}
	and hence, 
	\begin{equation*}
	\begin{aligned}
		&\quad\hat{p}_x\frac{1}{\nax{a}{x}}\sum_{i=1}^n(Y_i- \Yax{a}{x})^2\ind{A_i = a, X_i=x} -\sum_{s \in \mathcal{S}}\frac{\nxs{x}{s}}{n}\big\{\hat{\mu}_{n,ax}(s) - \Yax{a}{x}\big\}^2\\
		&\inp
		\px{x}\sigmaSqax{a}{x} - \sum_{s \in \mathcal{S}} \pxs{x}{s} \big[\Econd{\Ya}{S_i = s, X_i = x}- \mu_{ax}\big]^2\\
		&=  \px{x} \sigmaSqax{a}{x} - \sum_{s \in \mathcal{S}}\pxs{x}{s}\Econd{\maxZ{a}{x}}{S_i=s, X_i=x}^2 \\
		&= \px{x}\sigma^2_{\check{Y}(a)|X = x} ,
	\end{aligned}
	\end{equation*}
	where the convergence follows from the weak law of large numbers, Lemma \ref{lemma:weak_law}, and Slutsky's theorem, as shown in the proof of Theorem 8.
Thus, we can establish $\hatzetaCheckY{x} \inp \zetaCheckY{x}$. The proof of $\hatzetaCheckH{x} \inp \zetaCheckH{x}$ is similar. Hence, we can conclude that $\hat s^2_{1} + \hat s^2_{0} \inp s^2_{1} + s^2_{0}$.
	
	Then by Theorem 9, 
 the continuous mapping theorem, and Slutsky's theorem, when the null hypothesis is true,
	\begin{equation*}
		\begin{aligned}
					T_n^{\text{strat}^\prime}({V}^{(n)}) &=	\frac{\sqrt{n}\Big[\sum_{s \in \mathcal{S}} \frac{\nxs{1}{s}}{\nx{1}}\big\{ \hatmuaxs{1}{1}{s} - \hatmuaxs{0}{1}{s}\big\}- \sum_{s \in \mathcal{S}} \frac{\nxs{0}{s}}{\nx{0}}\big\{\hatmuaxs{1}{0}{s} - \hatmuaxs{0}{0}{s}\big\}\Big]}{(\hat s^2_{1} + \hat s^2_{0})^{1/2}}\\
			&\indist N(0,1).
		\end{aligned}
	\end{equation*}
\end{proof}

\subsection{Proof of Theorem 11
}
	\begin{proof}
		The proof of Theorem 6
  can be generalized to categorical $X_i$ and we omit the proof here.
	\end{proof}

\subsection{Proof of Theorem 12
}
\begin{proof}
	Let
		\begin{equation*}
	\begin{aligned}
		{\Sigma}_{hc} =  \operatorname{diag}\left(\frac{n\sigmaSqax{1}{x}}{\nax{1}{x}} + \frac{n\sigmaSqax{0}{x}}{\nax{0}{x}}: x \in \mathcal{X}\right), \quad 		\hat{\Sigma}_{mod} &= \begin{pmatrix}
			\hat w^2_{0} & \hat w_{10} & \cdots & \hat w_{K0}\\
			\hat w_{10}& \hat w^2_{1}& \cdots & \hat w_{K1} \\
			\vdots & \vdots& \ddots &\vdots \\
			\hat w_{K0} & \hat w_{K1} & \cdots & \hat w^2_{K}
		\end{pmatrix}.\\
	\end{aligned}
		\end{equation*}
Then by the weak law of large numbers and  Lemma \ref{lemma:weak_law}, $\hat{\Sigma}_{hc} \inp {\Sigma}_{hc}$. The proof of Theorem 8 
 can also be generalized to show $\hat{\Sigma}_{mod} \inp \Sigma$, where $\Sigma$ is the asymptotic covariance matrix in Theorem 11. 
 Then by Theorem 12
 and elementary properties of the multivariate normal distribution,
			\begin{equation}
		\begin{aligned}
			W_n^{\text{Wald,mod}}({V}^{(n)}) &= n(	{R}\hat{{\tau}})^{\mathrm{T}}(	{R}\hat{\Sigma}_{mod}	{R}^{\mathrm{T}})^{-1}(	{R}\hat{{\tau}})\\
			& \inp  \chi^2({K}).
		\end{aligned}
	\label{eq:chi_K}
	\end{equation}
	
	Furthermore, 
	\begin{equation}
		\begin{aligned}
			&\quad {\Sigma}_{hc} - \Sigma \\&= 
			\sum_{s \in \mathcal{S}}p(s)\{\pi(1-\pi)-q(s)\} \Big(\frac{1}{\px{x}}\big[\frac{1}{\pi} \Econd{\raxZ{1}{x}}{S_i = s} + \frac{1}{1-\pi} \Econd{\raxZ{0}{x}}{S_i = s}\big]: x \in \mathcal{X}\Big)\\
			&\qquad \Big(\frac{1}{\px{x}}\big[\frac{1}{\pi} \Econd{\raxZ{1}{x}}{S_i = s} + \frac{1}{1-\pi} \Econd{\raxZ{0}{x}}{S_i = s}\big]: x \in \mathcal{X}\Big)^{\mathrm{T}}, \\
		\end{aligned}
		\label{eq:sigma_mod}
	\end{equation}
	and 
	${R}{\Sigma}_{hc}{R}^{\mathrm{T}} \succeq {R}{\Sigma}{R}^{\mathrm{T}}$. The inequality is strict unless
	\begin{equation}
		\begin{aligned}
			&\quad \{\pi(1-\pi)-q(s)\}\Big(\frac{1}{\px{x}}\big[\frac{1}{\pi} \Econd{\raxZ{1}{x}}{S_i = s} + \frac{1}{1-\pi} \Econd{\raxZ{0}{x}}{S_i = s}\big]\\ &\qquad - \frac{1}{\px{0}}\big[\frac{1}{\pi} \Econd{\raxZ{1}{0}}{S_i = s} + \frac{1}{1-\pi} \Econd{\raxZ{0}{0}}{S_i = s}\big]\Big)^2 = 0\\
			\label{cond:1_add_multi}
		\end{aligned}
	\end{equation}
	for all $x \in \mathcal{X}$ and $s \in \mathcal{S}$.
\end{proof}

\subsection{Proof of Theorem 13
}

\begin{proof}
	In the proof of Theorem 12, we obtain (\ref{eq:chi_K}), which immediately implies our desired results.
\end{proof}

\subsection{Proof of Theorem 14
}
\begin{proof}
		The proof of Theorem 9
  can be generalized to categorical  $X_i$ and we omit the proof here.
\end{proof}

\subsection{Proof of Theorem 15
}
\begin{proof}
	By Theorem 14
 and elementary properties of the multivariate normal distribution,
	\begin{equation}
		\begin{aligned}
			W_n^{\text{strat}}({V}^{(n)}) &= n(	{R}\hat{{\tau}}_{strat})^{\mathrm{T}}(	{R}\hat{\Sigma}_{strat}	{R}^{\mathrm{T}})^{-1}(	{R}\hat{{\tau}}_{strat})\\
			& \inp \chi^2({K}).
		\end{aligned}
		\label{eq:chi_K_strata}
	\end{equation}
	Then our desired result follows.
\end{proof}

	\section{Additional results}
	\begin{lemma}
		\label{lemma:weak_law}
		Suppose that Assumptions 
     2 and 4
     hold. Let ${W}_i = (\Yt, \Yc, {Z}_i^{\mathrm{T}})^{\mathrm{T}}$, $i = 1, \dots, n$, and $f(\cdot)$ be a measurable function satisfying $\E{|f({W}_i)|} < \infty$. Then
		\begin{equation*}
			\begin{aligned}
				\frac{1}{n} \sum_{i=1}^{n}f({W}_i)A_i \inp \pi \E{f({W}_i)}.
			\end{aligned}
		\end{equation*}
	
	\begin{proof}
		\cite{Bugni2018} obtained a similar result with ${W}_i = (\Yt, \Yc, S_i)^{\mathrm{T}}$ (see Lemma B.3). \cite{ma2022} pointed out that their results can be generalized to ${W}_i = (\Yt, \Yc, {Z}_i^{\mathrm{T}})^{\mathrm{T}}$. For the completeness, here we rewrite the proof in \cite{Bugni2018} with only slight modification. 
		
		Let $F_i = f({W}_i)$. Independently for each $s \in \mathcal{S}$ and independently of $( A^{(n)}, S^{(n)})$, let $\{F_i^s: 1 \leq i \leq n\}$ be i.i.d. with marginal distribution equal to the distribution of $F_i\mid S_i = s$. Notice that 	
	\begin{equation*}
			\begin{aligned}
				\Big\{\frac{1}{n} \sum_{i=1}^{n}F_iA_i \Bigm| A^{(n)}, S^{(n)}\Big\} \equaldist \Big\{\sum_{s \in \mathcal{S}}\frac{1}{n}\sum_{i=1}^{\nas{1}{s}}F_i^s \Bigm| A^{(n)}, S^{(n)}\Big\},
			\end{aligned}
	\end{equation*}
and hence 
		\begin{equation*}
		\begin{aligned}
		\frac{1}{n} \sum_{i=1}^{n}F_iA_i  \equaldist \sum_{s \in \mathcal{S}}\frac{1}{n}\sum_{i=1}^{\nas{1}{s}}F_i^s. 
		\end{aligned}
	\end{equation*}

	By Assumption 4
 and the weak law of large numbers, ${\nas{1}{s}}/{n} = \{{\Dns}/{n(s)} 
 + \pi\}\{{n(s)}/{n}\} \inp\pi p(s)$. It remains to show 
		\begin{equation*}
		\begin{aligned}
			\frac{1}{\nas{1}{s}} \sum_{i=1}^{\nas{1}{s}} F_i^s \inp E(F_i^s),
		\end{aligned}
	\end{equation*}
	which implies that
		\begin{equation*}
		\begin{aligned}
			\sum_{s \in \mathcal{S}}\frac{1}{n}\sum_{i=1}^{\nas{1}{s}}F_i^s &= \sum_{s \in \mathcal{S}}\frac{\nas{1}{s}}{n}\frac{1}{\nas{1}{s}}\sum_{i=1}^{\nas{1}{s}}F_i^s\\
			&\inp \sum_{s \in \mathcal{S}} \pi p(s) E(F_i^s) \\
			&= \pi E(F_i).
			\end{aligned}
\end{equation*}
	Use the almost sure representation theorem to construct ${\tildenas{1}{s}}/{n}$ such that ${\tildenas{1}{s}}/{n} \equaldist {\nas{1}{s}}/{n}$ and ${\tildenas{1}{s}}/{n} \rightarrow \pi p(s)$ a.s. Then by the independence of $(A^{(n)}, S^{(n)})$ and $\{F_i^s: 1 \leq i \leq n\}$, for any $\epsilon > 0$, we have
			\begin{equation*}
		\begin{aligned}
			P\Bigg \{ \Bigg | \frac{1}{\nas{1}{s}} \sum_{i=1}^{\nas{1}{s}} F_i^s - E(F_i^s) \Bigg | > \epsilon \Bigg \} &= P\Bigg \{ \Bigg | \frac{1}{n \frac{\nas{1}{s}}{n}} \sum_{i=1}^{n \frac{\nas{1}{s}}{n}} F_i^s - E(F_i^s) \Bigg | > \epsilon \Bigg \} \\
			&= P\Bigg \{ \Bigg | \frac{1}{n \frac{\tildenas{1}{s}}{n}} \sum_{i=1}^{n \frac{\tildenas{1}{s}}{n}} F_i^s - E(F_i^s) \Bigg | > \epsilon \Bigg \} \\
			&= E \Bigg [P\Bigg \{ \Bigg | \frac{1}{n \frac{\tildenas{1}{s}}{n}} \sum_{i=1}^{n \frac{\tildenas{1}{s}}{n}} F_i^s - E(F_i^s) \Bigg | > \epsilon \Bigm | \frac{\tildenas{1}{s}}{n}  \Bigg \} \Bigg] \\
			&\rightarrow 0,
		\end{aligned}
	\end{equation*}
	where the convergence is due to the dominated convergence theorem, and 
 \begin{equation}
		\begin{aligned}
			P\Bigg \{ \Bigg | \frac{1}{n \frac{\tildenas{1}{s}}{n}} \sum_{i=1}^{n \frac{\tildenas{1}{s}}{n}} F_i^s - E(F_i^s) \Bigg | > \epsilon \Bigm | \frac{\tildenas{1}{s}}{n}  \Bigg \} \rightarrow 0 \text{ a.s. }
			\end{aligned}
		\label{eq:as_weak_law}
		\end{equation}
	To reach (\ref{eq:as_weak_law}), notice that the weak law of large numbers implies that
	\begin{equation}
		\begin{aligned}
			\frac{1}{n_k} \sum_{i=1}^{n_k}F_i^s \inp E(F_i^s)
		\label{eq:weak_law_cor}
		\end{aligned}
	\end{equation}
	for any subsequence $n_k \rightarrow \infty$ as $k \rightarrow \infty$. Since $n \{{\tildenas{1}{s}}/{n}\} \rightarrow \infty$ a.s., (\ref{eq:as_weak_law}) follows from the independence of ${\tildenas{1}{s}}/{n} $ and $\{F_i^s: 1 \leq i \leq n\}$, and (\ref{eq:weak_law_cor}). 
	\end{proof}
\end{lemma}

	\begin{lemma}
	\label{lemma:thm_1_M}
	Let $\Mx{1}$ and $\Mx{0}$ be defined as in (\ref{RZMx}). Then under assumptions of Theorem 1 
	 		\begin{equation}
		\begin{aligned}
			\frac{1}{\sqrt{n}} {\begin{pmatrix}
					\Mx{1} \\
					\Mx{0} \\
			\end{pmatrix}}
			\indist
			\pi(1-\pi)
			N\left({\begin{pmatrix}
					0 \\
					0 \\
			\end{pmatrix}}, {\begin{pmatrix}
					\px{1}^4 v^2_{1} & 0 \\
					0 & \px{0}^4 v^2_{0} \\
			\end{pmatrix}} \right),
		\end{aligned}
	\label{eq:t_M_dist}
	\end{equation}
	where $v^2_{0}$ and $v^2_{1}$ are defined as in Theorem 1.
	\begin{proof}
			 Let $\tildeYa{a} = \Ya - \Econd{\Ya}{S_i}$, $\maxZ{a}{x} =\Econd{\Ya}{Z_i} - \muax{a}{x}$, for  $a \in \{1,0\}$ and $x \in \{1,0\}$. Then 
		\begin{equation*}
			\begin{aligned}
				\Ya - \muax{a}{x} = \tildeYa{a} + \Econd{\maxZ{a}{x}}{S_i}.
			\end{aligned}
		\end{equation*}
        Notice that  
        {\allowdisplaybreaks
			\begin{align*}
				&\quad \frac{1}{\sqrt{n}}\Mx{1} \\
				&= \frac{\px{1}\pi(1-\pi)}{\sqrt{n}} \sum_{i=1}^{n}\Big[\frac{1}{\pi}\{\Yt - \muax{1}{1}\}A_i\ind{X_i = 1} - \frac{1}{1-\pi} \{\Yc - \muax{0}{1}\}(1-A_i)\ind{X_i = 1}\Big] \\
				&= \frac{\px{1}\pi(1-\pi)}{\sqrt{n}} \sum_{i=1}^{n}\Big\{\frac{1}{\pi}\tildeYa{1}  A_i\ind{X_i = 1} - \frac{1}{1-\pi} \tildeYa{0}(1-A_i)\ind{X_i = 1}\Big\} \\
				&\quad + \frac{\px{1}\pi(1-\pi)}{\sqrt{n}} \sum_{i=1}^{n}\Big[\frac{1}{\pi} \Econd{\maxZ{1}{1}}{S_i}A_i\ind{X_i = 1} - \frac{1}{1-\pi}  \Econd{\maxZ{0}{1}}{S_i}(1-A_i)\ind{X_i = 1}\Big] \\
				&= \px{1}\pi(1-\pi) \Ex{1} \\
                    &\quad + \frac{\px{1}\pi(1-\pi)}{\sqrt{n}} \sum_{i=1}^{n}A_i\Big[\frac{1}{\pi} \Econd{\maxZ{1}{1}}{S_i}\ind{X_i = 1} + \frac{1}{1-\pi}  \Econd{\maxZ{0}{1}}{S_i}\ind{X_i = 1}\Big] \\
				&\quad - \frac{\px{1}\pi(1-\pi)}{\sqrt{n}}\sum_{i=1}^{n}\frac{1}{1-\pi}  \Econd{\maxZ{0}{1}}{S_i}\ind{X_i = 1} \\
				&= \px{1}\pi(1-\pi) \Ex{1} \\
                    &\quad + \frac{\px{1}\pi(1-\pi)}{\sqrt{n}} \sum_{i=1}^{n}(A_i-\pi)\Big[\frac{1}{\pi} \Econd{\maxZ{1}{1}}{S_i}\ind{X_i = 1} + \frac{1}{1-\pi}  \Econd{\maxZ{0}{1}}{S_i}\ind{X_i = 1}\Big] \\
				&\quad + \frac{\px{1}\pi(1-\pi)}{\sqrt{n}} \sum_{i=1}^{n}\Big[\Econd{\maxZ{1}{1}}{S_i}\ind{X_i = 1} - \Econd{\maxZ{0}{1}}{S_i}\ind{X_i = 1}\Big] \\
				&= \px{1}\pi(1-\pi) \Ex{1}  + \frac{\px{1}\pi(1-\pi)}{\sqrt{n}} \sum_{i=1}^{n}(A_i-\pi)\Big[\frac{1}{\pi}\sum_{s \in \mathcal{S}} \Econd{\maxZ{1}{1}}{S_i=s}\ind{S_i = s, X_i = 1} \\
                &\hspace{2.4in} + \frac{1}{1-\pi}  \sum_{s \in \mathcal{S}}\Econd{\maxZ{0}{1}}{S_i=s}\ind{S_i = s, X_i = 1}\Big] \\
				&\quad + \frac{\px{1}\pi(1-\pi)}{\sqrt{n}} \sum_{i=1}^{n}\Big[\sum_{s \in \mathcal{S}}\Econd{\maxZ{1}{1}}{S_i=s}\ind{S_i = s, X_i = 1} \\
                &\hspace{1.4in} - \sum_{s \in \mathcal{S}}\Econd{\maxZ{0}{1}}{S_i=s}\ind{S_i = s, X_i = 1}\Big]\\
				&= \px{1}\pi(1-\pi) \Ex{1} + \frac{\px{1}\pi(1-\pi)}{\sqrt{n}} \sum_{i=1}^{n}(A_i-\pi)\Big[\frac{1}{\pi}\sum_{s,X(s)=1} \Econd{\maxZ{1}{1}}{S_i=s}\ind{S_i = s}\\
                &\hspace{2.4in}+ \frac{1}{1-\pi}  \sum_{s,X(s)=1}\Econd{\maxZ{0}{1}}{S_i=s}\ind{S_i = s}\Big] \\
				&\quad + \frac{\px{1}\pi(1-\pi)}{\sqrt{n}} \sum_{i=1}^{n}\Big[\sum_{s,X(s)=1}\Econd{\maxZ{1}{1}}{S_i=s}\ind{S_i = s} \\
                &\hspace{1.4in} - \sum_{s,X(s)=1}\Econd{\maxZ{0}{1}}{S_i=s}\ind{S_i = s}\Big] \\
				&= \px{1}\pi(1-\pi)(\Ex{1} + \Ux{1} + \Sx{1}),
			\end{align*}
   }
		and similarly, 
		\begin{equation*}
			\begin{aligned}
				\frac{1}{\sqrt{n}}\Mx{0} = \px{0}\pi(1-\pi)(\Ex{0} + \Ux{0} + \Sx{0}),
			\end{aligned}
		\end{equation*}
		where	
			\begin{equation}
		\begin{aligned}
			\Ex{x} &=\frac{1}{\sqrt{n}} \sum_{i=1}^{n}\Big\{\frac{1}{\pi}\tildeYa{1}  A_i\ind{X_i = x} - \frac{1}{1-\pi} \tildeYa{0}(1-A_i)\ind{X_i = x}\Big\}, \\	
			\Ux{x} &=  \sum_{s,X(s)=x}\frac{D_n(s)}{\sqrt{n}} \Big[\frac{1}{\pi} \Econd{\maxZ{1}{x}}{S_i = s} + \frac{1}{1-\pi} \Econd{\maxZ{0}{x}}{S_i = s}\Big], \\		
			\Sx{x} &=  \sum_{s,X(s)=x}\sqrt{n}\Big\{\frac{n(s)}{n} - p(s)\Big\} \Big[\Econd{\maxZ{1}{x}}{S_i = s}- \Econd{\maxZ{0}{x}}{S_i = s}\Big], \\
		\end{aligned}
	\label{EUSx}
	\end{equation}
	for $x \in \{1, 0\}$.
	
	To get (\ref{eq:t_M_dist}), by Cramer-Wold device, it suffices to check the following convergence: for all $c_1, c_0 \in \mathbb{R}$, 
			\begin{equation*}
		\begin{aligned}
			&\quad \frac{c_1}{\px{1}}\frac{1}{\sqrt{n}}\Mx{1} + \frac{c_0}{\px{0}}\frac{1}{\sqrt{n}}\Mx{0} \\ &=\pi(1-\pi) \big\{ (c_1 \Ex{1}+c_0  \Ex{0}) + (c_1  \Ux{1}+c_0  \Ux{0}) + (c_1  \Sx{1}+c_0  \Sx{0})\big\}\\
			&\indist\pi(1-\pi) N\Big(0, c^2_1  \big\{\zetaStrY{1} + \zetaStrA{1} + \zetaStrH{1}\big\} \\
            &\hspace{1.2in}+ c^2_0  \big\{\zetaStrY{0} + \zetaStrA{0} + \zetaStrH{0}\big\}\Big). \\
		\end{aligned}
	\end{equation*}
	The convergence can be obtained from Lemma \ref{lemma:t_joint} and the continuous mapping theorem.
	\end{proof}
	\end{lemma}

	\begin{lemma}
		\label{lemma:t_joint}
		For all $c_1, c_0 \in \mathbb{R}$, let $\Ex{c} = c_1\Ex{1}+c_0\Ex{0}$, $\Ux{c} = c_1\Ux{1}+c_0\Ux{0}$, and $\Sx{c} = c_1\Sx{1}+c_0\Sx{0}$, where $\Ex{x}, \Ux{x}$, and $\Sx{x}$, $ x \in \{1, 0\}$, are defined as in (\ref{EUSx}). Then under assumptions of Theorem 1, 
					\begin{equation*}
			\begin{aligned}
				(\Ex{c}, \Ux{c}, \Sx{c}) 
				&\indist (\etax{c}, \upsilonx{c},  \xix{c}), \\				
			\end{aligned}
		\end{equation*}
	where $\etax{c}$, $\upsilonx{c}$, and $\xix{c}$ are independent and satisfy $\etax{c} \sim  N\big(0, c^2_1 \zetaStrY{1} + c^2_0 \zetaStrY{0}\big)$, $\upsilonx{c} \sim  N\big(0, c^2_1 \zetaStrA{1} + c^2_0 \zetaStrA{0}\big)$, and $\xix{c} \sim  N\big(0, c^2_1 \zetaStrH{1} + c^2_0 \zetaStrH{0}\big)$.
	
 \begin{proof}
		This proof is similar to the proof of Lemma B.2 in \cite{Bugni2018}.
		Suppose that independently for each $s \in \mathcal{S}$ and independently of $(A^{(n)}, S^{(n)})$, $\{(\tilde Y_i^s(1), \tilde Y_i^s(0) ): 1 \leq i \leq n\}$ are i.i.d. with marginal distribution equal to the distribution of $(\tilde Y_i(1), \tilde Y_i(0) )\mid S_i = s$. Define $N(s) = \sum_{i=1}^{n}\ind{S_i < s}$ and $F(s) = \Prob{S_i < s}$. Furthermore, define
		\begin{equation}
			\begin{aligned}
				\tilde{E}_{n,c} &= c_1\sum_{s, X(s) = 1} \left[\frac{1}{\sqrt{n}} \sum_{i=n \frac{N(s)}{n}+1}^{n\big\{\frac{N(s)}{n}+\frac{\nas{1}{s}}{n}\big\}} \frac{1}{\pi} \tildeYas{1}-\frac{1}{\sqrt{n}} \sum_{i=n\big\{\frac{N(s)}{n}+\frac{\nas{1}{s}}{n}\big\}+1}^{n\big\{\frac{N(s)}{n}+\frac{n(s)}{n}\big\}} \frac{1}{1-\pi} \tildeYas{0}\right] \\
				& \quad +  c_0 \sum_{s, X(s) = 0} \left[\frac{1}{\sqrt{n}} \sum_{i=n \frac{N(s)}{n}+1}^{n\big\{\frac{N(s)}{n}+\frac{\nas{1}{s}}{n}\big\}} \frac{1}{\pi} \tildeYas{1}-\frac{1}{\sqrt{n}} \sum_{i=n\big\{\frac{N(s)}{n}+\frac{\nas{1}{s}}{n}\big\}+1}^{n\big\{\frac{N(s)}{n}+\frac{n(s)}{n}\big\}} \frac{1}{1-\pi} \tildeYas{0}\right]. \\
			\end{aligned}
		\label{eq:tildeEc}
		\end{equation}
		By construction, $\{\Ex{c}\mid A^{(n)}, S^{(n)}\} \equaldist \{\tilde E_{n,c}\mid A^{(n)}, S^{(n)}\}$, and hence $\Ex{c} \equaldist \tilde E_{n,c}$. Since $\Ux{c}$ and $\Sx{c}$ are functions of $S^{(n)}$ and $A^{(n)}$, we have $(\Ex{c}, \Ux{c}, \Sx{c}) \equaldist (\tilde E_{n,c}, \Ux{c}, \Sx{c})$. Next, we define 	
		\begin{equation*}
			\begin{aligned}
				{E}^*_{n,c} &=  c_1 \sum_{s, X(s) = 1} \left(\frac{1}{\sqrt{n}} \sum_{i=\lfloor nF(s) \rfloor+1}^{\lfloor n\{F(s)+\pi p(s)\} \rfloor} \frac{1}{\pi} \tildeYas{1}-\frac{1}{\sqrt{n}} \sum_{i=\lfloor n\{F(s)+\pi p(s)\} \rfloor+1}^{\lfloor n\{F(s)+p(s)\} \rfloor} \frac{1}{1-\pi} \tildeYas{0}\right) \\
				& \quad +  c_0 \sum_{s, X(s) = 0} \left(\frac{1}{\sqrt{n}} \sum_{i=\lfloor nF(s) \rfloor+1}^{\lfloor n\{F(s)+\pi p(s)\} \rfloor} \frac{1}{\pi} \tildeYas{1}-\frac{1}{\sqrt{n}} \sum_{i=\lfloor n\{F(s)+\pi p(s)\} \rfloor+1}^{\lfloor n\{F(s)+p(s)\} \rfloor} \frac{1}{1-\pi} \tildeYas{0}\right). \\
			\end{aligned}
		\end{equation*}
		Note that ${E}^*_{n,c}$ is a function of $\{(\tilde Y_i^s(1), \tilde Y_i^s(0) ): 1 \leq i \leq n\}$, which is independent of $(A^{(n)}, S^{(n)})$. Moreover, by noticing $({U}_{n,c}, {S}_{n,c})$ is a function of $(A^{(n)}, S^{(n)})$, we can conclude that ${E}^*_{n,c} \Perp ({U}_{n,c}, {S}_{n,c})$.
		
		We next show
			\begin{equation}
			\begin{aligned}
				(\Ex{c}, \Ux{c}, \Sx{c}) \equaldist (E^*_{n,c}, \Ux{c}, \Sx{c}) + o_P(1).
			\end{aligned}
		\label{eq:t_approx}
		\end{equation}
	It suffices to show $\Delta_{n,c} =  \tilde{E}_{n,c} - {E}^*_{n,c} \inp 0$. For each $s \in \mathcal{S}$, define 
	\begin{equation*}
		\begin{aligned}
			g_n^s(u) = \frac{1}{\sqrt{n}} \sum_{i=1}^{\lfloor n u \rfloor} \tildeYas{1}.
		\end{aligned}
	\end{equation*}
	Since $\{(\tilde Y_i^s(1), \tilde Y_i^s(0) ): 1 \leq i \leq n\}$ are i.i.d. with mean $0$ and finite variance, $g_n^s(u)$ converges weakly to Brownian motion. Moreover, 
	\begin{equation*}
		\begin{aligned}
			\left\{ \frac{N(s)}{n}, \frac{\nas{1}{s}}{n}\right\} \inp \left\{ F(s), \pi p(s)\right\},
		\end{aligned}
	\end{equation*}
	we can conclude that 
		\begin{equation*}
		\begin{aligned}
			g_n^s\left(\frac{N(s)+\nas{1}{s}}{n}\right) - g_n^s\left(\frac{N(s)}{n}\right)-\Big\{g_n^s\big(F(s)+\pi p(s)\big)-g_n^s\big(F(s)\big)\Big\} \inp 0,
		\end{aligned}
	\end{equation*}
	where the convergence follows from elementary properties of Brownian motion and the continuous mapping theorem. 
	Similarly arguing for $\tildeYas{0}$, we can prove that $\Delta_{n,c} =  \tilde{E}_{n,c} - {E}^*_{n,c} \inp 0$.
	
	We proceed to show that ${E}^*_{n,c} \indist \etax{c}$. By elementary properties of Brownian motion, we have
	\begin{equation*}
		\begin{aligned}
			\frac{1}{\sqrt{n}} \sum_{i=\lfloor nF(s) \rfloor+1}^{\lfloor n\{F(s)+\pi p(s)\} \rfloor} \frac{1}{\pi} \tildeYas{1} \indist N\left(0, \frac{p(s)\sigma^2_{\tilde Y(1)}(s)}{\pi}\right),
		\end{aligned}
	\end{equation*}
	where we write $\Var{\tildeYas{1}} = \sigma^2_{\tilde Y(1)}(s)$.
	Repeating the similar steps for $\tildeYas{1}$ and $\tildeYas{0}$ for each $s \in \mathcal{S}$ and using the independence of $\{(\tilde Y_i^s(1), \tilde Y_i^s(0) ): 1 \leq i \leq n\}$ across both $i$ and $s$, we can conclude that $E^*_{n,c} \indist \etax{c}$.
	
	By Assumption 3
 and the continuous mapping theorem,
	\begin{equation}
		\begin{aligned}
			\{\Ux{c}\mid S^{(n)}\} \indist  N\big(0,c_1^2\zetaStrA{1}+ c_0^2\zetaStrA{0}\big) \text{ a.s. }
		\end{aligned}
	\label{eq:Uconvergent}
	\end{equation}
	By the central limit theorem and the continuous mapping theorem, 	
	\begin{equation}
		\begin{aligned}
			\Sx{c} \indist  N\big(0,c_1^2\zetaStrH{1} + c_0^2\zetaStrH{0}\big).
		\end{aligned}
		\label{eq:Sconvergent}
	\end{equation}
	We are left with showing $(\Ex{c}, \Ux{c}, \Sx{c}) \indist (\etax{c}, \upsilonx{c}, \xix{c})$,  where  $\etax{c} \equaldist N\big(0,c_1^2\zetaStrY{1} + c_0^2\zetaStrY{0}\big)$, $\upsilonx{c} \equaldist  N\big(0,c_1^2\zetaStrA{1} + c_0^2\zetaStrA{0}\big)$,  $\xix{c} \equaldist N\big(0,c_1^2\zetaStrH{1} + c_0^2\zetaStrH{0}\big)$, and that $\{\etax{c}, \upsilonx{c}, \xix{c}\}$ are independent. From (\ref{eq:t_approx}), it suffices to show $(E^*_{n,c}, \Ux{c}, \Sx{c}) \indist (\etax{c}, \upsilonx{c}, \xix{c})$, which is equivalent to 
	\begin{equation*}
		\begin{aligned}
		\Prob{E^*_{n,c} \leq h_1}\Prob{\Ux{c} \leq h_2, \Sx{c} \leq h_3}  {\rightarrow} \Prob{\etax{c} \leq h_1}\Prob{\upsilonx{c} \leq h_2} \Prob{\xix{c} \leq h_3},
		\end{aligned}
	\end{equation*}
	for any $h = (h_1, h_2, h_3) \in \mathbb{R}^3$ s.t. $\Prob{\eta_1 \leq h_1}\Prob{\upsilon_{1} \leq h_2} \Prob{\xi_{1} \leq h_3}$ is continuous.
	
	We first assume that $\Prob{\etax{c} \leq h_1}$, $\Prob{\upsilonx{c} \leq h_2}$, and $\Prob{\xix{c} \leq h_3}$ are continuous at $h_1$, $h_2$, and $h_3$, respectively. Then $E^*_{n,c} \indist\etax{c}$ implies $\Prob{E^*_{n,c}\leq h_1} \rightarrow \Prob{\etax{c} \leq h_1}$. Moreover,	
	\begin{equation*}
		\begin{aligned}
			&\quad \Prob{\Ux{c} \leq h_2,\Sx{c}  \leq h_3} \\
			&= E[\ProbB{\Ux{c} \leq h_2, \Sx{c} \leq h_3\mid S^{(n)}}] \\
			&= E[\ProbB{\Ux{c} \leq h_2\mid S^{(n)}}\ind{\Sx{c} \leq h_3}] \\
			&= E\big([\ProbB{\Ux{c} \leq h_2\mid S^{(n)}} - \Prob{\upsilonx{c} \leq h_2}]\ind{\Sx{c} \leq h_3}\big)  +  E\{\Prob{\upsilonx{c} \leq h_2}\ind{\Sx{c} \leq h_3}\} \\
			&= E\big([\ProbB{\Ux{c} \leq h_2\mid S^{(n)}} - \Prob{\upsilonx{c} \leq h_2}]\ind{\Sx{c} \leq h_3}\big) +   \Prob{\upsilonx{c} \leq h_2}\Prob{ \Sx{c} \leq h_3} \\
			&\rightarrow \Prob{\upsilonx{c} \leq h_2} \Prob{\xix{c} \leq h_3},
		\end{aligned}
	\end{equation*}
	where the convergence follows from the dominated convergence theorem, (\ref{eq:Uconvergent}), and (\ref{eq:Sconvergent}).
	
	If any of $\Prob{\etax{c} \leq h_1}$, $\Prob{\upsilonx{c} \leq h_2}$, $\Prob{\xix{c} \leq h_3}$ is discontinuous at $h_j$. Since $\etax{c}$, $\upsilonx{c}$ and $\xix{c}$ are all mean 0 normally distributed, any one of which distribution function is discontinuous at $h_j$, must be degenerate and equal to 0 and thus $h_j = 0$. Without loss of generality, we assume that $\Prob{\upsilonx{c} \leq h_2}$ is discontinuous at $h_2 = 0$ and $\upsilonx{c}$ is degenerate and equal to $0$. For example, in the stratified biased-coin design when $q(s)  = 0$ for each $s \in \mathcal{S}$. Then if $\Prob{\etax{c} \leq h_1}\Prob{\upsilonx{c} \leq h_2} \Prob{\xix{c} \leq h_3}$ is continuous at $(h_1, h_2, h_3) = (h_1, 0, h_3)$, $\Prob{\etax{c} \leq h_1}\Prob{\xix{c} \leq h_3}$ must be equal to 0, which implies that $\etax{c}$ or $\xix{c}$ is also degenerate and equal to $0$ and that its $h_k < 0$. Without loss of generality, we can assume that $\xix{c}$ is degenerate and $h_3 < 0$. It follows that 
	\begin{equation*}
		\begin{aligned}
			&\quad \Prob{E^*_{n,c} \leq h_1, \Ux{c} \leq h_2, \Sx{c} \leq h_3}\\ &\leq \Prob{\Sx{c} \leq h_3} \\
			& \rightarrow 0 \\
			&=  \Prob{\etax{c} \leq h_1}\Prob{\upsilonx{c} \leq h_2} \Prob{\xix{c} \leq h_3}.
		\end{aligned}
	\end{equation*}
     Thus we establish $(E^*_{n,c}, \Ux{c}, \Sx{c}) \indist (\etax{c}, \upsilonx{c}, \xix{c})$, and $(\Ex{c}, \Ux{c}, \Sx{c}) \indist (\etax{c}, \upsilonx{c}, \xix{c})$ follows.
		\end{proof}
		\end{lemma}
	
		\begin{lemma}
		\label{lemma:thm_strata_joint}
		Under the assumptions of Theorem 4,
\begin{equation*}
	\begin{aligned}
	 {\begin{pmatrix}
				\Lx{1} + \Sx{1} \\
				\Lx{0} + \Sx{0} \\
		\end{pmatrix}}
		\indist
		N\left({\begin{pmatrix}
				0 \\
				0 \\
		\end{pmatrix}}, {\begin{pmatrix}
				p^2_{1}u^2_{1} & 0 \\
				0 & p^2_{0}u^2_{0}\\
		\end{pmatrix}} \right),
	\end{aligned}
\end{equation*}
where $u^2_{x}$, $\Lx{x}$, $\Sx{x}$, $ x \in \{1, 0\}$, are defined as in Theorem 4, (\ref{Lx}), and (\ref{EUSx}), respectively.
		\begin{proof}
			By Cramer-Wold device, it suffices to check that for all
			$c_1, c_0 \in \mathbb{R}$, 		
			\begin{equation*}
				\begin{aligned}
					&\quad c_1(\Lx{1} + \Sx{1}) + c_0(\Lx{0} + \Sx{0}) \\
					&\indist N\Big(0, c^2_1  \big\{\zetaStrY{1}  + \zetaStrH{1}\big\} + c^2_0  \big\{\zetaStrY{0} +  \zetaStrH{0}\big\}\Big). \\
				\end{aligned}
			\end{equation*}
			Let $\Lx{c} = c_1\Lx{1}+c_0\Lx{0}$ and $\Sx{c} = c_1\Sx{1}+c_0\Sx{0}$.
			We will show
			\begin{equation}
				\begin{aligned}
					(\Lx{c}, \Sx{c}) 
					&\indist (\etax{c},  \xix{c}), \\				
				\end{aligned}
			\label{eq:strata_LS}
			\end{equation}
			where $\etax{c}$ and $\xix{c}$ are independent and satisfy $\etax{c} \sim  N\big(0, c^2_1 \zetaStrY{1} + c^2_0 \zetaStrY{0}\big)$ and $\xix{c} \sim  N\big(0, c^2_1 \zetaStrH{1} + c^2_0 \zetaStrH{0}\big)$. Then by the  continuous mapping theorem, we finish the proof.
			
			Define
					\begin{equation*}
			\begin{aligned}
				\tilde{L}_{n,c} &= c_1 \sum_{s,X(s)=1}\left( \frac{1}{\sqrt{n}}\frac{n(s)}{\nas{1}{s}} \sum_{i=n \frac{N(s)}{n}+1}^{n\big\{\frac{N(s)}{n}+\frac{\nas{1}{s}}{n}\big\}}\tildeYas{1} -  \frac{1}{\sqrt{n}}\frac{n(s)}{\nas{0}{s}} \sum_{i=n\big\{\frac{N(s)}{n}+\frac{\nas{1}{s}}{n}\big\}+1}^{n\big\{\frac{N(s)}{n}+\frac{n(s)}{n}\big\}}\tildeYas{0}\right) \\
				&\quad + c_0 \sum_{s,X(s)=0}\left( \frac{1}{\sqrt{n}}\frac{n(s)}{\nas{1}{s}} \sum_{i=n \frac{N(s)}{n}+1}^{n\big\{\frac{N(s)}{n}+\frac{\nas{1}{s}}{n}\big\}}\tildeYas{1}  -  \frac{1}{\sqrt{n}}\frac{n(s)}{\nas{0}{s}} \sum_{i=n\big\{\frac{N(s)}{n}+\frac{\nas{1}{s}}{n}\big\}+1}^{n\big\{\frac{N(s)}{n}+\frac{n(s)}{n}\big\}}\tildeYas{0} \right)\\
				&= c_1 \sum_{s,X(s)=1}\left( \frac{1}{\sqrt{n}} \sum_{i=n \frac{N(s)}{n}+1}^{n\big\{\frac{N(s)}{n}+\frac{\nas{1}{s}}{n}\big\}}\frac{1}{\pi}\tildeYas{1} -  \frac{1}{\sqrt{n}} \sum_{i=n\big\{\frac{N(s)}{n}+\frac{\nas{1}{s}}{n}\big\}+1}^{n\big\{\frac{N(s)}{n}+\frac{n(s)}{n}\big\}}\frac{1}{1-\pi}\tildeYas{0}\right) \\
				&\quad + c_0 \sum_{s,X(s)=0}\left( \frac{1}{\sqrt{n}} \sum_{i=n \frac{N(s)}{n}+1}^{n\big\{\frac{N(s)}{n}+\frac{\nas{1}{s}}{n}\big\}}\frac{1}{\pi}\tildeYas{1} -  \frac{1}{\sqrt{n}} \sum_{i=n\big\{\frac{N(s)}{n}+\frac{\nas{1}{s}}{n}\big\}+1}^{n\big\{\frac{N(s)}{n}+\frac{n(s)}{n}\big\}}\frac{1}{1-\pi}\tildeYas{0}\right)  + o_P(1)\\
				&= \tilde{E}_{n,c} + o_P(1),
			\end{aligned}
		\end{equation*}
		where $\tilde{E}_{n,c}$ is defined as in (\ref{eq:tildeEc}).
		By construction, $\{\Lx{c}\mid A^{(n)}, S^{(n)}\} \equaldist \{\tilde L_{n,c}\mid A^{(n)}, S^{(n)}\}$, and hence $\Lx{c} \equaldist \tilde L_{n,c} $. Since $\Sx{c}$ is a function of $A^{(n)}$ and $S^{(n)}$, we have $(\Lx{c}, \Sx{c}) \equaldist (\tilde L_{n,c}, \Sx{c}) $. Moreover, 	$\tilde L_{n,c} = \tilde E_{n,c} + o_P(1) = E^*_{n,c} + o_P(1)$ and 
		$(\tilde L_{n,c}, \Sx{c}) = ( \tilde E_{n,c}, \Sx{c}) + o_P(1)= ( E^*_{n,c}, \Sx{c}) + o_P(1)$. Recall the proof of Lemma \ref{lemma:t_joint} and we can derive (\ref{eq:strata_LS}).
		\end{proof}
			\end{lemma}

			\begin{lemma}
			\label{lemma:thm_add_M}
			Let $\Mx{1}$ and $\Mx{0}$ be defined as in (\ref{RZMx}). Then under the assumptions of Theorem 6, 
	\begin{equation}
		\begin{aligned}
			\frac{1}{\sqrt{n}} {\begin{pmatrix}
					\Mx{1} \\
					\Mx{0} \\
			\end{pmatrix}}
			\indist
			\pi(1-\pi)
			N\left({\begin{pmatrix}
					0 \\
					0 \\
			\end{pmatrix}}, {\begin{pmatrix}
					\px{1}^4 w^2_{1} & \px{1}^2\px{0}^2w_{10} \\
					\px{1}^2\px{0}^2w_{10} & \px{0}^4 w^2_{0} \\
			\end{pmatrix}} \right),
		\end{aligned}
\label{eq:t_add_M_dist}
\end{equation}
where $w^2_{1}$, $w^2_{0}$, and $w_{10}$ are defined as in Theorem 6.

	\begin{proof}
	This proof is inspired by the transformed outcomes method used by \cite{ma2022}.
	For $a \in \{1, 0\}$ and $x \in \{1, 0\}$, let $\rax{a}{x} = \big(\Ya - \mu_{ax}\big)\ind{X_i = x}$, $\tilderax{a}{x} = \rax{a}{x} - \Econd{\rax{a}{x}}{ S_i}$, and $\raxZ{a}{x} =  \Econd{\rax{a}{x}}{ Z_i}$. Then 
	\begin{equation*}
		\begin{aligned}
			\{\Ya - \mu_{ax}\}\ind{X_i = x} = \tilderax{a}{x} + \Econd{\raxZ{a}{x}}{S_i}.
		\end{aligned}
	\end{equation*}		
				\begin{equation*}
					\begin{aligned}
						&\quad \frac{1}{\sqrt{n}}\Mx{1} \\
						&= \frac{\px{1}\pi(1-\pi)}{\sqrt{n}} \sum_{i=1}^{n}\Big[\frac{1}{\pi}\{\Yt - \muax{1}{1}\}A_i\ind{X_i = 1} - \frac{1}{1-\pi} \{\Yc - \muax{0}{1}\}(1-A_i)\ind{X_i = 1}\Big] \\
						&= \frac{\px{1}\pi(1-\pi)}{\sqrt{n}} \sum_{i=1}^{n}\Big\{\frac{1}{\pi}\tilderax{1}{1} A_i - \frac{1}{1-\pi} \tilderax{0}{1}(1-A_i)\Big\} \\
						&\quad + \frac{\px{1}\pi(1-\pi)}{\sqrt{n}} \sum_{i=1}^{n}\Big[\frac{1}{\pi} \Econd{\raxZ{1}{1}}{ S_i}A_i - \frac{1}{1-\pi}  \Econd{\raxZ{0}{1}}{ S_i}(1-A_i)\Big] \\
						&= \px{1}\pi(1-\pi)\Erx{1} + \frac{\px{1}\pi(1-\pi)}{\sqrt{n}} \sum_{i=1}^{n}(A_i-\pi)\Big[\frac{1}{\pi} \Econd{\raxZ{1}{1}}{ S_i} + \frac{1}{1-\pi}  \Econd{\raxZ{0}{1}}{ S_i}\Big] \\
						&\quad + \frac{\px{1}\pi(1-\pi)}{\sqrt{n}} \sum_{i=1}^{n}\Big[\Econd{\raxZ{1}{1}}{ S_i} - \Econd{\raxZ{0}{1}}{ S_i}\Big]\\
						&= \px{1}\pi(1-\pi)\Erx{1} + \frac{\px{1}\pi(1-\pi)}{\sqrt{n}} \sum_{i=1}^{n}(A_i-\pi)\Big[\frac{1}{\pi}\sum_{s \in \mathcal{S}} \Econd{\raxZ{1}{1}}{ S_i=s}\ind{S_i = s} \\
                        &\hspace{3in}
						 + \frac{1}{1-\pi}  \sum_{s \in \mathcal{S}}\Econd{\raxZ{0}{1}}{ S_i=s}\ind{S_i = s}\Big] \\
						&\quad + \frac{\px{1}\pi(1-\pi)}{\sqrt{n}} \sum_{i=1}^{n}\Big[\sum_{s \in \mathcal{S}}\Econd{\raxZ{1}{1}}{ S_i=s}\ind{S_i = s} - \sum_{s \in \mathcal{S}}\Econd{\raxZ{0}{1}}{ S_i=s}\ind{S_i = s}\Big] \\
						&=\px{1}\pi(1-\pi) (\Erx{1} +\Urx{1} + \Srx{1}),
					\end{aligned}
				\end{equation*}
					and similarly, 
			\begin{equation*}
				\begin{aligned}
					\frac{1}{\sqrt{n}}\Mx{0} = \px{0}\pi(1-\pi)(\Erx{0} +\Urx{0} + \Srx{0}),
				\end{aligned}
			\end{equation*}
		where
		\begin{equation}
			\begin{aligned}
				\Erx{x} &= \frac{1}{\sqrt{n}} \sum_{i=1}^{n}\Big\{\frac{1}{\pi}\tilderax{1}{x} A_i - \frac{1}{1-\pi} \tilderax{0}{x}(1-A_i)\Big\},\\				
				\Urx{x} &=  \sum_{s \in \mathcal{S}}\frac{D_{n}(s)}{\sqrt{n}} \Big[\frac{1}{\pi} \Econd{\raxZ{1}{x}}{S_i=s} + \frac{1}{1-\pi} \Econd{\raxZ{0}{x}}{S_i=s}\Big],\\		
				\Srx{x} &= \sum_{s \in \mathcal{S}}\sqrt{n}\Big\{\frac{n(s)}{n} - p(s)\Big\}\Big[\Econd{\raxZ{1}{x}}{S_i=s}- \Econd{\raxZ{0}{x}}{S_i=s}\Big],
			\end{aligned}
		\label{EUSrx}
		\end{equation}
	for $x \in \{1, 0\}$.
			To get (\ref{eq:t_add_M_dist}), by Cramer-Wold device, it suffices to check that
		for all $c_1, c_0 \in \mathbb{R}$, 
		\begin{equation*}
			\begin{aligned}
				&\quad \frac{c_1}{\px{1}}\frac{1}{\sqrt{n}}\Mx{1} + \frac{c_0}{\px{0}}\frac{1}{\sqrt{n}}\Mx{0}\\ &=\pi(1-\pi) \big\{ (c_1 \Erx{1}+c_0 \Erx{0}) + (c_1 \Urx{1}+c_0 \Urx{0}) + (c_1  \Srx{1}+c_0 \Srx{0})\big\}\\
				&\indist\pi(1-\pi) N\Big(0, c^2_1  \big\{\zetaAddY{1} + \zetaAddA{1} + \zetaAddH{1}\big\} + c^2_0  \big\{\zetaAddY{0} + \zetaAddA{0} + \zetaAddH{0}\big\} \\
                &\hspace{1.2in} + 2c_1c_0\big\{\zetaAddYxy{1}{0}+ \zetaAddAxy{1}{0}+ \zetaAddHxy{1}{0}\big\}\Big). \\
			\end{aligned}
		\end{equation*}
		The above convergence can be obtained from Lemma \ref{lemma:t_add_joint} and the continuous mapping theorem.
			\end{proof}
		\end{lemma}
	
		\begin{lemma}
		\label{lemma:t_add_joint}
		For all $c_1, c_0 \in \mathbb{R}$, let $\Erx{c} = c_1\Erx{1}+c_0\Erx{0}$, $\Urx{c} = c_1\Urx{1}+c_0\Urx{0}$, and $\Srx{c} = c_1\Srx{1}+c_0\Srx{0}$, where $\Erx{x}, \Urx{x}$, and $\Srx{x}$, $ x \in \{1, 0\}$, are defined as in (\ref{EUSrx}). Then under the assumptions of Theorem 6, 
		\begin{equation}
			\begin{aligned}
				(\Erx{c}, \Urx{c}, \Srx{c}) 
				&\indist (\etarx{c}, \upsilonrx{c}, \xirx{c}), \\				
			\end{aligned}
		\end{equation}
		where $\etarx{c}$, $\upsilonrx{c}$, and $\xirx{c}$ are independent random variables satisfying $\etarx{c} \sim  N\big(0, c^2_1 \zetaAddY{1} + c^2_0 \zetaAddY{0} +2c_1c_0\zetaAddYxy{1}{0}\big)$, $\upsilonrx{c} \sim  N\big(0, c^2_1 \zetaAddA{1} + c^2_0 \zetaAddA{0}+2c_1c_0\zetaAddAxy{1}{0}\big)$, and $\xirx{c} \sim  N\big(0, c^2_1 \zetaAddH{1} + c^2_0 \zetaAddH{0}+2c_1c_0\zetaAddHxy{1}{0}\big)$.
  
		\begin{proof}
			Suppose that independently for each $s \in \mathcal{S}$ and independently of $(A^{(n)}, S^{(n)})$, $\{(\tilderaxs{1}{1}, \tilderaxs{0}{1},\\ \tilderaxs{1}{0},\tilderaxs{0}{0}): 1 \leq i \leq n\}$ are i.i.d. with marginal distribution equal to the distribution of $(\tilderax{1}{1}, \tilderax{0}{1}, \tilderax{1}{0},\tilderax{0}{0})\mid S_i = s$. Define $N(s) = \sum_{i=1}^{n}\ind{S_i < s}$ and $F(s) = \Prob{S_i < s}$. Furthermore, define		
					\begin{equation*}
				\begin{aligned}
					\tilde{E}_{n,r_c} 
					&= \sum_{s \in \mathcal{S}} \Bigg[\frac{1}{\sqrt{n}} \sum_{i=n \frac{N(s)}{n}+1}^{n\big\{\frac{N(s)}{n}+\frac{\nas{1}{s}}{n}\big\}} \frac{1}{\pi}  \Big\{c_1\tilderaxs{1}{1} + c_0\tilderaxs{1}{0}\Big\}\\
                    &\hspace{0.6in} -\frac{1}{\sqrt{n}} \sum_{i=n\big\{\frac{N(s)}{n}+\frac{\nas{1}{s}}{n}\big\}+1}^{n\big\{\frac{N(s)}{n}+\frac{n(s)}{n}\big\}} \frac{1}{1-\pi}\Big\{c_1\tilderaxs{0}{1} + c_0\tilderaxs{0}{0}\Big\}\Bigg] \\
				\end{aligned}
			\end{equation*}
		
				By construction, $\{\Erx{c}\mid A^{(n)}, S^{(n)}\} \equaldist \{\tilde{E}_{n,r_c}\mid A^{(n)}, S^{(n)}\}$, and hence $\Erx{c} \equaldist \tilde{E}_{n,r_c} $. Since $\Urx{c}$ and $\Srx{c}$ are functions of $A^{(n)}$ and $S^{(n)}$, we have $(\Erx{c},\Urx{c}, \Srx{c}) \equaldist (\tilde{E}_{n,r_c}, \Urx{c}, \Srx{c})$. Next, we define 
		\begin{equation*}
			\begin{aligned}
				{E}^*_{n,r_c} 
				&=  \sum_{s \in \mathcal{S}} \Bigg[\frac{1}{\sqrt{n}} \sum_{i=\lfloor nF(s) \rfloor+1}^{\lfloor n\{F(s)+\pi p(s) \} \rfloor} \frac{1}{\pi} \Big\{c_1\tilderaxs{1}{1} + c_0\tilderaxs{1}{0}\Big\}\\
                &\hspace{0.6in}-\frac{1}{\sqrt{n}}\sum_{i=\lfloor n\{F(s)+\pi p(s) \} \rfloor+1}^{\lfloor n\{F(s)+p(s)\} \rfloor} \frac{1}{1-\pi}\Big\{c_1\tilderaxs{0}{1} + c_0\tilderaxs{0}{0}\Big\}\Bigg]. \\
			\end{aligned}
		\end{equation*}
	
		Note that ${E}^*_{n,r_c} $ is a function of $\{(\tilderaxs{1}{1}, \tilderaxs{0}{1}, \tilderaxs{1}{0},\tilderaxs{0}{0}): 1 \leq i \leq n\}$, which is independent of $(A^{(n)}, S^{(n)})$. Moreover, by noticing $(\Urx{c}, \Srx{c})$ is a function of $(A^{(n)}, S^{(n)})$, we can conclude that ${E}^*_{n,r_c} \Perp (\Urx{c}, \Srx{c})$.
		
		We next show
		\begin{equation}
			\begin{aligned}
				(\Erx{c}, \Urx{c}, \Srx{c}) \equaldist ({E}^*_{n,r_c} , \Urx{c}, \Srx{c}) + o_P(1).
			\end{aligned}
			\label{eq:t_add_approx}
		\end{equation}
		It suffices to show $\Delta_{n,r_c} =  \tilde{E}_{n,r_c} - {E}^*_{n,r_c} \inp 0$, which can be deduced by a similar approach as employed in the proof of Lemma \ref{lemma:t_joint}. Moreover, we can prove	
			\begin{equation*}
			\begin{aligned}
				\frac{1}{\sqrt{n}} \sum_{i=\lfloor nF(s) \rfloor+1}^{\lfloor n\{F(s)+\pi p(s) \} \rfloor} \frac{1}{\pi}\Big\{c_1\tilderaxs{1}{1} + c_0\tilderaxs{1}{0}\Big\} \indist N\left(0, \frac{p(s)\sigma^2_{c_1\tilde r_{1}(1) + c_0\tilde r_{0}(1) }(s)}{\pi}\right).
			\end{aligned}
		\end{equation*}
		where we write $\Var{{c_1\tilderaxs{1}{1} + c_0\tilderaxs{1}{0}} } = \sigma^2_{c_1\tilde r_{1}(1) + c_0\tilde r_{0}(1) }(s)$.
	Repeating the similar steps for $\{c_1\tilderaxs{1}{1} + c_0\tilderaxs{1}{0}\}$ and $\{c_1\tilderaxs{0}{1} + c_0\tilderaxs{0}{0}\}$ for each $s \in \mathcal{S}$ and using the independence of $\{(\tilderaxs{1}{1}, \tilderaxs{0}{1}, \tilderaxs{1}{0},\tilderaxs{0}{0}): 1 \leq i \leq n\}$ across both $i$ and $s$, we obtain ${E}^*_{n,r_c} \indist  N\big(0,({1}/{\pi})\sigma^2_{c_1\tilde r_{1}(1) + c_0\tilde r_{0}(1) } + \{{1}/(1-\pi)\}\sigma^2_{c_1\tilde r_{1}(0) + c_0\tilde r_{0}(0) }\big)$. To see ${E}^*_{n,r_c} \indist \etarx{c}$, notice that
				\begin{equation*}
		\begin{aligned}
			&\quad \frac{1}{\pi}\sigma^2_{c_1\tilde r_{1}(1) + c_0\tilde r_{0}(1) } + \frac{1}{1-\pi}\sigma^2_{c_1\tilde r_{1}(0) + c_0\tilde r_{0}(0) }\\
			&= \frac{1}{\pi}\left[c_1^2\sigma^2_{\tilde{r}_{1}(1)} + c_0^2\sigma^2_{\tilde{r}_{0}(1)} + 2c_1c_0\E{\tilderax{1}{1}\tilderax{1}{0}}\right] \\
            &\quad + \frac{1}{1-\pi}\left[c_1^2\sigma^2_{\tilde{r}_{1}(0)} + c_0^2\sigma^2_{\tilde{r}_{0}(0)} + 2c_1c_0\E{\tilderax{0}{1}\tilderax{0}{0}}\right]\\
			&= c_1^2\zetaAddY{1} + c_0^2\zetaAddY{0} + 2c_1c_0\left[\frac{1}{\pi}\E{\tilderax{1}{1}\tilderax{1}{0}} + \frac{1}{1-\pi}\E{\tilderax{0}{1}\tilderax{0}{0}}\right]\\
			&= c_1^2\zetaAddY{1} + c_0^2\zetaAddY{0} + 2c_1c_0\zetaAddYxy{1}{0}, \\
		\end{aligned}
	\end{equation*}
	where the last equality follows from
		\begin{equation*}
		\begin{aligned}
			&\quad \E{\tilderax{1}{1}\tilderax{1}{0}} \\
			&= 	E\Big(\Big[\{\Yt - \muax{1}{1}\}\ind{X_i = 1} - \Econd{\raxZ{1}{1}}{S_i}\Big]\Big[\{\Yt - \muax{1}{0}\}\ind{X_i = 0} - \Econd{\raxZ{1}{0}}{S_i}\Big]\Big)\\
			&= E\Big[-\{\Yt - \muax{1}{1}\}\ind{X_i = 1}\Econd{\raxZ{1}{0}}{ S_i} - \{\Yt - \muax{1}{0}\}\ind{X_i = 0}\Econd{\raxZ{1}{1}}{ S_i}\\
            &\quad \qquad  + \Econd{\raxZ{1}{1}}{ S_i}\Econd{\raxZ{1}{0}}{ S_i}\Big] \\ 
			&= -E\Big[\Econd{\raxZ{1}{1}}{ S_i}\Econd{\raxZ{1}{0}}{ S_i}\Big] \\ 
			&= -\sum_{s \in \mathcal{S}}  p(s)\Econd{\raxZ{1}{1}}{ S_i = s}\Econd{\raxZ{1}{0}}{ S_i = s}
		\end{aligned}
	\end{equation*}
	and similarly
	\begin{equation*}
		\begin{aligned}
			&\quad \E{\tilderax{0}{1}\tilderax{0}{0}} = -\sum_{s \in \mathcal{S}}  p(s)\Econd{\raxZ{0}{1}}{ S_i = s}\Econd{\raxZ{0}{0}}{ S_i = s}.
		\end{aligned}
	\end{equation*}
	Moreover,
	\begin{equation*}
		\begin{aligned}
		\Urx{c} &=  \sum_{s \in \mathcal{S}}\frac{D_{n}(s)}{\sqrt{n}} \Big(c_1\Big[\frac{1}{\pi} \Econd{\raxZ{1}{1}}{ S_i=s} + \frac{1}{1-\pi} \Econd{\raxZ{0}{1}}{ S_i=s}\Big]\\
		 &\hspace{1in} + c_0\Big[\frac{1}{\pi} \Econd{\raxZ{1}{0}}{ S_i = s} + \frac{1}{1-\pi} \Econd{\raxZ{0}{0}}{ S_i = s}\Big] \Big).
		\end{aligned}
	\end{equation*}
	By Assumption 3
 and the continuous mapping theorem,	
	\begin{equation*}
		\begin{aligned}
			\{\Urx{c}\mid S^{(n)}\} \indist  N\big(0,c_1^2\zetaAddA{1}+ c_0^2\zetaAddA{0} + 2c_1c_0\zetaAddAxy{1}{0}\big) \text{ a.s.}
		\end{aligned}
	\end{equation*}
    Also, note that 
		\begin{equation*}
		\begin{aligned}
		\Srx{c} &= \sum_{s \in \mathcal{S}}\sqrt{n}\Big\{\frac{n(s)}{n} - p(s)\Big\} \Big(c_1\big[ \Econd{\raxZ{1}{1}}{ S_i=s}- \Econd{\raxZ{0}{1}}{ S_i=s}\big]\\  
		&+ c_0\big[\Econd{\raxZ{1}{0}}{ S_i = s}- \Econd{\raxZ{0}{0}}{ S_i = s}\big]\Big).
		\end{aligned}
	\end{equation*}
	Then, by the central limit theorem and the continuous mapping theorem, 
	\begin{equation*}
		\begin{aligned}
			\Srx{c} \indist  N(0, c_1^2\zetaAddH{1}+ c_0^2\zetaAddH{0} + 2c_1c_0\zetaAddHxy{1}{0}).
		\end{aligned}
	\end{equation*}
	
	We are left with showing $(\Erx{c}, \Urx{c}, \Srx{c}) \indist (\etarx{c}, \upsilonrx{c}, \xirx{c})$,  where  $\etarx{c} \equaldist N\big(0,c_1^2\zetaAddY{1} + c_0^2\zetaAddY{0}+ 2c_1c_0\zetaAddYxy{1}{0}\big)$, $\upsilonrx{c} \equaldist N\big(0,c_1^2\zetaAddA{1} + c_0^2\zetaAddA{0} + 2c_1c_0\zetaAddAxy{1}{0}\big)$, and $\xirx{c} \equaldist N\big(0,c_1^2\zetaAddH{1} + c_0^2\zetaAddH{0}+ 2c_1c_0\zetaAddHxy{1}{0}\big)$, and $\etarx{c}, \upsilonrx{c}, \xirx{c}$ are independent. This can be shown in the same way as that in the proof of Lemma \ref{lemma:t_joint}.
		\end{proof}
	\end{lemma}

	\begin{lemma}
	\label{lemma:thm_strata_joint_add}
	Under the assumptions of Theorem 9, 
	\begin{equation}
		\begin{aligned}
			{\begin{pmatrix}
					\Lcheckx{1} + \Scheckx{1}\\
					\Lcheckx{0} + \Scheckx{0}\\
			\end{pmatrix}}
			\indist
			N\left({\begin{pmatrix}
					0 \\
					0 \\
			\end{pmatrix}}, {\begin{pmatrix}
					\px{1}^2 s^2_{1}& 0 \\
					0 & \px{0}^2 s^2_{0}\\
			\end{pmatrix}} \right),
		\end{aligned}
\end{equation}
where $\Lcheckx{x}, \Scheckx{x}$ , $ x \in \{1, 0\}$, are defined as in (\ref{LSyx}).
	\begin{proof}
		By Cramer-Wold device, it suffices to check that for all $c_1, c_0 \in \mathbb{R}$, 
		\begin{equation*}
			\begin{aligned}
				&\quad c_1(\Lcheckx{1} + \Scheckx{1}) + c_0(\Lcheckx{0} + \Scheckx{0}) \\
				&\indist N\Big(0, c^2_1  \big\{\zetaCheckY{1} + \zetaCheckH{1}\big\} +  c^2_0 \big\{\zetaCheckY{0} +  \zetaCheckH{0}\big\}\Big). \\
			\end{aligned}
		\end{equation*}
		Let $\Lcheckx{c}= c_1\Lcheckx{1}+c_0\Lcheckx{0}$ and $\Scheckx{c}= c_1\Scheckx{1}+c_0\Scheckx{0}$.
		We next show
		\begin{equation}
			\begin{aligned}
				(\Lcheckx{c}, \Scheckx{c}) 
				&\indist (\etacheckx{c},  \xicheckx{c}), \\			
			\end{aligned}
			\label{eq:strata_LS_add}
		\end{equation}
		where $\etacheckx{c}$ and $\xicheckx{c}$ are independent and satisfy $\etacheckx{c} \sim  N\big(0, c^2_1 \zetaCheckY{1} + c^2_0 \zetaCheckY{0}\big)$ and $\xicheckx{c} \sim  N\big(0, c^2_1 \zetaCheckH{1} + c^2_0 \zetaCheckH{0}\big)$; then by the continuous mapping theorem, we finish the proof.
		
		Suppose that independently for each $s \in \mathcal{S}$, $x \in \{1,0\}$  and indepdently of $(A^{(n)}, S^{(n)}, X^{(n)})$, $\{(\checkYaxs{1}{x}, \checkYaxs{0}{x}): 1 \leq i \leq n\}$ are i.i.d. with marginal distribution equal to the distribution of $(\check Y_i(1), \check Y_i(0))\mid S_i = s, X_i = x$. Define $N_{.x}(s) = \sum_{i=1}^{n} \ind{S_i < s, X_i = x}$ and $F_{.x}(s) = \Prob{S_i < s, X_i = x}$. Furthermore, define
		\begin{equation*}
			\begin{aligned}
				\tilde{L}_{n,Y^c} &= c_1 \sum_{s \in \mathcal{S}}\left( \frac{1}{\sqrt{n}}\frac{n_{.1}(s)}{\nax{1}{1}( s)} \sum_{i=n \frac{N_{.1}(s)}{n}+1}^{n\big\{\frac{N_{.1}(s)}{n}+\frac{\naxs{1}{1}{s}}{n}\big\}}\checkYaxs{1}{1} -  \frac{1}{\sqrt{n}}\frac{n_{.1}(s)}{\naxs{0}{1}{s}} \sum_{i=n\big\{\frac{N_{.1}(s)}{n}+\frac{\naxs{1}{1}{s}}{n}\big\}+1}^{n\big\{\frac{N_{.1}(s)}{n}+\frac{n_{.1}(s)}{n}\big\}}\checkYaxs{0}{1}\right) \\
				&\quad + c_0 \sum_{s \in \mathcal{S}}\left( \frac{1}{\sqrt{n}}\frac{n_{.0}(s)}{\nax{1}{0}( s)} \sum_{i=n \frac{N_{.0}(s)}{n}+1}^{n\big\{\frac{N_{.0}(s)}{n}+\frac{\naxs{1}{0}{s}}{n}\big\}}\checkYaxs{1}{0} -  \frac{1}{\sqrt{n}}\frac{n_{.0}(s)}{\naxs{0}{0}{s}} \sum_{i=n\big\{\frac{N_{.0}(s)}{n}+\frac{\naxs{1}{0}{s}}{n}\big\}+1}^{n\big\{\frac{N_{.0}(s)}{n}+\frac{n_{.0}(s)}{n}\big\}}\checkYaxs{0}{0}\right) \\
				&=c_1 \sum_{s \in \mathcal{S}}\left( \frac{1}{\sqrt{n}} \sum_{i=n \frac{N_{.1}(s)}{n}+1}^{n\big\{\frac{N_{.1}(s)}{n}+\frac{\naxs{1}{1}{s}}{n}\big\}}\frac{1}{\pi}\checkYaxs{1}{1} -  \frac{1}{\sqrt{n}} \sum_{i=n\big\{\frac{N_{.1}(s)}{n}+\frac{\naxs{1}{1}{s}}{n}\big\}+1}^{n\big\{\frac{N_{.1}(s)}{n}+\frac{n_{.1}(s)}{n}\big\}}\frac{1}{1-\pi}\checkYaxs{0}{1}\right) \\
				&\quad + c_0 \sum_{s \in \mathcal{S}}\left( \frac{1}{\sqrt{n}} \sum_{i=n \frac{N_{.0}(s)}{n}+1}^{n\big\{\frac{N_{.0}(s)}{n}+\frac{\naxs{1}{0}{s}}{n}\big\}}\frac{1}{\pi}\checkYaxs{1}{0} -  \frac{1}{\sqrt{n}} \sum_{i=n\big\{\frac{N_{.0}(s)}{n}+\frac{\naxs{1}{0}{s}}{n}\big\}+1}^{n\big\{\frac{N_{.0}(s)}{n}+\frac{n_{.0}(s)}{n}\big\}}\frac{1}{1-\pi}\checkYaxs{0}{0}\right) + o_P(1).\\
			\end{aligned}
		\end{equation*}
		By construction, $\{\Lcheckx{c} \mid A^{(n)}, S^{(n)}, X^{(n)}\} \equaldist \{\tilde{L}_{n,Y^c}\mid A^{(n)}, S^{(n)}, X^{(n)}\}$, and hence $\Lcheckx{c}\equaldist \tilde{L}_{n,Y^c} $. Since $\Scheckx{c}$ is a function of $S^{(n)}$ and $X^{(n)}$, we have $(\Lcheckx{c}, \Scheckx{c}) \equaldist (\tilde{L}_{n,Y^c}, \Scheckx{c}) $. Next, we define 
		\begin{equation*}
			\begin{aligned}
				{L}^*_{n,Y^c} &=  c_1 \sum_{s \in \mathcal{S}}\left( \frac{1}{\sqrt{n}} \sum_{i=\lfloor n F_{.1}(s)\rfloor +1}^{\lfloor n\{F_{.1}(s)+\pi p_{.1}(s)\}\rfloor }\frac{1}{\pi}\checkYaxs{1}{1} -  \frac{1}{\sqrt{n}} \sum_{i=\lfloor n\{F_{.1}(s)+\pi p_{.1}(s)\}\rfloor+1}^{\lfloor n\{F_{.1}(s)+ p_{.1}(s)\}\rfloor}\frac{1}{1-\pi}\checkYaxs{0}{1}\right) \\
				&\quad + c_0 \sum_{s \in \mathcal{S}}\left( \frac{1}{\sqrt{n}} \sum_{i=\lfloor n F_{.0}(s)\rfloor +1}^{\lfloor n\{F_{.0}(s)+\pi p_{.0}(s)\}\rfloor }\frac{1}{\pi}\checkYaxs{1}{0} -  \frac{1}{\sqrt{n}} \sum_{i=\lfloor n\{F_{.0}(s)+\pi p_{.0}(s)\}\rfloor+1}^{\lfloor n\left(F_{.0}(s)+ p_{.0}(s)\right)\rfloor}\frac{1}{1-\pi}\checkYaxs{0}{0}\right).
			\end{aligned}
		\end{equation*}
	
		Note that ${L}^*_{n,Y^c}$ is a function of $\{(\checkYaxs{1}{x}, \checkYaxs{0}{x}): 1 \leq i \leq n\}$, which is independent of $(A^{(n)}, S^{(n)}, X^{(n)})$. Moreover, by noticing $ \Scheckx{c}$ is a function of $(A^{(n)}, S^{(n)}, X^{(n)})$, we can conclude that ${L}^*_{n,Y^c} \Perp \Scheckx{c}$.
		
		We next show
		\begin{equation*}
			\begin{aligned}
				(\Lcheckx{c}, \Scheckx{c}) \equaldist ({L}^*_{n,Y^c}, \Scheckx{c}) + o_P(1).
			\end{aligned}
		\end{equation*}
		It suffices to show $\Delta_{n,Y^c} =  \tilde{L}_{n,Y^c} - {L}^*_{n,Y^c} \inp 0$. For each $s \in \mathcal{S}$ and $x \in \{1,0\}$, define 
		\begin{equation}
			\begin{aligned}
				g_n^{x,s}(u) = \frac{1}{\sqrt{n}} \sum_{i=1}^{\lfloor nu \rfloor} \checkYaxs{1}{x}.
			\end{aligned}
		\end{equation}
		Since $\{(\checkYaxs{1}{x}, \checkYaxs{0}{x}): 1 \leq i \leq n\}$ are i.i.d. with mean $0$ and finite variance, and by further noticing 	
		\begin{equation}
			\begin{aligned}
				\left( \frac{\nxs{x}{s}}{n}, \frac{\naxs{1}{x}{s}}{n}\right) \inp \left( F_{.x}(s), \pi \pxs{x}{s}\right),
			\end{aligned}
		\end{equation}
		we can conclude that 
		\begin{equation*}
			\begin{aligned}
				g_n^{x,s}\left(\frac{\nxs{x}{s}}{n}+\frac{\naxs{1}{x}{s}}{n}\right) - g_n^{x,s}\left(\frac{\nxs{x}{s}}{n}\right)-\Big(g_n^{x,s}\big(F_{.x}(s)+\pi \pxs{x}{s}\big)-g_n^{x,s}\big(F_{.x}(s)\big)\Big) \stackrel{P}{\rightarrow} 0,
			\end{aligned}
		\end{equation*}
		where the convergence follows from elementary properties of Brownian motion and the continuous mapping theorem. 
		Similarly arguing for $\checkYaxs{0}{x}$, we can prove that $\Delta_{n,Y^c} =  \tilde{L}_{n,Y^c} - {L}^*_{n,Y^c} \inp 0$.
		
		We proceed to show that ${L}^*_{n,Y^c} \indist \etacheckx{c}$. Under our assumptions, $g_n^{x,s}(\mu)$ converges weakly to Brownian motion. By elementary properties of Brownian motion, we have that 
		\begin{equation*}
			\begin{aligned}
				\frac{1}{\sqrt{n}} \sum_{i=\lfloor n F_{.x}(s)\rfloor +1}^{\lfloor n\{F_{.x}(s)+\pi \pxs{x}{s}\}\rfloor }\frac{1}{\pi}\checkYaxs{1}{x} \indist N\left(0, \frac{\pxs{x}{s}\sigma^2_{\check Y(1)}(x,s)}{\pi}\right),
			\end{aligned}
		\end{equation*}
		where we write $\Var{\checkYaxs{1}{x}} = \sigma^2_{\check Y(1)}(x, s)$.
		Repeating the similar steps for $\checkYaxs{1}{x}$ and $\checkYaxs{0}{x}$ for each $s \in \mathcal{S}$, $x \in \{1, 0\}$ and using the independence of $\{(\checkYaxs{1}{x}, \checkYaxs{0}{x}): 1 \leq i \leq n\}$ across both $i$ and $s$ and $x$, we can conclude that ${L}^*_{n,Y^c} \indist \etacheckx{c}$.
		
			By the central limit theorem and the continuous mapping theorem, 	
		\begin{equation*}
			\begin{aligned}
				\Scheckx{c}\indist  N(0,c_1^2\zetaCheckH{1} + c_0^2\zetaCheckH{0}).
			\end{aligned}
		\end{equation*}	
		We are left with showing $(\Lcheckx{c}, \Scheckx{c}) \indist (\etacheckx{c}, \xicheckx{c})$,  where  $\etacheckx{c} \equaldist N\big(0,c_1^2\zetaCheckY{1} + c_0^2\zetaCheckY{0}\big)$,  $\xicheckx{c} \equaldist N\big(0,c_1^2\zetaCheckH{1} + c_0^2\zetaCheckH{0}\big)$ and $\etacheckx{c}, \xicheckx{c}$ are independent.
		
		It suffices to show  $({L}^*_{n,Y^c}, \Scheckx{c}) \indist (\etacheckx{c}, \xicheckx{c})$, which immediately follows from ${L}^*_{n,Y^c} \Perp \Scheckx{c}$.
	\end{proof}
\end{lemma}

\begin{lemma}
	The asymptotic covariance matrices in Theorem 11
 and Theorem 14
 satisfy $\Sigma \succeq \Sigma_{strat}$. Thus, the stratified-adjusted test is asymptotically more powerful than the modified test.
	\label{prop:power}
\end{lemma}

\begin{proof}
	First, 
	\begin{equation}
		\begin{aligned}
			&\quad{\Sigma}_{hc} - {\Sigma}_{strat} \\
			&= \pi(1-\pi)\sum_{s \in \mathcal{S}}p(s)\text{diag}\Big(\frac{1}{p^2_x}\frac{\pxs{x}{s}}{\ps{s}} \big[\frac{1}{\pi} \Econd{\maxZ{1}{x}}{S_i = s, X_i = x} \\
            &\hspace{1.8in}+ \frac{1}{1-\pi} \Econd{\maxZ{0}{x}}{S_i = s, X_i = x}\big]^2: x \in \mathcal{X}\Big), \\
		\end{aligned}
		\label{eq:sigma_strata}
	\end{equation}
and (\ref{eq:sigma_mod}) shows
	\begin{equation}
	\begin{aligned}
		&\quad {\Sigma}_{hc} - \Sigma \\
             &= 
			\sum_{s \in \mathcal{S}}p(s)\{\pi(1-\pi)-q(s)\} \Big(\frac{1}{\px{x}}\big[\frac{1}{\pi} \Econd{\raxZ{1}{x}}{S_i = s} + \frac{1}{1-\pi} \Econd{\raxZ{0}{x}}{S_i = s}\big]: x \in \mathcal{X}\Big)\\
			&\qquad \Big(\frac{1}{\px{x}}\big[\frac{1}{\pi} \Econd{\raxZ{1}{x}}{S_i = s} + \frac{1}{1-\pi} \Econd{\raxZ{0}{x}}{S_i = s}\big]: x \in \mathcal{X}\Big)^{\mathrm{T}} \\
        &= 
		\sum_{s \in \mathcal{S}}p(s)\{\pi(1-\pi)-q(s)\} \Big(\frac{1}{p_x}\frac{\pxs{x}{s}}{\ps{s}}\big[\frac{1}{\pi} \Econd{\maxZ{1}{x}}{S_i = s, X_i = x} \\
        &\hspace{2.1in}+ \frac{1}{1-\pi}\Econd{\maxZ{0}{x}}{S_i = s, X_i = x}\big]: x \in \mathcal{X}\Big)\\
		&\qquad \Big(\frac{1}{p_x}\frac{\pxs{x}{s}}{\ps{s}}\big[\frac{1}{\pi} \Econd{\maxZ{1}{x}}{S_i = s, X_i = x}  + \frac{1}{1-\pi} \Econd{\maxZ{0}{x}}{S_i = s, X_i = x}\big]: x \in \mathcal{X}\Big)^{\mathrm{T}}. \\
	\end{aligned}
\label{eq:sigma_mod2}
\end{equation}
Subtracting (\ref{eq:sigma_mod2}) from (\ref{eq:sigma_strata}) yields
	\begin{equation}
	\begin{aligned}
		&\quad{\Sigma} - {\Sigma}_{strat} \\
		&= \sum_{s \in \mathcal{S}}p(s)q(s) \Big(\frac{1}{p_x}\frac{\pxs{x}{s}}{\ps{s}}\big[\frac{1}{\pi} \Econd{\maxZ{1}{x}}{S_i = s, X_i = x} \\
        &\hspace{1.2in}+ \frac{1}{1-\pi}\Econd{\maxZ{0}{x}}{S_i = s, X_i = x}\big]: x \in \mathcal{X}\Big)\\
		&\qquad \Big(\frac{1}{p_x}\frac{\pxs{x}{s}}{\ps{s}}\big[\frac{1}{\pi} \Econd{\maxZ{1}{x}}{S_i = s, X_i = x} + \frac{1}{1-\pi} \Econd{\maxZ{0}{x}}{S_i = s, X_i = x}\big]: x \in \mathcal{X}\Big)^{\mathrm{T}} \\
		&\quad + \pi(1-\pi)\sum_{s \in \mathcal{S}}p(s)\Bigg\{\text{diag}\Big(\frac{1}{p^2_x} \frac{\pxs{x}{s}}{\ps{s}} \big[\frac{1}{\pi} \Econd{\maxZ{1}{x}}{S_i = s, X_i = x} \\
        &\hspace{2in}+ \frac{1}{1-\pi} \Econd{\maxZ{0}{x}}{S_i = s, X_i = x}\big]^2: x \in \mathcal{X}\Big) \\
		&\qquad -
		\Big(\frac{1}{\px{x}}\frac{\pxs{x}{s}}{\ps{s}}\big[\frac{1}{\pi} \Econd{\maxZ{1}{x}}{S_i = s, X_i = x} + \frac{1}{1-\pi}\Econd{\maxZ{0}{x}}{S_i = s, X_i = x}\big]: x \in \mathcal{X}\Big)\\
		&\quad \qquad \Big(\frac{1}{\px{x}}\frac{\pxs{x}{s}}{\ps{s}}\big[\frac{1}{\pi} \Econd{\maxZ{1}{x}}{S_i = s, X_i = x} + \frac{1}{1-\pi} \Econd{\maxZ{0}{x}}{S_i = s, X_i = x}\big]: x \in \mathcal{X}\Big)^{\mathrm{T}} \Bigg\} \\
	\end{aligned}
	\label{eq:cov_diff}
\end{equation}
is positive semi-definite because both two terms in (\ref{eq:cov_diff}) are positive semi-definite. By noticing that $\sum_{x \in \mathcal{X}} \{\pxs{x}{s} / p(s)\} = 1$, the positive semi-definiteness of the second term can be verified using Jensen's inequality by checking the definition of the positive semi-definite matrix.
\end{proof}

\section{Additional Simulation Results}
In this section, we present simulations in the scenarios where the categorical covariate $X$ with more than two levels is tested for interaction. We examine the type \RNum{1} error and power of tests for all three interaction tests, i.e.,  the usual test, modified test, and stratified-adjusted test. We consider three randomization methods: simple randomization (SR), stratified block randomization (SBR), and stratified biased-coin design (SBCD). The block size is 6 and the biased-coin probability is  0.75. The nominal level $\alpha$ is $5\%$, the  number of units $n$ is 800, and simulations are based on 10,000 runs.

Model 1: Linear model: 	
$$
\begin{cases}
	\Yt = \mu_1 +(\beta + \delta)^{\mathrm{T}}X_{i}  + \alpha_{1}W^*_{i} + \gamma^{\mathrm{T}} W^*_{i}X_{i} + \sigma_1\epsilon_{1,i},\\
	\Yc  = \mu_0 + \beta^{\mathrm{T}}X_{i} + \alpha_{1}W^*_{i} + \sigma_0 \epsilon_{0,i},\\
\end{cases}
$$
where $X$ is a discrete variable with three possible values $(0, 0)^{\mathrm{T}}$, $(1, 0)^{\mathrm{T}}$, or $(0, 1)^{\mathrm{T}}$ with equal probabilities; $W^* \sim N(0, 3^2)$ independent of $X$ is categorized to Bernoulli random variable  $W = \ind{W^* > 0}$. It is assumed $\mu_1 = 4$, $\mu_0 = 1$, $\beta =(3, 2)^{\mathrm{T}}$ , $\alpha_{1} = -2$, $\gamma = (4, 3)^{\mathrm{T}}$, $\sigma_1 = 1$, $\sigma_0 = 0.5$, $ \epsilon_{1,i} \sim N(0, 1)$  and $ \epsilon_{0,i} \sim N(0, 1)$. $\delta$ equals $(0, 0)^{\mathrm{T}}$ (null) or $(1, 2)^{\mathrm{T}}$ (alternative).

Model 2: Nonlinear model: 
$$
\begin{cases}
	\Yt = \mu_1 + \exp\big\{(\beta_{1} + \delta_{1})X^*_{i}\big\}   + \alpha_{1}W^*_{i}+ \gamma_{1}X^*_{i}W^*_{i}+  \sigma_1(X^*_{i})\epsilon_{1,i}, \\
	\Yc  = \mu_0 +\exp\big(\beta_{1} X^*_{i}\big) + \alpha_{1}W^*_{i} + \sigma_0(X^*_{i})\epsilon_{0,i}, \\
\end{cases}
$$
where $X^* \sim \text{Unif}(-1, 2)$ is categorized to $X$ that equals 0, 1, or 2 when $X^*$ falls in $\left[-1, 0\right)$, $\left[0, 1\right)$, or $\left[1, 2\right]$, respectively; $W^* \sim N(0, 2^2)$ independent of $X^*$ is categorized to $W = \ind{W^* > 0}$. It is assumed $\mu_1 = 5$, $\mu_0 = 4$, $\beta_{1} =0.5$, $\alpha_{1} = 2$, $\gamma_{1} = 6$, $\sigma_1(X^*_{i}) = \exp(0.5X^*_{i})$, $\sigma_0(X^*_{i}) = 0.5\exp(0.5X^*_{i})$, $ \epsilon_{1,i} \sim N(0, 1)$,  and $ \epsilon_{0,i} \sim N(0, 1)$. $\delta_1$ equals $0$ (null) or $0.4$ (alternative).

Model 3: Binary outcomes model: 
$$
\begin{cases}
	\Yt = I\big\{\mu_1 + (\beta_{1} + \delta_{1})(X_i-1/2)+ \alpha_{1}W_{i} + \gamma_{1}X_{i}W_{i} > U_{1,i}\big\}, \\
	\Yc = I\big\{\mu_0 + \beta_{1} (X_i-1/2) +\alpha_{1}W_{i}  > U_{0,i}\big\}, \\
\end{cases}
$$
where $X$ is a discrete variable with three possible values $0$, $0.5$, or $1$ with equal probabilities; $W$ equals $1$ or $-1$ with equal probability. It is assumed $\mu_1 = \mu_0 = 4$, $\beta_{1} =1$ , $\alpha_{1} = -3$, $\gamma_{1} = 6$, $U_{1,i} \sim \text{Unif}(0, 10)$, and $U_{0,i} \sim \text{Unif}(0, 10)$. $\delta_1$ equals $0$ (null) or $1.5$ (alternative).
~\\

Tables \ref{tab:simulation_equal_multi} and \ref{tab:simulation_unequal_multi}
present the empirical type \RNum{1} error and power of tests under equal and unequal allocation when $X$ is categorical with more than two levels.

\begin{table}[]
\caption{Rejection probabilities (in percentage points) for three interaction tests (on categorical $X$), stratification mechanisms, and randomization methods under equal allocation, $\pi = 1/2$, and sample size $n = 800$}{
\tabcolsep=4.25pt
\renewcommand\arraystretch{0.5}
\small
\begin{tabular}{cccccccc}
\hline
              &           & \multicolumn{2}{c}{Stratification covariates} &                      & \multicolumn{3}{c}{Randomization methods}                        \\ \cline{3-4} \cline{6-8} \multicolumn{1}{l}{} &\multicolumn{1}{c}{Model} & $X$ stratified         & $S(\cdot)$             &                      & SR                & SBR                  & SBCD                  \\ \hline
$H_0$         & 1         & Yes                  & $X$                    &                      & 4.7 / 4.8 / 4.8 & 5.3 / 5.3 / 5.3 & 5.4 / 5.4 / 5.4 \\ 
              &           &                      & $X \times W$         &                      & 5.4 / 5.5 / 5.7 & 2.5 / 5.6 / 5.4 & 2.4 / 5.5 / 5.3 \\ 
              &           & No                   & --                     &                      & 5.4 / 5.4 / 5.4 & 5.6 / 5.6 / 5.6 & 5.3 / 5.4 / 5.4 \\
              &           &                      & $W$                  &                      & 5.1 / 5.1 / 5.5 & 4.5 / 5.5 / 5.5 & 4.6 / 5.6 / 5.7 \\ 
              &           &                      &                        & \multicolumn{1}{l}{} &                    & \multicolumn{1}{l}{} & \multicolumn{1}{l}{} \\
              & 2         & Yes                  & $X$                    &                      & 5.3 / 5.4 / 5.4 & 5.4 / 5.4 / 5.4 & 5.3 / 5.3 / 5.3 \\ 
              &           &                      & $X \times W$         &                      & 5.0 / 5.1 / 5.3 & 0.7 / 5.3 / 5.4 & 1.0 / 5.5 / 5.3 \\ 
              &           & No                   & --                     &                      & 5.6 / 5.6 / 5.6 & 5.4 / 5.5 / 5.5 & 5.3 / 5.4 / 5.4 \\ 
              &           &                      & $W$                  &                      & 5.4 / 5.4 / 5.2 & 3.8 / 5.1 / 5.4 & 4.1 / 5.6 / 5.6 \\ 
              &           &                      &                        & \multicolumn{1}{l}{} &                    & \multicolumn{1}{l}{} & \multicolumn{1}{l}{} \\
              & 3         & Yes                  & $X$                    &                      & 5.3 / 5.3 / 5.3 & 5.4 / 5.4 / 5.4 & 5.0 / 5.0 / 5.0 \\ 
              &           &                      & $X \times W$         &                      & 5.2 / 5.2 / 5.1 & 2.9 / 5.0 / 4.9 & 3.0 / 5.0 / 5.0 \\ 
              &           & No                   & --                     &                      & 5.0 / 5.1 / 5.1 & 4.9 / 4.9 / 4.9 & 4.9 / 4.9 / 4.9  \\
              &           &                      & $W$                  &                      & 5.2 / 5.2 / 5.1 & 4.7 / 5.4 / 5.3 & 4.5 / 5.2 / 5.3 \\ 
              &           &                      &                        & \multicolumn{1}{l}{} &                    & \multicolumn{1}{l}{} & \multicolumn{1}{l}{} \\
$H_1$         & 1         & Yes                  & $X$                    &                      & 47.4 / 47.7 / 47.7 & 47.6 / 47.6 / 47.6 & 47.2 / 47.2 / 47.2 \\ 
              &           &                      & $X \times W$         &                      & 47.9 / 48.0 / 70.1 & 44.8 / 70.2 / 70.4 & 44.3 / 70.2 / 70.1 \\ 
              &           & No                   & --                     &                      & 47.0 / 47.2 / 47.2 & 48.1 / 48.2 / 48.2 & 46.6 / 46.9 / 46.9 \\ 
              &           &                      & $W$                  &                      & 46.8 / 47.1 / 70.6 & 46.4 / 50.0 / 70.8 & 46.6 / 50.0 / 70.2  \\    
              &           &                      &                        & \multicolumn{1}{l}{} &                    & \multicolumn{1}{l}{} & \multicolumn{1}{l}{} \\
              & 2         & Yes                  & $X$                    &                      &  41.6 / 41.8 / 41.8 & 42.4 / 42.4 / 42.4 & 42.5 / 42.4 / 42.4 \\ 
              &           &                      & $X \times W$         &                      & 42.2 / 42.3 / 60.8 & 35.3 / 60.2 / 59.9 & 35.1 / 60.5 / 60.3 \\ 
              &           & No                   & --                     &                      & 41.4 / 41.5 / 41.5 & 41.9 / 42.1 / 42.1 & 41.3 / 41.4 / 41.4 \\ 
              &           &                      & $W$                  &                      & 41.3 / 41.4 / 60.1 & 40.2 / 47.7 / 60.1 & 40.9 / 48.5 / 60.4 \\ 
              &           &                      &                        & \multicolumn{1}{l}{} &                    & \multicolumn{1}{l}{} & \multicolumn{1}{l}{} \\
              & 3         & Yes                  & $X$                    &                      & 37.5 / 37.6 / 37.6 & 36.9 / 36.9 / 36.9 & 37.2 / 37.2 / 37.2 \\ 
              &           &                      & $X \times W$         &                      & 37.5 / 37.6 / 44.6 & 36.1 / 44.2 / 44.2 & 35.4 / 43.7 / 43.7  \\ 
              &           & No                   & --                     &                      & 37.4 / 37.5 / 37.5 & 37.3 / 37.4 / 37.4 & 37.0 / 37.1 / 37.1 \\
              &           &                      & $W$                  &                      & 37.5 / 37.6 / 44.9 & 36.6 / 38.8 / 43.9 & 36.7 / 39.4 / 44.2 \\ \hline
\end{tabular}}
\begin{tablenotes}
\item Notes: $H_0$ and $H_1$ denote the null and alternative conditions, respectively. $X$ stratified indicates whether $X$ is used for stratification. $S(\cdot)$ are the covariates used for stratification. $X \times W$ indicates that both $X$ and $W$ are used to form the strata. -- indicates that no covariate is used to form the strata. SR, simple randomization; SBR, stratified block randomization; SBCD, stratified biased-coin design. The three numbers in each column correspond to the rejection probabilities for the usual test, modified test, and stratified-adjusted test.
\end{tablenotes}
\label{tab:simulation_equal_multi}
\end{table}

\begin{table}[]
\caption{Rejection probabilities (in percentage points) for three interaction tests (on categorical $X$), stratification mechanisms, and randomization methods under unequal allocation, $\pi = 2/3$, and sample size $n = 800$}{
\tabcolsep=4.25pt
 \renewcommand\arraystretch{0.5}
\small
\begin{tabular}{cccccccc}
\hline
              &           & \multicolumn{2}{c}{Stratification covariates} &                      & \multicolumn{3}{c}{Randomization methods}                        \\ \cline{3-4} \cline{6-8}
\multicolumn{1}{l}{} &\multicolumn{1}{c}{Model} & $X$ stratified         & $S(\cdot)$             &                      & SR                & SBR                  & SBCD                  \\ \hline
$H_0$         & 1         & Yes                  & $X$                    &                      & 5.0 / 5.1 / 5.1 & 5.2 / 5.2 / 5.2 & 5.5 / 5.4 / 5.4 \\ 
              &           &                      & $X \times W$         &                      & 5.3 / 5.3 / 5.8 & 1.9 / 5.7 / 5.7 & 1.9 / 5.7 / 5.3 \\ 
              &           & No                   & --                     &                      & 5.5 / 5.6 / 5.6 & 5.2 / 5.2 / 5.2 & 5.8 / 5.8 / 5.8 \\ 
              &           &                      & $W$                  &                      & 5.4 / 5.3 / 5.8 & 5.3 / 5.8 / 5.7 & 5.4 / 5.8 / 5.8 \\ 
              &           &                      &                        & \multicolumn{1}{l}{} &                    & \multicolumn{1}{l}{} & \multicolumn{1}{l}{} \\
              & 2         & Yes                  & $X$                    &                      & 5.5 / 5.5 / 5.5 & 5.3 / 5.3 / 5.3 & 5.1 / 5.1 / 5.1 \\ 
              &           &                      & $X \times W$         &                      & 5.2 / 5.3 / 5.7 & 1.2 / 5.3 / 5.2 & 1.3 / 5.7 / 5.2 \\
              &           & No                   & --                     &                      & 5.4 / 5.5 / 5.5 & 5.2 / 5.3 / 5.3 & 5.2 / 5.1 / 5.1 \\
              &           &                      & $W$                  &                      & 5.1 / 5.1 / 5.2 & 4.5 / 5.4 / 5.0 & 4.8 / 5.6 / 5.6 \\ 
              &           &                      &                        & \multicolumn{1}{l}{} &                    & \multicolumn{1}{l}{} & \multicolumn{1}{l}{} \\
              & 3         & Yes                  & $X$                    &                      & 5.3 / 5.3 / 5.3 & 5.3 / 5.3 / 5.3 & 5.6 / 5.6 / 5.6 \\
              &           &                      & $X \times W$         &                      & 5.3 / 5.4 / 5.5 & 2.6 / 4.8 / 4.7 & 2.9 / 5.0 / 4.9 \\
              &           & No                   & --                     &                      & 4.8 / 4.9 / 4.9 & 5.0 / 5.1 / 5.1 & 5.3 / 5.4 / 5.4 \\ 
              &           &                      & $W$                  &                      & 5.3 / 5.3 / 5.5 & 4.9 / 5.2 / 5.3 & 4.7 / 5.0 / 5.4 \\ 
              &           &                      &                        & \multicolumn{1}{l}{} &                    & \multicolumn{1}{l}{} & \multicolumn{1}{l}{} \\
$H_1$         & 1         & Yes                  & $X$                    &                      & 39.8 / 40.1 / 40.1 & 40.3 / 40.2 / 40.2 & 40.6 / 40.1 / 40.1 \\ 
              &           &                      & $X \times W$         &                      & 41.0 / 41.0 / 65.7 & 35.0 / 65.0 / 65.3 & 35.5 / 64.6 / 64.7 \\ 
              &           & No                   & --                     &                      & 40.2 / 40.4 / 40.4 & 40.3 / 40.5 / 40.5 & 40.9 / 41.0 / 41.0 \\ 
              &           &                      & $W$                  &                      & 39.8 / 40.0 / 65.6 & 40.1 / 41.7 / 65.8 & 39.3 / 40.6 / 64.8 \\  
              &           &                      &                        & \multicolumn{1}{l}{} &                    & \multicolumn{1}{l}{} & \multicolumn{1}{l}{} \\
              & 2         & Yes                  & $X$                    &                      &  48.7 / 48.7 / 48.7 & 49.8 / 49.8 / 49.8 & 48.8 / 49.0 / 49.0 \\ 
              &           &                      & $X \times W$         &                      & 49.5 / 49.7 / 66.7 & 45.4 / 66.0 / 66.1 & 45.0 / 66.4 / 66.6 \\ 
              &           & No                   & --                     &                      & 49.1 / 49.2 / 49.2 & 49.5 / 49.7 / 49.7 & 48.9 / 48.9 / 48.9 \\ 
              &           &                      & $W$                  &                      & 49.2 / 49.2 / 65.9 & 48.5 / 52.8 / 65.8 & 48.9 / 53.1 / 66.6 \\
              &           &                      &                        & \multicolumn{1}{l}{} &                    & \multicolumn{1}{l}{} & \multicolumn{1}{l}{} \\
              & 3         & Yes                  & $X$                    &                      & 33.2 / 33.5 / 33.5 & 33.6 / 33.6 / 33.6 & 33.9 / 33.6 / 33.6 \\
              &           &                      & $X \times W$         &                      & 33.7 / 33.9 / 41.5 & 32.0 / 41.4 / 41.4 & 32.1 / 41.3 / 41.0 \\ 
              &           & No                   & --                     &                      & 33.8 / 33.8 / 33.8 & 32.9 / 33.1 / 33.1 & 33.4 / 33.5 / 33.5 \\ 
              &           &                      & $W$                  &                      & 33.5 / 33.5 / 41.7 & 33.1 / 34.2 / 40.9 & 33.5 / 34.6 / 41.2 \\ \hline
\end{tabular}}
\begin{tablenotes}
\item Notes: $H_0$ and $H_1$ denote the null and alternative conditions, respectively. $X$ stratified indicates whether $X$ is used for stratification. $S(\cdot)$ are the covariates used for stratification. $X \times W$ indicates that both $X$ and $W$ are used to form the strata. -- indicates that no covariate is used to form the strata. SR, simple randomization; SBR, stratified block randomization; SBCD, stratified biased-coin design. The three numbers in each column correspond to the rejection probabilities for the usual test, modified test, and stratified-adjusted test.
\end{tablenotes}
\label{tab:simulation_unequal_multi}
\end{table}

\section{Clinical Trial Example}\label{section:real_data}
We compare the performances of all three tests on the synthetic data of the Nefazodone cognitive behavioral-analysis system of psychotherapy (CBASP) trial, which was conducted to compare the efficacy of three treatments for chronic depression \citep{keller2000}.  We focus on two of the treatments: Nefazodone and CBASP. The synthetic data are generated following \cite{Liu2023}. The total number of generated units is 600, and the outcome of interest is the final score of the 24-item Hamilton rating scale for depression. Details of the baseline covariates can be found in Table \ref{tab:syn_details}.  

We stratify the units based on GENDER and $\text{HAMD17}^d$, where $\text{HAMD17}^d$ is the categorized HAMD17 score, determined by the relative values of 18 and 21. Moreover, we categorize AGE into $\text{AGE}^d$ based on the relative values of 40 and 49. Subsequently, we implement simple randomization, stratified block randomization, and stratified biased-coin design to obtain the treatment assignments under equal allocation, $\pi = {1}/{2}$. We investigate the interaction of two covariates: the stratification covariate $\text{HAMD17}^d$ and additional covariate $\text{AGE}^d$. The results are summarized in Table \ref{tab:syn}.

Table \ref{tab:syn} shows that for covariate $\text{HAMD17}^d$, all three interaction tests yield small \textit{p}-values, indicating the presence of an interaction effect for $\text{HAMD17}^d$. For covariate $\text{AGE}^d$, the interaction effect is marginally significant. Notably, the stratified-adjusted test produces smaller \textit{p}-values than those obtained using the usual and modified test.

\begin{table}[ht]
	\caption{Description of baseline covariates}
	\centering
  \renewcommand\arraystretch{0.8}
	\begin{tabular}{ll}
		\hline
		Covariate  &  Description \\ 
		\hline
		AGE	& Age of patients (in years)\\ 
		GENDER &  1 female and 0 male\\ 
		HAMA\_SOMATI & HAMA somatic anxiety score\\
		HAMD17 & Total HAMD-17 score\\
		HAMD24 &  Total HAMD-24 score\\
		HAMD COGIN & HAMD cognitive disturbance score\\
		Mstatus2 & Marriage status: 1 if married or living with someone and 0 otherwise\\
		NDE & Number of depressive episodes\\
		TreatPD & Treated for past depression: 1 yes and 0 no\\
		\hline
	\end{tabular}
	\begin{tablenotes}
		\item[1] Notes: HAMD, Hamilton Depression Scale; HAMA, Hamilton Anxiety Scale.
	\end{tablenotes}
	\label{tab:syn_details}
\end{table}

\begin{table}[H]
	\caption{\textit{p}-values (in percentage points) for three interaction tests for synthetic Nefazodone CBASP trial data, $\pi = 1/2$, and sample size $n = 600$}
	\centering
\begin{tabular}{cccc}
\hline
Covariate         & SR                    & SBR                           & SBCD                   \\ \hline
$\text{HAMD17}^d$ &  4.95 / 4.89 / 3.99 & 4.26 / 3.74 / 3.74 & 3.88 / 3.37 / 3.38 \\
$\text{AGE}^d$    & 12.67 / 12.59 / 7.70 & 12.61 / 11.90 / 7.80 & 12.45 / 11.70 / 7.42 \\  \hline
\end{tabular}
	\begin{tablenotes}
		\item[1] Notes: \textit{p}-values (in percentage points) for interaction tests on two covariates ($\text{HAMD17}^d$, $\text{AGE}^d$) under three randomization methods: simple randomization (SR), stratified block randomization (SBR), and stratified biased-coin design (SBCD). Covariates GENDER and $\text{HAMD17}^d$ are used for stratification. The three numbers in each column correspond to the rejection probabilities for the usual test, modified test, and stratified-adjusted test.
	\end{tablenotes}
	\label{tab:syn}
\end{table}

\end{singlespace}

\end{document}